\documentclass[11pt]{article}
\usepackage[a4paper, margin=1in]{geometry}
\usepackage{authblk,amssymb,graphicx}
\usepackage[authoryear,square]{natbib}
\usepackage{amsmath,amsthm,tikz,stmaryrd,float}
\usepackage{tgtermes}
\usepackage{environ}
\usepackage[tworuled,vlined]{algorithm2e}
\usepackage{listings}
\usepackage{bookmark}
\bookmarksetup{depth=2}

\usepackage{url}

\renewcommand{\cite}[1]{\citep{#1}}

\definecolor{amber}{rgb}{1.0, 0.75, 0.0}
\definecolor{fire}{rgb}{0.99215686, 0.64705882, 0.05882353}
\definecolor{amaranth}{rgb}{0.9, 0.17, 0.31}
\definecolor{azure}{rgb}{0.0, 0.5, 1.0}
\hypersetup{colorlinks,linkcolor={fire},citecolor={azure},urlcolor={amaranth}}

\definecolor{codegreen}{rgb}{0,0.6,0}
\definecolor{codegray}{rgb}{0.5,0.5,0.5}
\definecolor{codepurple}{rgb}{0.58,0,0.82}
\definecolor{backcolour}{rgb}{0.94,0.94,0.96}
\lstdefinestyle{mystyle}{
    backgroundcolor=\color{backcolour},   
    stringstyle=\color{codepurple},
    commentstyle=\color{codegreen},
    numberstyle=\tiny\color{codegray},
    keywordstyle=\bfseries\color{red!40!black},
    basicstyle=\ttfamily\footnotesize,
    numbers=left,      
}
\usepackage[margin=1in]{geometry}
\usepackage[utf8]{inputenc} 
\usepackage[T1]{fontenc}    
\usepackage{hyperref}       
\usepackage{url}            
\usepackage{booktabs}       
\usepackage{amsfonts}       
\usepackage{nicefrac}       
\usepackage{microtype}      
\usepackage{csquotes}
\usepackage{amsmath,amssymb,amsthm,latexsym,color,mathtools}
\usepackage{algorithm2e}
\usepackage{array,booktabs}
\usepackage{subcaption}
\usepackage{tikz}
\usepackage{times}
\usepackage{wrapfig}
\newcommand{\acksection}{\section*{Acknowledgments and Disclosure of Funding}}
\NewEnviron{ack}{%
  \acksection
  \BODY
}

\definecolor{GreyBlue1}{rgb}{0.62, 0.73, 0.85} 
\definecolor{GreyBlue2}{rgb}{0.20, 0.35, 0.60}

\newcommand{\bm}{{\mathbf{m}}}
\newcommand{\bM}{{\mathbf{M}}}
\newcommand{\bx}{{\mathbf{x}}}

\newcommand{\bZ}{{\mathbf{Z}}}

\newcommand{\bA}{\mathbf{A}}

\newcommand{\bX}{\mathbf{X}}

\newcommand{\approxP}{\mathrel{\stackrel{{\rm P}}{\mathrel{\scalebox{1.8}[1]{$\simeq$}}}}}

\let\temp\phi
\let\phi\varphi
\let\varphi\temp

\let\temp\epsilon
\let\epsilon\varepsilon
\let\varepsilon\temp

\newcommand{\norm}[1]{\left\lVert#1\right\rVert}

\newcommand{\R}{{\mathbb R}}

\newcommand{\E}{{\mathbb E}}

\newcommand{\Pcl}{\mathcal{P}}

\newcommand{\Scl}{\mathcal{S}}

\DeclareMathOperator{\diag}{diag}

\renewcommand{\epsilon}{\ensuremath{\varepsilon}}

\usepackage{enumitem}
\usepackage{tikz}

\newtheoremstyle{plain}
  {3pt}   
  {0pt}   
  {}  
  {0pt}       
  {\bfseries} 
  {.}         
  {3pt} 
  {}          

\newtheoremstyle{claim}
 {3pt}
 {3pt}
 {}
 {}
 {\bfseries}
 {.}
 {3pt}
 {\thmname{#1} \thmnumber{#2}}
 
 \theoremstyle{plain}
\newtheorem{theorem}{Theorem}[section]

\newtheorem{lemma}[theorem]{Lemma}

\theoremstyle{definition}
\newtheorem{definition}[theorem]{Definition}

\usepackage{todonotes}

\title{Multi-layer State Evolution Under Random Convolutional Design}

\author[1]{Mara Daniels\footnote{The first two authors contributed equally to this work.}}
\author[2]{C\'edric Gerbelot$^*$}
\author[2]{Florent Krzakala}
\author[3]{Lenka Zdeborov\'a}

\affil[1]{College of Computer Science and Department of Mathematics, Northeastern University, 02120 Boston, USA}
\affil[2]{Ecole Polytechnique Fédérale de Lausanne (EPFL), Information, Learning and Physics (IdePHIcs) Laboratory, CH-1015 Lausanne, Switzerland}
\affil[3]{Ecole Polytechnique Fédérale de Lausanne (EPFL), Statistical Physics of Optimization (SPOC) Laboratory, CH-1015 Lausanne, Switzerland}

\date{}

\begin{document}

\maketitle

\begin{abstract}
 Signal recovery under generative neural network priors has emerged as a promising direction in statistical inference and computational imaging. Theoretical analysis of reconstruction algorithms under generative priors is, however, challenging. For generative priors with fully connected layers and Gaussian i.i.d. weights, this was achieved by the multi-layer approximate message (ML-AMP) algorithm via a rigorous state evolution. However, practical generative priors are typically convolutional, allowing for computational benefits and inductive biases, and so the Gaussian i.i.d. weight assumption is very limiting. In this paper, we overcome this limitation and establish the state evolution of ML-AMP for random convolutional layers. We prove in particular that random convolutional layers belong to the same universality class as Gaussian matrices. Our proof technique is of an independent interest as it establishes a mapping between convolutional matrices and spatially coupled sensing matrices used in coding theory.   
\end{abstract}

\section{Introduction}
In a typical signal recovery problem, one seeks to recover a data signal $x_0$ given access to measurements $y_0 = G_\theta(x_0)$, where the parameters $\theta$ of the signal model are known. In many problems, it is natural to view the measurement generation process as a composition of simple forward operators, or `layers.' In this work, we are concerned with multi-layer signal models of the form 
\begin{align}
	\label{eq:gen_prior}
	G_\theta(h) = \phi^{(1)}(W^{(1)} \phi^{(2)} ( W^{(2)} \hdots \phi^{(L)} ( W^{(L)} h) ) )).
\end{align}
where $W^{(l)} \in \mathbb{R}^{n_{l-1} \times n_l}$ are linear sensing matrices and where $\phi^{(l)}(z)$ are separable, possibly non-linear channel functions. In the $L=1$ case, this signal model naturally generalizes problems such as phase retrieval $\phi(z) = |z|$ or compressive sensing $\phi(z) = z$, and for multi-layer models $L > 1$, $G_\theta(h)$ may be viewed as a deep neural network. 

Recently, convolutional Generative Neural Networks (GNNs) have shown promise as generalizations of sparsity priors for a variety of signal processing applications \cite{bora2017compressed}. Motivated by this success, we take interest in a variant of the recovery problem \eqref{eq:gen_prior} in which some of the sensing matrices $W^{(l)}$ may be \textit{multi-channel convolutional} (MCC) matrices, having a certain block-sparse circulant structure which captures the convolutional layers used by many modern generative neural network architectures \cite{karras2018progressive, karras2019style}.

In this work, we develop an asymptotic analysis of the performance of an \textit{Approximate Message Passing} (AMP) algorithm \cite{donoho2009message} for recovery from multichannel convolutional signal models. This family of algorithms originates in statistical physics \cite{mezard2009information, zdeborova2016statistical} and allows to compute the marginals of an elaborate posterior distribution defined by an inference problem involving dense random matrices. A number of AMP iterations have been proposed for various inference problems, such as compressed sensing \cite{donoho2009message}, low-rank matrix recovery \cite{rangan2012iterative} or generalized linear modeling \cite{rangan2011generalized}. More recently, composite AMP iterations (ML-AMP) have been proposed to study multilayer inference problems \cite{manoel2017multi,aubin2019spiked}. Here we consider the ML-AMP proposed in \cite{manoel2017multi} to compute marginals of a multilayer generalized linear model, however the usual dense Gaussian matrices will be replaced by random convolutional ones. A major benefit of AMP lies in the fact that the asymptotic distribution of their iterates can be exactly determined by a low-dimensional recursion: the state evolution equations. This enables to obtain precise theoretical results for the reconstruction performance of the proposed algorithm.
Another benefit of such iterations is their low computational complexity, as they only involve matrix-multiplication and, in the separable case, pointwise non-linearities.

\begin{figure}[t]
    \centering
    \includegraphics[width=0.48\linewidth]{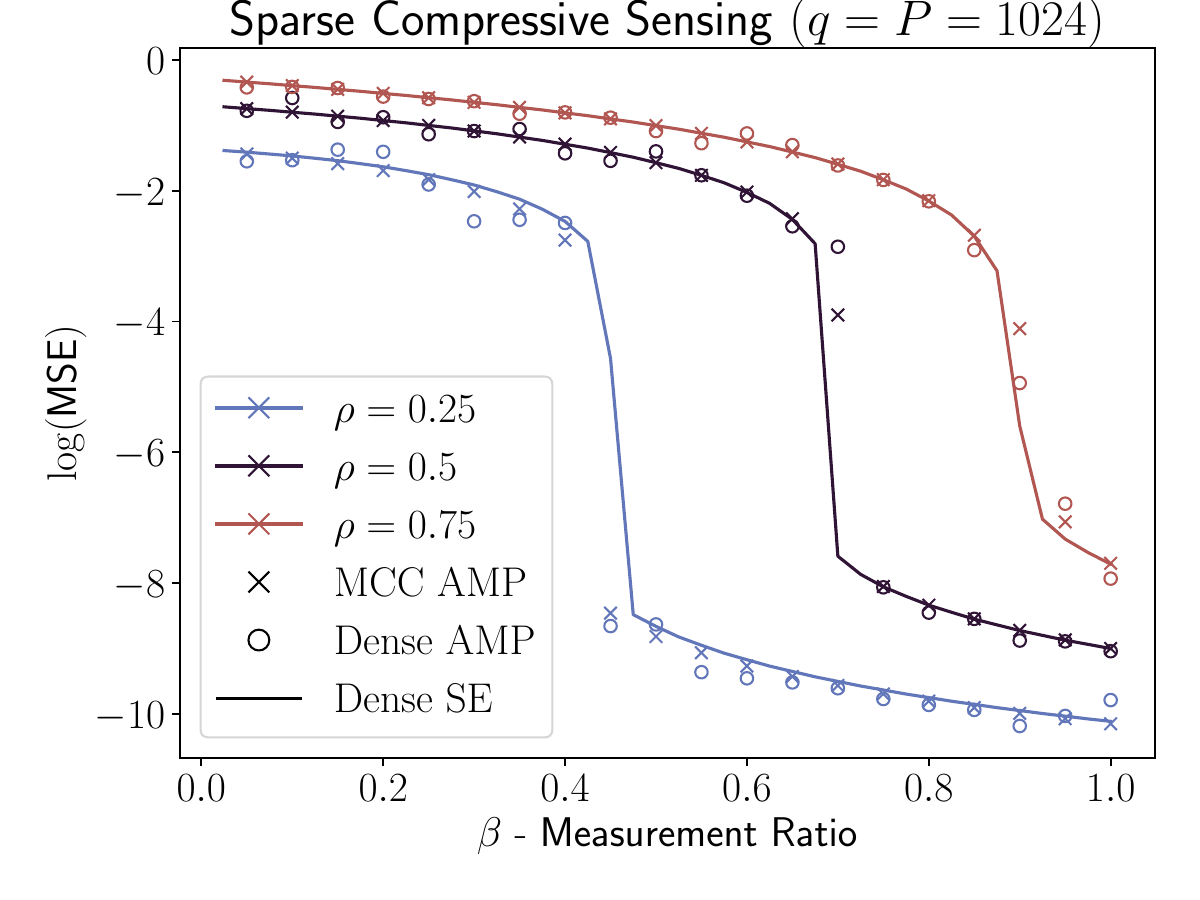}
    \includegraphics[width=0.51\linewidth]{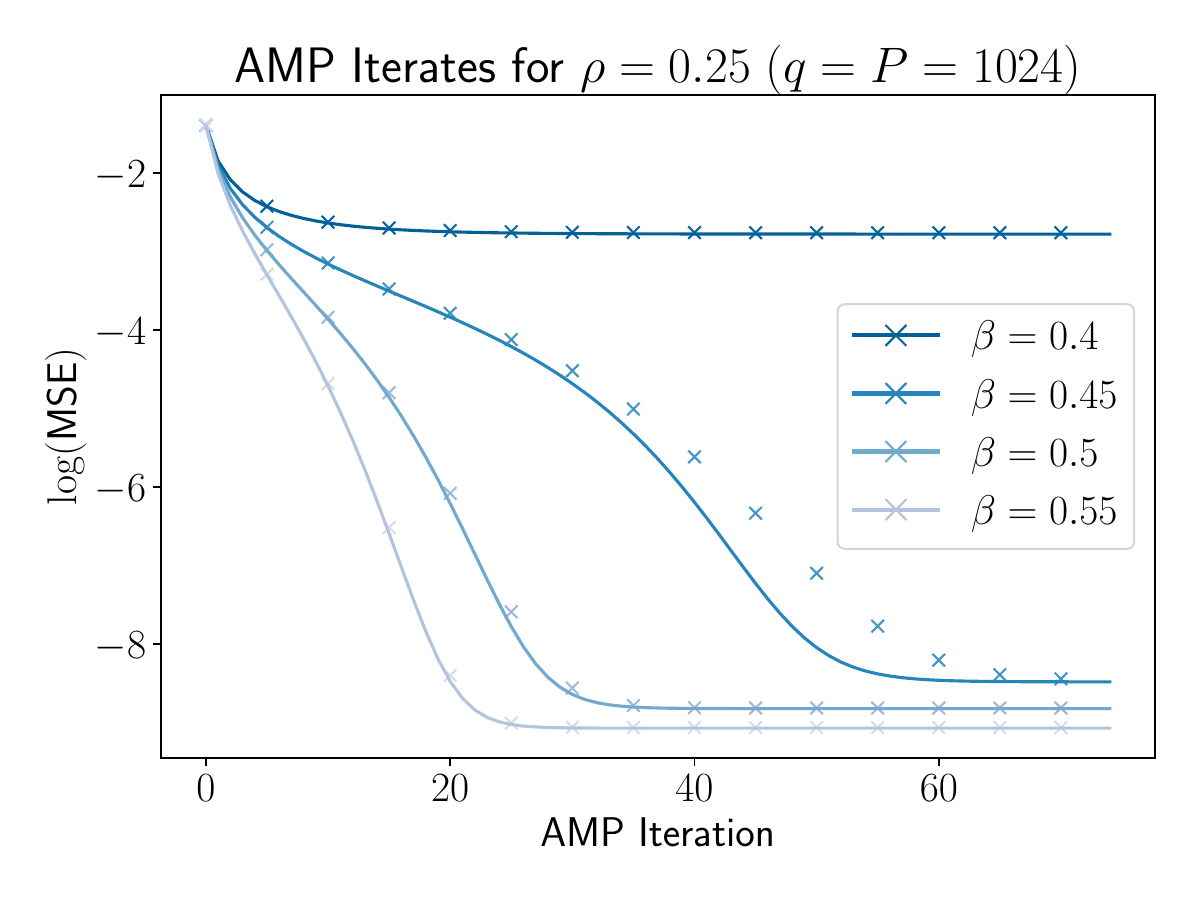}
    \caption{
    Agreement between the performance of the AMP algorithm run with random multichannel convolutional matrices and its state evolution as proven in this paper. 
    (\textbf{left}) Compressive sensing $y_0 = W x_0 + \zeta$ for noise $\zeta_i \sim \mathcal{N}(0, 10^{-4})$ and signal prior $x_0 \sim \rho \mathcal{N}(0, 1) + (1-\rho) \delta(x)$, where $W \in \mathbb{R}^{Dq \times Pq}$ has varying aspect ratio $\beta = D / P$. Crosses correspond to AMP evaluations for $W \sim \text{MCC}(D, P, q, k)$ according to Definition \ref{dfn:mcc}, averaged over 10 independent trials. Dots correspond to AMP evaluations for $W \in \mathbb{R}^{D \times P}$ with i.i.d. Gaussian entries $W_{ij} \sim \mathcal{N}(0, 1/P)$. Lines show the state evolution predictions when $W_{ij} \sim \mathcal{N}(0,1/Pq)$. The system size is $P = 1024$, $q=1024$, $k=3$, where $\beta$ and $D = \beta P$ vary. While our theorem treats the limit $P, D \to \infty$, $q, k = O(1)$, we observe strong empirical agreement even when $q \sim P$. In Appendix \ref{sec:q10-sparse-cs} we give the same figure for $q=10 \ll P$.
    (\textbf{right}) AMP iterates at $\rho = 0.25$ and $\beta$ near the recovery transition. Rather than showing these models have equivalent fixed points, we show a stronger result: the state evolution equations are equivalent \textit{at each iteration}.} \label{fig:sparse-cs}
\end{figure}

Previous works on AMP suggest that the state evolution is not readily applicable to our setting because its derivation requires strong independence assumptions on the coordinates of the $\{W^{(l)}\}$ which are violated by structured multi-channel convolution matrices. Despite this, we use AMP for our setting and rigorously prove its state evolution. Our main contributions are:
\begin{enumerate}
	\item We rigorously prove state evolution equations for models of the form \eqref{eq:gen_prior}, where weights are allowed to be either i.i.d. Gaussian or random structured MCC matrices, as in Definition \ref{dfn:mcc}. 
	\item For separable channel functions $\phi^{(l)}$ and separable signal priors, we show that the original ML-AMP of \cite{manoel2017multi} used with dense Gaussian matrices or random convolutional ones admits the same state evolution equations, up to a rescaling. Multi-layer MCC signal models can therefore simulate dense signal models while making use of fast structured matrix operations for convolutions. 
	
	\item The core of our proof shows how an AMP iteration involving random convolutional matrices may be reduced to another one with dense Gaussian matrices. We first show that random convolutional matrices are equivalent, through permutation matrices, to dense Gaussian ones with a (sparse) block-circulant structure. We then show how the block-circulant structure can be embedded in a new, matrix-valued, multilayer AMP with dense Gaussian matrices, the state evolution equations of which are proven using the results of \cite{gerbelot2021graph}, with techniques involving spatially coupled matrices \cite{krzakala2012statistical,javanmard2013state}.
\item  We validate our theory numerically and observe close agreement between convolutional AMP iterations and its state evolution predictions, as shown in Figure \ref{fig:sparse-cs} and in Section \ref{sec:numerics}. Our code can be used as a general purpose library to build compositional models and evaluate AMP and its state evolution. We make this code available at \url{https://github.com/mdnls/conv-ml-amp}.
\end{enumerate}

\section{Related Work}
AMP-type algorithms arose independently in the contexts of signal recovery and spin-glass theory. In the former case, \cite{donoho2009message} derives AMP for Gaussian compressive sensing. This approach was later generalized by \cite{rangan2011generalized} to recovery problems with componentwise \textit{channel functions} that may be stochastic and/or nonlinear, and generalized further by \cite{manoel2017multi} to multi-layer or compositional models. Due to the versatility of this approach, a wide variety of general purpose frameworks for designing AMP variants have since been popularized \cite{fletcher2018inference,baker2020tramp,gerbelot2021graph}. Proof techniques to show the concentration of AMP iterates to the state evolution prediction developed alongside new variants of the algorithm. In the context of spin-glass theory, Bolthausen's seminal work \cite{Bolthausen_2009} introduces a Gaussian conditioning technique used widely to prove AMP concentration. Following this approach, \cite{bayati2011dynamics,javanmard2013state,berthier2020state} treat signal models with dense couplings and generalized channel functions. More recently, a proof framework adaptable to composite inference problems was proposed in \cite{gerbelot2021graph}, which we use in our proof.

There has also been significant interest in relaxing the strong independence assumptions required by AMP algorithms on sensing matrix coordinates. In one direction, \textit{Vector AMP} (VAMP) algorithms target signal models whose sensing matrices are drawn from \textit{right orthogonally invariant} distributions. The development of VAMP algorithms followed a similar trajectory to that of vanilla AMP \cite{schniter2016vector,fletcher2018inference, rangan2019vector, baker2020tramp}. The MCC ensemble considered in this work is not right orthogonally invariant, but we observe strong empirical evidence that an analogous version of Theorem \ref{th:main} holds for VAMP as well, as described in Appendix \ref{sec:vamp-results}. In a second direction, there has been much interest in \textit{spatial coupling} sensing matrices, which were used to achieve the information-theoretically optimal performance in sparse compressive sensing \cite{donoho2013information, barbier2015approximate, krzakala2012statistical}, with complementary state evolution guarantees \cite{javanmard2013state}. The concept of spatial coupling and proofs of its performance originated in the literature of error correcting codes \cite{kudekar2011threshold,kudekar2013spatially}, where it developed from the so-called convolutional codes \cite{felstrom1999time}. The connection between spatial coupling and convolution layers of neural networks, that we establish in this paper, is as far as we know novel.

Another direction of related work is the design of generative neural network architectures, and correspondingly, the design of signal recovery procedures that can make use of new generative prior models. There is a wide variety of architectures for feedforward convolutional generative priors, which are often hand-crafted to be stably trained on real world datasets. For instance, the DC-GAN architecture, studied by \citet{bora2017compressed} in compressive sensing and superresolution tasks, achieves stable training through the use of multi-channel \textit{strided} convolutional layers and batch normalization \cite{radford2015unsupervised}. Following this, the PG-GAN architecture \cite{karras2018progressive} uses multichannel convolutional layers, upsampling layers, and a parameter free alternative to batch normalization. Recently, Style-GAN has emerged as a popular architecture for generating large, high-resolution images \cite{karras2019style}. The StyleGAN generator uses residual and skip connections to encourage a hierarchical image generation process, contributing to stable training even when generating high resolution images. Style-GAN and PG-GAN, and further domain-specific modifications, have been studied as priors for a variety of signal recovery problems \cite{daras2021ilo, gu2020image}. For simplicity and theoretical tractibility, we do not consider fine-grained practical modifications like batch normalization or strided convolution, focusing instead on the essential elements of simple convolutional networks. Lastly, while our focus is on feedforward convolutional priors such as GAN/VAE networks, there is growing interest in alternative approaches to signal recovery under neural network priors, such as normalizing flows \cite{durk2018glow, asim2020glowip} and score-based generative models \cite{song2020scorebased,jalal2021posteriorsampling}. These approaches fall outside the scope of our work and may be interesting directions for future investigation. 

\section{Definition of the problem}

\subsection{Multi-channel Convolutional Matrices}
We focus our attention on \textit{multichannel convolution matrices} that have \textit{localized convolutional filters}. In this section, we introduce our notation and define the random matrix ensembles which are relevant to our result. 
We consider block structured signal vectors $x \in \mathbb{R}^{P q}$ of the form  $x = [x^{(i)}]_{i=1}^{P}$, and we refer to the blocks $x^{(i)} \in \mathbb{R}^{q}$ as `channels.' For any vector of dimension $d$, we denote by $\Pcl_d \in \mathbb{R}^{d \times d}$ the cyclic coordinate permutation matrix of order $d$, whose coordinates are $\langle e_i , \Pcl_{d} e_j \rangle = \mathbf{1}[i = j+1]$. For a block-structured vector $x \in \mathbb{R}^{Pq}$, we denote by $\Pcl_{P, q} \in \mathbb{R}^{Pq \times Pq}$ the block cyclic permutation matrix satisfying $(\Pcl_{P, q} x)^{(i)} = x^{(i+1)}$ for $1 \leq i < P$, and $(\Pcl_{P, q} x)^{(P)} = x^{(1)}$. Similarly, we denote by $\Scl_{i, j} \in \mathbb{R}^{Pq \times Pq}$ the swap permutation matrix which exchanges blocks $i, j$: $[\Scl_{i, j} x]^{(i)} = x^{(j)}$, $[\Scl_{i, j} x]^{(j)} = x^{(i)}$, and $[\Scl_{i, j} x]^{(k)} = x^{(k)}$ for $k \not = i, j$. Last, given a vector $\omega \in \mathbb{R}^k$ for $k \leq q$, denote by $\texttt{Zero-Pad}_{q, k}(\omega)$ the vector whose first $k$ coordinates are $\omega$, and whose other coordinates are zero. 
\begin{align*}
    \texttt{Zero-Pad}_{q, k}(\omega) = \begin{bmatrix} \omega_1 & \omega_2 & \ldots & \omega_k & 0 & \ldots & 0 \end{bmatrix} \in \mathbb{R}^{q}.
\end{align*}
We define the following ensemble for random multi-channel convolution matrices.
\begin{definition}[Gaussian i.i.d. Convolution] \label{dfn:c-ensemble}
Let $q \geq k$ be integers. The convolutional ensemble $\mathcal{C}(q, k)$ contains random circulant matrices $C \in \mathbb{R}^{q \times q}$ whose first row is given by $C_1 = \texttt{Zero-pad}_{q, k}[\omega]$
where $\omega \in \mathbb{R}^{k}$ has i.i.d. Gaussian coordinates $\omega_{i} \sim \mathcal{N}(0, 1/k)$. The remaining rows $C_i$ are determined by circulant structure, ie. $C_i = \Pcl_q ^{i-1}  \texttt{Zero-pad}_{q, k}[\omega]$. 
\end{definition}
Random multi-channel convolutions are block-dense matrices with independent $\mathcal{C}(q, k)$ blocks.
\begin{definition}[Multi-channel Gaussian i.i.d. Convolution] \label{dfn:mcc}
Let $D, P \geq 1$ and $q \geq k \geq 1$ be integers. The random multi-channel convolution ensemble $\mathcal{M}(D, P, k, q)$ contains random block matrices $M \in \mathbb{R}^{Dq \times Pq}$ of the form
\begin{align*}
    M = \frac{1}{\sqrt{P}} \begin{bmatrix} 
    C_{1,1} & C_{1,2} & \ldots & C_{1, P} \\
    C_{2, 1} & \ddots & &\vdots  \\
    \vdots & & & \\
    C_{D, 1} & \ldots & & C_{D, P}
    \end{bmatrix}
\end{align*}
where each $C_{i, j} \sim \mathcal{C}(q, k)$ is sampled independently. 
\end{definition}

\begin{figure}
    \centering
    \includegraphics[width=0.5\textwidth]{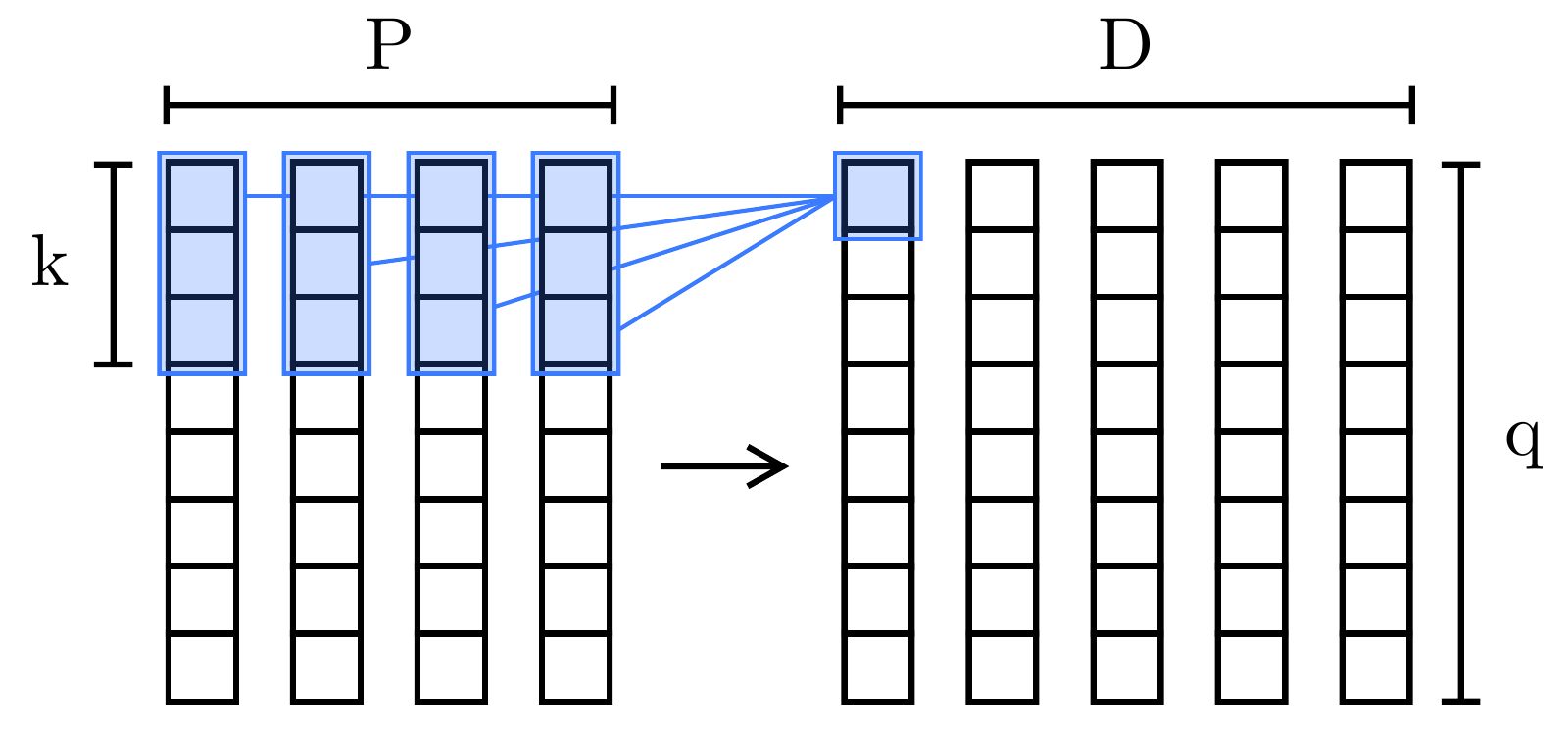}
    \caption{MCC matrices operate on $Pq$ dimensional input data, composed of $q$-dimensional signals for each of $P$ separate channels. The $i$-th output channel is a linear combination of convolutional features extracted from input channels, where $k$ is the convolutional filter size: $y^{(i)} = \sum_{j = 1 \ldots P} C_{ij} x^{(j)}$. Blue boxes show linear dependencies between signal coordinates.}
    \label{fig:conv-demo}
\end{figure}

Fig.~\ref{fig:conv-demo} gives a graphical explanation of the link between these matrices and the convolutional layers. The parameter $P$ ($D$) is the number of input (output) channels, $q$ is the dimension of the input and $k$ the filter size. 

\subsection{Thermodynamic-like Limit and Finite-size Regimes} \label{sec:finite-sizes}
We prove our main result in a thermodynamic-like limit $D, P \to \infty$ while $\beta = D/P$ is fixed and $q, k = O(1)$. From a practical perspective, convolutional layers in deep neural networks often use large channel dimensions ($D, P \gg 1$), large signal dimensions ($q \gg 1$), and a small filter size ($k = O(1)$). As an example, we show in Figure \ref{fig:dcgan-table} the sizes of convolutional layers used by the DC-GAN architecture to generate LSUN images \citep[Figure 1]{radford2015unsupervised}. 

Interestingly, our theoretical predictions do not depend explicitly on the \textit{relative} sizes of $q$ and $(D, P)$. We observe empirically that these predictions become accurate at finite sizes of $(D, P)$ which may seem small relative to $q$, and which are realistic from a practical neural network perspective. For example, in Figure \ref{fig:sparse-cs}, we observe strong empirical agreement with predictions for $q = P = 1024$ as $\beta$ and $D = \beta P$ vary.

\begin{figure}
\begin{center}
\begin{tabular}{lllll}
Layer             & $D$  & $P$ & $q$    & $k$ \\ \hline
$1 \rightarrow 2$ & 1024 & 512 & $4^2$  & 1   \\ \hline
$2 \rightarrow 3$ & 512  & 256 & $8^2$  & 5   \\ \hline
$3 \rightarrow 4$ & 256  & 128 & $16^2$ & 5   \\ \hline
$4 \rightarrow 5$ & 128  & 3   & $64^2$ & 5  
\end{tabular}
\end{center}
    \caption{System sizes for convolutional layers in a DC-GAN architecture used to generate LSUN images \citep[Figure 1]{radford2015unsupervised}. These are \textit{not} directly comparable to MCC matrices, as DCGAN uses \textit{fractionally strided convolutions}, which can be thought of as a composition of an MCC matrix with superresolution. However, they give a reasonable picture of the sizes of typical layers in convolutional neural networks.}
    \label{fig:dcgan-table}
\end{figure}

\subsection{Multi-layer AMP}
\label{sec:MLAMP}
In this section, we define a class of probabilistic graphical models (PGMs) that captures the inference problems of interest, and we state the Multi-layer Approximate Message Passing (ML-AMP) \cite{manoel2017multi} iterations, which can be used for inference on these PGMs. We consider the following signal model.
\begin{definition}[Multi-layer Signal Model]\label{dfn:ml-glm}
Let $\{W^{(l)}\}_{1 \leq l \leq L}$ be matrices of dimension $W^{(l)} \in \mathbb{R}^{n_{l-1} \times n_l}$. Let $\{\phi^{(l)}_\zeta(z) \}_{1 \leq l \leq L}$ be scalar channel functions $\phi^{(l)}_\zeta: \mathbb{R} \to \mathbb{R}$ for which $z$ is the estimation quantity and $\zeta$ represents channel noise. We write $\phi^{(l)}_\zeta(z)$ for vectors $z \in \mathbb{R}^{n_{l-1}}$ to indicate the coordinatewise application of $\phi^{(l)}$. The multi-layer GLM signal model is given by
\begin{align*}
    y = \phi^{(1)}_\zeta(W^{(1)} \phi^{(2)}_\zeta ( W^{(2)} ( \ldots \phi^{(L)}_\zeta W^{(L)} x))).
\end{align*}
We assume $x \in \mathbb{R}^{n_L}$ follows a known separable prior, $x_i \sim P_X(x)$ i.i.d., and that $\zeta \sim \mathcal{N}(0, 1)$.
\end{definition}
The full estimation quantities of the model are the coordinates of the vectors $\{h^{(l)}\}_{1 \leq l \leq L}$, $\{z^{(l)}\}_{1 \leq l \leq L}$, which are related by
\begin{align} \label{eqn:signal-model}
    y_\mu & = \phi^{(1)}_\zeta(z^{(1)}) & \quad z^{(1)}_\mu = \sum_i W^{(1)}_{\mu i} h^{(1)}_i\, , \\ \nonumber
    h^{(1)}_i & = \phi^{(2)}_\zeta(z^{(2)}) & \quad z^{(2)}_\mu = \sum_i W^{(2)}_{\mu i} h^{(2)}_i \, ,\\ \nonumber
    & \vdots & \\ \nonumber
    h^{(L-1)}_i & = \phi^{(L)}_\zeta(z^{(L)} ) & \quad z^{(L)}_\mu = \sum_i W^{(L)}_{\mu i} x _i\, 
\end{align}
and the corresponding conditional probabilities, which define the factor nodes of the underlying PGM, are given by 
\begin{align*}
    P^{(l)}( h \mid z ) & = \int \, d\zeta \, e^{-\frac{1}{2} \zeta^2} \delta(h - \phi_\zeta(z))\, .
\end{align*}
To compute the posterior marginals, ML-AMP iteratively updates the parameters of independent 1D Gaussian approximations to each marginal. Each coordinate $h^{(l)}_i(t)$ has corresponding parameters $\{A^{(l)}_{i}(t), B^{(l)}_{i}(t) \}$ and each $z^{(l)}_\mu(t)$ has corresponding $\{V^{(l)}_{\mu}(t), \omega^{(l)}_{\mu}(t)\}$, where $t \geq 1$ indexes the ML-AMP iterations. The recursive relationship between these parameters is defined in terms of scalar \textit{denoising functions}, $\hat{h}^{(l)}$ and $g^{(l)}$, which compute posterior averages of the estimation quantities given their prior parameters.  

In general, these denoising functions can be chosen (up to regularity assumptions) to adjust ML-AMP's performance in applied settings, such as in \cite{metzler2015bm3d}, and in these cases the denoisers may be nonseparable vector valued functions.  However, in the separable, Bayes-optimal regime where $P_x(x)$ and $P^{(l)}(h \mid z)$ are known, the optimal denoisers are given by,
\begin{align} \label{eqn:bayes-denoiser}
\hat{h}_{i}^{(l)}(t+1) & := \partial_{B} \log \mathcal{Z}^{(l+1)}(A_{i}^{(l)}, B_{i}^{(l)}, V_{i}^{(l+1)}, \omega_{i}^{(l+1)}) \\ \nonumber
\sigma_{i}^{(l)}(t+1) & := \partial_{B} \hat{h}_{i}^{(l)}(t+1) \\ \nonumber
g_{\mu}^{(l)}(t) & := \partial_{\omega} \log \mathcal{Z}^{(l)}(A_{\mu}^{(l-1)}, B_{\mu}^{(l-1)}, V_{\mu}^{(l)}, \omega_{\mu}^{(l)})  \\ \nonumber
\eta_{\mu}^{(l)}(t) & := \partial_{\omega} g_{\mu} ^{(l)} (t) \\ \nonumber 
    \mathcal{Z}^{(l)}(A, B, V, \omega) & := \frac{1}{\sqrt{2 \pi V}}\int P^{(l)}(h \mid z) \exp\left( B h - \frac{1}{2} A h^2 - \frac{(z-\omega)^2}{2 V } \right) \, dh \, dz 
\end{align} 
where $2 \leq L \leq L-1$, $t \geq 2$ and the prior parameters on the right hand side are taken at iteration $t \geq 2$. The corresponding ML-AMP iterations are given by,
\begin{align} \label{eqn:ml-amp-iterate}
    V^{(l)}_{\mu}(t) & = \sum_{i} [W^{(l)}_{\mu i}]^2\,  \sigma^{(l)}_{i} (t) & \qquad 
    \omega^{(l)}_{\mu} (t) & = \sum_i W^{(l)}_{\mu i}\, \hat{h}^{(l)}_{i}(t) - V^{(l)}_{\mu}(t) \, g^{(l)}_{\mu}(t-1) \\ \nonumber
    A^{(l)}_{i} (t) & = -\sum_\mu [W^{(l)}_{\mu i}]^2 \, \eta_\mu^{(l)}(t) & \qquad 
    B^{(l)}_{i} (t) & = \sum_\mu W^{(l)}_{\mu i} g^{(l)}_{\mu}(t) + A^{(l)}_{i}(t) \hat{h}_{i}^{(l)}(t).
\end{align}
For the boundary cases $t=1$, $l = 1$, and $l=L$, the iterations \eqref{eqn:bayes-denoiser}, \eqref{eqn:ml-amp-iterate} are modified as follows. 
\begin{enumerate}
    \item At $t=1$, we initialize $B^{(l)}_i \sim P^{(l)}_{B_0}$ and $\omega^{(l)}_\mu \sim P^{(l)}_{\omega_0}$, where $P^{(l)}_{B_0}$, $P^{(l)}_{\omega_0}$ are the distributions of the signal model parameters \eqref{eqn:signal-model} when $x_i \sim P_X$. We take $(A^{(l)}_i)^{-1} = \text{Var}(B^{(l)}_i)$ and $V^{(l)}_\mu = \text{Var}(\omega^{(l)}_\mu)$. 
    \item At $l=1$, the denoiser $g_\mu^{(1)}(t) = \partial_\omega \log \mathcal{Z}^{(1)}(y, V^{(1)}_\mu, \omega^{(1)}_\mu)$, where
    \begin{align*}
        \mathcal{Z}^{(1)}(y, V^{(1)}_\mu, \omega^{(1)}_\mu) = \frac{1}{\sqrt{2 \pi V}} \int P^{(1)}(y \mid z) \exp\left(-\frac{(z - \omega_\mu^{(1)})^2}{2 V^{(1)}_\mu}\right)  \, dz.
    \end{align*}
    \item At $l=L$, the denoiser $\hat{h}^{(L)}(t) = \partial_B \log \mathcal{Z}^{(L)}(A^{(L)}_i, B^{(L)}_i)$, where
    \begin{align*}
        \mathcal{Z}^{(L)}(A^{(L)}_i, B^{(L)}_i) = \int P_X(h) \exp \left( B^{(L)}_\mu h - \frac{1}{2} A^{(L)}_\mu h^2 \right) \, dh\, .
    \end{align*}
\end{enumerate}

\subsubsection{Computational Savings of MCC Matrices} 
As ML-AMP requires only matrix-vector products, its computational burden can be significantly reduced when using structured and/or sparse sensing matrices. In our setting, multi-channel convolutions $M \sim \text{MCC}(D, P, q, k)$ have $DP k$ independent, nonzero coordinates, compared to $DP q^2$ nonzero coordinates of a Gaussian i.i.d. matrix. Typically, $k$ represents the size of a localized filter applied to a larger image, with $k \ll q$  \citep[Section 3.4]{gonzalez2008digital}, leading to significant space savings by a factor $k/q^2$. This same is true in convolutional neural networks, where the use of localized convolutional filters represents an inductive bias towards localized features that is considered a key aspect of their practical success \cite{krizhevsky2012imagenet, zeiler2014visualizing}.

In addition to space savings, specialized matrix-vector product implementations can reduce the time complexity of ML-AMP with MCC sensing matrices. Simple routines for sparse matrix-vector products run in time proportional to the number of nonzero entries, resulting in the same $k/q^2$ constant factor speed up for MCC matrix-vector products. Alternatively, if $k \gg \log q$, then a simple algorithm using a fast Fourier transform for convolution-vector products yields time complexity $O(DP q \log q)$. Such an algorithm is sketched Appendix \ref{sec:fast-mcc-vec}.

\section{Main result} \label{sec:main-result}
We now state our main technical result, starting with the set of required assumptions.
\begin{itemize}[wide=1pt]
\label{set:assump}
    \item[(A1)] for any $1 \leqslant l \leqslant L$, the function $\phi^{l}$ is continuous and there exists a polynomial $b^{(l)}$ of finite order such that, for any $x \in \mathbb{R}$, $\vert \phi^{(l)}(x)\vert \leqslant \vert b^{(l)}(x)\vert$
    \item[(A2)] for any $1\leqslant l \leqslant L$, the matrix $\mathbf{W}^{(l)}$ is sampled from the ensemble $\mathcal{M}(D^{l},P^{l},k^{l},q^{l})$ where $P^{l}q^{l} = D^{l-1}q^{l-1}$
    \item[(A3)] the iteration \ref{eqn:ml-amp-iterate} is initialized with a random vector independent of the mixing matrices verifying $\frac{1}{N}\norm{\mathbf{h}_{0}}_{2}^{2}<+\infty$ almost surely
    \item[(A4)] for any $1 \leqslant l \leqslant L$, $D_{l},P_{l} \to \infty$ with constant ratio $\beta_{l} = D_{l}/P_{l}$, with finite $q_{l}$.
\end{itemize}   
Under these assumptions, we may define the following \emph{state evolution} recursion
\begin{definition}[State Evolution] \label{dfn:state-evolution}
Consider the following recursion,
\begin{align}
\label{eq:SE1}
    \hat{m}^{(l)}(t) & = - \beta^{(l)} \mathbb{E}^{(l)}[\partial_\omega g(\hat{m}^{(l-1)}, \hat{m}  b, \tau_1 - m^{(l)}, h)] \\
    m^{(l-1)}(t+1) & = \mathbb{E}^{(l)}[h\,  \hat{h}^{(l-1)}(\hat{m}^{(l-1)}, \hat{m}b, \tau_1 - m^{(l)}, h)],
    \label{eq:SE2}
\end{align} 
where $\tau^{(l)}$ is the second moment of $P_{B_0}^{(l)}$, where the right hand side parameters are taken at time $t$, and the expectations $\mathbb{E}^{(l)}$ are taken with respect to
\begin{align*}
    P^{(l)}(w, z, h, b) =  P^{(l)}_{\text{out}}(h \mid z) \, \mathcal{N}(z; w, \tau^{(l)}-m^{(l)}) \, \mathcal{N}(w; 0, m^{(l)}) \, \mathcal{N}(b; \hat{m}^{(l-1)} h, \hat{m}^{(l-1)}). 
\end{align*}
\end{definition}
At $t=1$, the state evolution is initialized at $\kappa^{(l)} = 0$ and $(\hat{\kappa}^{(l)})^{-1} = \tau^{(l)}$. At the boundaries $l = 1, L$, the expectations are modified analogously to the ML-AMP iterations as described by \citet{manoel2017multi}. 
We then have the following asymptotic characterization of the iterates from the convolutional ML-AMP algorithm
\begin{theorem}
\label{th:main}
Under the set of assumptions (A1)-(A4), for any sequences of uniformly pseudo-Lipschitz functions $\psi^{N}_{1},\psi^{N}_{2}$ of order $k$, for any $1 \leqslant l \leqslant L$ and any $t \in \mathbb{N}$, the following holds 
\begin{align} &\frac{1}{D_{l}q_{l}}\sum_{i=1}^{D_{l}q_{l}}\psi_{1}(\omega_{i}^{(l)}(t)) \approxP \mathbb{E}\left[\psi_{1}\left(Z^{l}(t)\right)\right] \\
&\frac{1}{P_{l}q_{l}}\sum_{i=1}^{P_{l}q_{l}}\psi_{2}(B_{i}^{(l)}(t)) \approxP \mathbb{E}\left[\psi_{2}\left(\hat{Z}^{l}(t)\right)\right]
\end{align}
where $Z^{l}(t) \sim \mathcal{N}(0,\kappa^{l}(t))$, $\hat{Z}^{l}(t) \sim \mathcal{N}(0,\hat{\kappa}^{l}(t))$ are independent random variables.
\end{theorem}

\subsection{Proof Sketch}
The proof of Theorem \ref{th:main}, which is given in Appendix \ref{sec:AppA}, has two key steps. First, we construct permutation matrices $U, \tilde{U}$ such that for $W \sim \text{MCC}(D, P, q, k)$, the matrix $\tilde{W} = U W \tilde{U}^T$ is a block matrix whose blocks either have i.i.d. Gaussian elements or are zero valued, and has a block-circulant structure. The effect of the permutation is that entries of $\tilde{W}$ which are correlated due to circulant structure of $W$ are relocated to different blocks. Once these permutation matrices are defined, we define a new, matrix-valued AMP iteration involving the dense Gaussian matrices obtained from the permutations, and whose non-linearities account for the block-circulant structures and the permutation matrices. The state evolution of this new iteration is proven using the results of \cite{gerbelot2021graph}. This provides an explicit example of how the aforementioned results can be used to obtain rigorous, non Bayes-optimal SE equations on a composite 
AMP iteration. The separability assumption is key in showing that the AMP iterates obtained with the convolutional matrices can be \emph{exactly}
embedded in a larger one. Note that this is a stronger result than proving SE equations for an algorithm that computes marginals of a random convolutional posterior: we show the SE equations are the same as in the dense case. We finally invoke the Nishimori conditions, see e.g. \cite{krzakala2012statistical}, to simplify the generic, non Bayes-optimal SE equations to the Bayes-optimal ones. 

The idea of embedding a non-separable effect such as a block-circulant structure or different variances in a mixing matrix is the core idea in the proofs of SE equations for spatially coupled systems, notably as done in \cite{javanmard2013state,donoho2013information}.
We note that in the numerical experiments shown at Figure \ref{fig:sparse-cs}, the parameter $q$, considered finite in the proof, is actually comparable to the number of channel, considered to be extensive. Empirically we observe that this does not hinder the validity of the result, something that was also observed in the spatial coupling literature, e.g. \cite{krzakala2012statistical}, where large number of different blocks in spatially coupled matrices were considered, with convincing numerical agreement. \\
\noindent
The existence of permutations matrices verifying the property described above is formalized in the following lemma:
\begin{lemma}[Permutation Lemma] \label{lem:permutation}
Let $W \sim \mathcal{M}(D, P, k, q)$ be a multi-channel convolution matrix. There exist row and column permutation matrices $U \in \mathbb{R}^{Dq \times Dq}$, $\tilde{U} \in \mathbb{R}^{Pq \times Pq}$ such that $\tilde{W} = U W \tilde{U}^T$ is a block-convolutional matrix with dense, Gaussian i.i.d. blocks. That is,
\begin{align*} 
    \tilde{W} = \frac{1}{\sqrt{k}} \begin{bmatrix}
    A^{(1)} & A^{(2)} & \ldots & A^{(k)} & & & & &  \\
     & A^{(1)} & A^{(2)} & \ldots & A^{(k)} & & & & \\
    \vdots & & A^{(2)} & \ldots & A^{(k)} & & &&  \\
    & & & & \ddots & & & \vdots \\
    A^{(2)} & A^{(3)} & \ldots & A^{(k)} & & & & & A^{(1)}
    \end{bmatrix} 
\end{align*}
where each $A^{(s)} \in \mathbb{R}^{D, P}$, $1 \leq s \leq k$ has i.i.d. $\mathcal{N}(0, 1/P)$ coordinates.
\end{lemma}

\begin{proof}
Consider the elements of the matrix $M$ which are non-zero and sampled i.i.d. as opposed to exact copies of other variables. They are positioned on the first line of each block of size $q\times q$, and thus the indexing for their lines and columns can be written as $M_{aq+1,bq+c}$ where $a,b,c$ are integers such that $0\leqslant a \leqslant D-1$, $0\leqslant b \leqslant P-1$ and $1\leqslant c \leqslant k$. The integers $a,b$ describe the position of the $q\times q$ block the variable is in, and $c$ describes, for each block, the position in the initial random Gaussian vector of size $k$ that is zero-padded and circulated to generate the block. The goal is to find the mapping that groups these variables into $k$ dense blocks of extensive size $D \times P$. To do so, one can use the following bijection $\tilde{M}_{\gamma,\alpha P+\beta} = M_{aq+1,bq+c}$ where $\gamma = a+1$, $\alpha = c-1$ and $\beta=b+1$. By doing this, $c$ becomes the block index and $a,b$ become the position in the dense block. This mapping can be represented by left and right permutation matrices which also prescribe the permutation for the rest of the elements of $M$. A graphical sketch of this coordinate permutation is shown in Figure \ref{fig:permutation-lemma-diagram}.
\end{proof}

\begin{figure}
    \begin{align*}
        \begin{bmatrix}
        \resizebox{0.45\columnwidth}{!}{
        $\begin{array}{ccc|ccc|ccc}
        z_{11} & w_{11} &        & z_{12} & w_{12} &        & z_{13} & w_{13} &        \\
       & z_{11} & w_{11} &        & z_{12} & w_{12} &        & z_{13} & w_{13} \\
w_{11} &        & z_{11} & w_{12} &        & z_{12} & w_{13} &        & z_{13} \\ \hline
z_{21} & w_{21} &        & z_{22} & w_{22} &        & z_{23} & w_{23} &        \\
       & z_{21} & w_{21} &        & z_{22} & w_{22} &        & z_{23} & w_{23} \\
w_{21} &        & z_{21} & w_{22} &        & z_{22} & w_{23} &        & z_{23} \\ \hline
z_{31} & w_{31} &        & z_{32} & w_{32} &        & z_{33} & w_{33} &        \\
       & z_{31} & w_{31} &        & z_{32} & w_{32} &        & z_{33} & w_{33} \\
w_{31} &        & z_{31} & w_{32} &        & z_{32} & w_{33} &        & z_{33} \\ \hline
z_{41} & w_{41} &        & z_{42} & w_{42} &        & z_{43} & w_{43} &        \\
       & z_{41} & w_{41} &        & z_{42} & w_{42} &        & z_{43} & w_{43} \\
w_{41} &        & z_{41} & w_{42} &        & z_{42} & w_{43} &        & z_{43}
        \end{array}$
        }
        \end{bmatrix} 
        \quad
        \begin{bmatrix}
        \resizebox{0.45\columnwidth}{!}{
        $\begin{array}{ccc|ccc|ccc}
z_{11} & z_{12} & z_{13} & w_{11} & w_{12} & w_{13} &        &        &        \\
z_{21} & z_{22} & z_{23} & w_{21} & w_{22} & w_{23} &        &        &        \\
z_{31} & z_{32} & z_{33} & w_{31} & w_{32} & w_{33} &        &        &        \\ 
z_{41} & z_{42} & z_{43} & w_{41} & w_{42} & w_{43} &        &        &        \\ \hline
       &        &        & z_{11} & z_{12} & z_{13} & w_{11} & w_{12} & w_{13} \\
       &        &        & z_{21} & z_{22} & z_{23} & w_{21} & w_{22} & w_{23} \\ 
       &        &        & z_{31} & z_{32} & z_{33} & w_{31} & w_{32} & w_{33} \\
       &        &        & z_{41} & z_{42} & z_{43} & w_{41} & w_{42} & w_{43} \\ \hline
w_{11} & w_{12} & w_{13} &        &        &        & z_{11} & z_{12} & z_{13} \\ 
w_{21} & w_{22} & w_{23} &        &        &        & z_{21} & z_{22} & z_{23} \\
w_{31} & w_{32} & w_{33} &        &        &        & z_{31} & z_{32} & z_{33} \\
w_{41} & w_{42} & w_{43} &        &        &        & z_{41} & z_{42} & z_{43}
        \end{array}$
        }
        \end{bmatrix} 
    \end{align*}
    \caption{A sketch of the permutation lemma applied to matrix $W \sim \text{MCC}(4, 3, 3, 2)$. Left: $W$ before permutation. Right: after permutation, $U W \tilde{U}^T$.}
    \label{fig:permutation-lemma-diagram}.
\end{figure}

\section{Numerical Experiments} \label{sec:numerics}
In this section, we compare state evolution predictions from Theorem \ref{th:main} with a numerical implementation of the ML-AMP algorithm described in Section \ref{sec:MLAMP}. 

Our first experiment, shown in Figure \ref{fig:sparse-cs}, is a noisy compressive sensing task under a sparsity prior $P_X(x) = \rho \mathcal{N}(x; 0, 1) + (1 - \rho) \delta(x)$, where $\rho$ is the expected fraction of nonzero components of $x_0$. Measuremements are generated $y_0 = W x_0 + \eta$ for noise $\eta \sim \mathcal{N}(0, 10^{-4})$, where $W \sim \text{MCC}(D, P, q, k)$. We show recovery performance at sparsity levels $\rho \in \{0.25, 0.5, 0.75\}$ as the measurement ratio $\beta = D/P$ varies, averaged over 10 independent AMP iterates. Additionally, we show convergence of the (averaged) AMP iterates for sparsity $\rho = 0.25$ at a range of $\beta$ near the recovery threshold. We observe strong agreement between AMP empirical performance and the state evolution prediction. The system sizes are $P = 1024$, $q=1024$, with $D = \beta P$ varying.

In Figure \ref{fig:multi-layer}, we show two examples of $L = 2, 3, 4$ layer models following Equation \eqref{eqn:signal-model}. In both, the output channel $l=1$ generates noisy, compressive linear measurements $y = z^{(1)} + \zeta$ for $\zeta_i \sim \mathcal{N}(0, \sigma^2)$ and for dense couplings $W^{(1)}_{ij} \sim \mathcal{N}(0, 1/n^{(1)})$. Layers $2 \leq l \leq 4$ use MCC couplings $W^{(l)} \sim \text{MCC}(D_l, P_l, q, k)$, where $q P_l = n_l$ and $D_l = \beta P_l = q n_{l-1}$. Channel functions $\{\phi^{(l)}\}$ vary across the two experiments. The input prior is $P_X(x) = \mathcal{N}(x; 0, 1)$ and model has $q=10$ channels, filter size $k=3$, noise level $\sigma^2 = 10^{-4}$, input dimension $n^{(L)} = 5000$, layerwise aspect ratios $\beta^{(L)} = 2$ and $\beta^{(l)} = 1$ for $2 \leq l < L$. The channel aspect ratio $\beta^{(1)}$ varies in each experiment. 

\begin{figure}[h]
    \centering
    \includegraphics[width=0.51\linewidth]{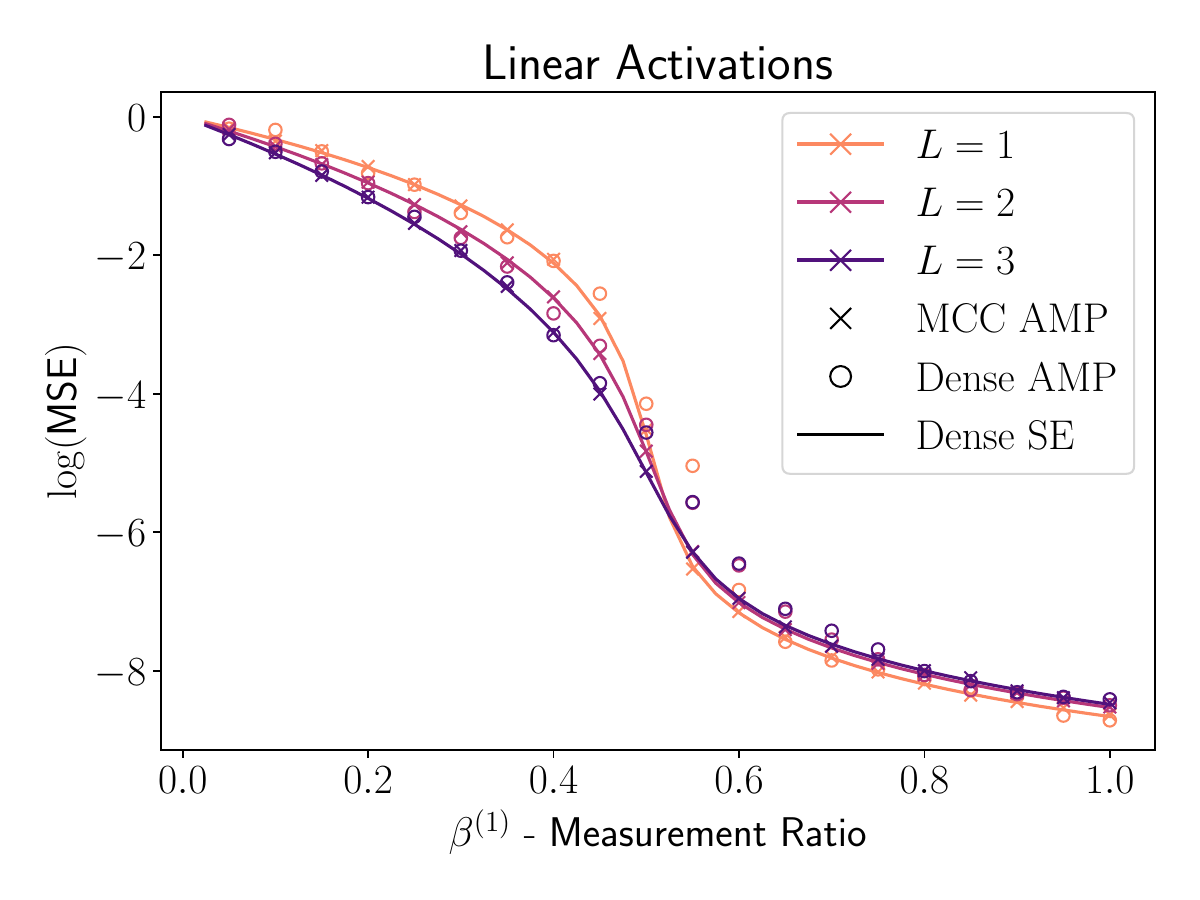}
    \includegraphics[width=0.48\linewidth]{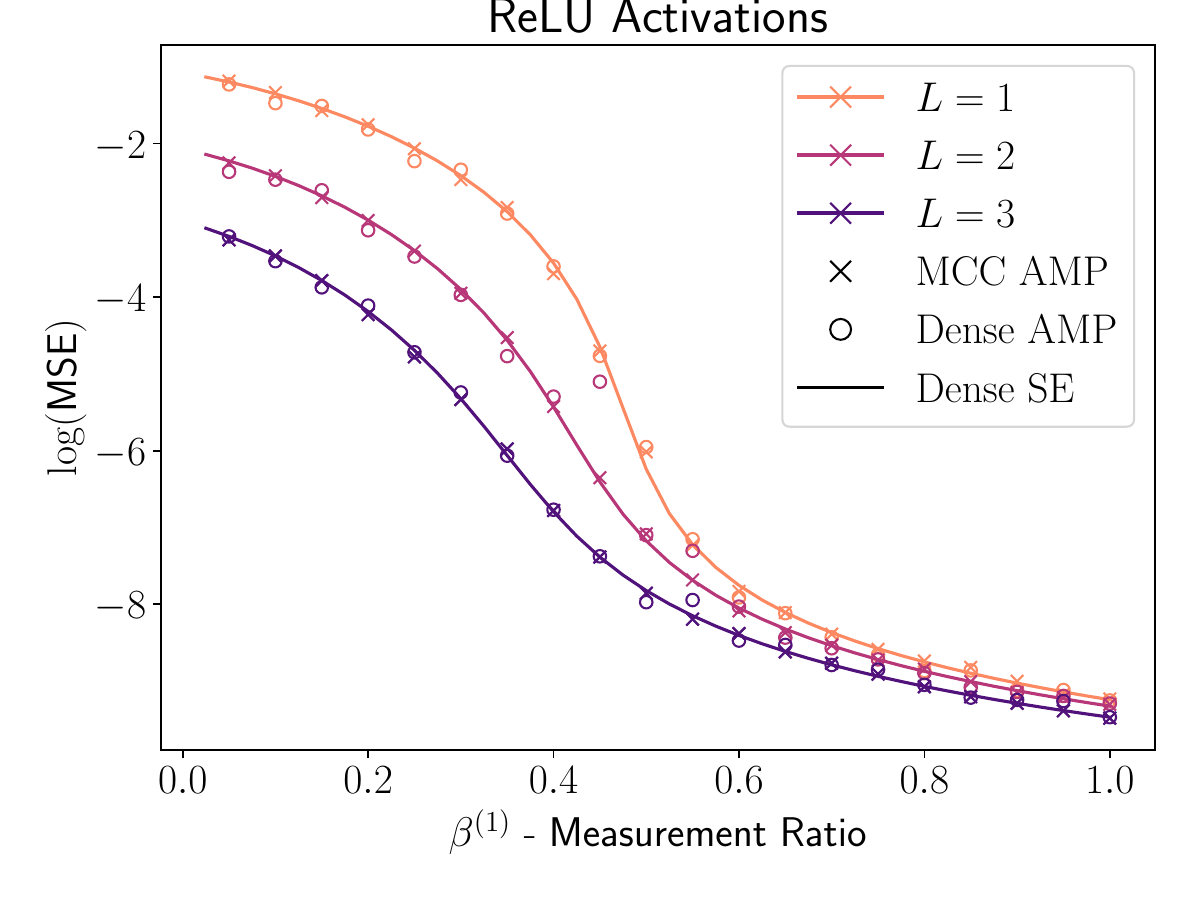}
    \caption{ML-AMP compressive sensing recovery under multichannel convolutional designs (crossed) and the state evolution for the corresponding fully connected model (lined). For comparison, we also plot the corresponding fully connected AMP iterations (dotted), in which $W^{(l)} \in \mathbb{R}^{D_l \times P_l}$ with $W_{ij} \sim \mathcal{N}(0, 1/P_l)$, with the dimensions of the prior and output channel adjusted appropriately.
    Left: For $2 \leq l \leq L$, the channel functions are $\phi^{(l)}(z; \zeta) = z + \zeta$ where $\zeta_i \sim \mathcal{N}(0, \sigma^2)$. 
    Right: 
    For $2 \leq l \leq L$, the channel functions are $\phi^{(l)}(z; \zeta) = \max(z, 0)$ where the maximum is applied coordinatewise. This channel function is the popular ReLU activation function used by generative convolutional neural networks such as in \cite{radford2015unsupervised,bora2017compressed}. }
    \label{fig:multi-layer}
\end{figure}

We compare the state evolution equations to empirical AMP results in two cases. In the left panel, we show multilayer models with identity channel functions, and in the right panel, we show models with ReLU channel functions. The latter model captures a simple but accurate example of a convolutional generative neural network.


\section{Discussion and Future Work}\label{sec:discussion}

We have proven state evolution recursions for the ML-AMP algorithm for signal recovery from multi-layer convolutional networks. We consider networks whose weight matrices are drawn either i.i.d. Gaussian or from an ensemble of random multi-channel convolution matrices. Interestingly, under a separable prior and separable channel functions, these two matrix ensembles yield the same state evolution (up to a rescaling). These predictions closely match empirical observations in compressive sensing under a sparsity prior (Figure \ref{fig:sparse-cs}) and under multi-layer priors (Figure \ref{fig:multi-layer}). 

Lemma \eqref{lem:permutation} allows to rewrite an MCC matrix $M$ as a block circulant matrix $\tilde{M}$ with random extensive blocks, reminiscent of the block structure of spatially coupled sensing matrices. As a consequence of separability, the nonzero blocks of $\tilde{M}$ have identical statistics, which is key to our equivalence theorem between MCC matrices and their dense i.i.d. counterparts. This is in contrast to spatial coupling, where extensive blocks may have different variances, or equivalently when denoising functions may be non-separable. We prove in Appendix \ref{sec:AppA} a more general result for non-separable channel functions, of which Theorem \ref{th:main} is a specialization to the separable case. In Appendix \ref{app:nonseparable}, we discuss a potential application to signal recovery with non-i.i.d. convolutional filters, in which the dynamics of ML-AMP is expected to differ from the analogous fully connected model. Ultimately, studying non-separable models is an interesting and potentially fruitful avenue for future work. 

Another important direction for future work is to go beyond random convolutional layers and study how to account for trained layers in the ML-AMP algorithm and its state evolution.

\begin{ack}
M.D. acknowledges funding from Northeastern University's Undergraduate Research \& Fellowships office and the Goldwater Award. We acknowledge funding from the ERC under the European Union’s Horizon 2020 Research and Innovation Program Grant Agreement 714608-SMiLe.
\end{ack}

\clearpage

\bibliographystyle{plainnat}
\bibliography{citations}

\clearpage
\appendix

\section{Proof of the main theorem} \label{sec:AppA}

The proof of the main theorem is presented in this section. We start with a generic result on a family of AMP iterations including the (non Bayes-optimal) MLAMP one, using the framework of \cite{gerbelot2021graph}, from which we remind the required notions.
\subsection{Notations and definitions}
If $f: \R^{N \times q} \to \R^{N \times q}$ is an function and $i \in \{1, \dots N\}$, we write $f_{i}:\mathbb{R}^{N \times q} \to \mathbb{R}^{q}$  the component of $f$ generating the $i$-th line of its image, i.e., if $\bX \in \mathbb{R}^{N \times q}$, 
\begin{equation*}
    f(\bX) = \begin{bmatrix}
    f_{1}(\bX) \\
    \vdots \\
    f_{N}(\bX)\end{bmatrix} \in \mathbb{R}^{N \times q} \, .
\end{equation*}
We write $\frac{\partial f_{i}}{\partial \bX_{i}}$ the $q\times q$ Jacobian containing the derivatives of $f_{i}$ with respect to (w.r.t.) the $i$-th line $\bX_{i}\in \mathbb{R}^{q}$:
\begin{equation}
\label{eq:Ons_jacob}
    \frac{\partial f_{i}}{\partial \bX_{i}} = \begin{bmatrix}\frac{\partial (f_{i}(\bX))_{1}}{\partial \bX_{i1}} & \dots & \frac{\partial (f_{i}(\bX))_{1}}{\partial \bX_{iq}} \\
    \vdots& &\vdots \\
    \frac{\partial (f_{i}(\bX))_{q}}{\partial \bX_{i1}} & \dots & \frac{\partial (f_{i}(\bX))_{q}}{\partial \bX_{iq}}
    \end{bmatrix} \in \mathbb{R}^{q \times q} \, .
\end{equation}
For two sequences of random variables $X_{n},Y_{n}$, we write $X_{n} \approxP Y_{n}$ when their difference converges in probability to $0$, i.e., $X_{n}-Y_{n} \xrightarrow[]{P}0$.
Oriented graphs with a set of vertices $V$ and edges $\overrightarrow{E}$ are denoted $G = (V,\overrightarrow{E})$. The set of edges may be split into right-pointing and left-pointing edges, i.e., $\overrightarrow{E} = \left\{\overrightarrow{e}_{1},...,\overrightarrow{e}_{L}\right\}, \overleftarrow{E} = \left\{\overleftarrow{e}_{1},...,\overleftarrow{e}_{L}\right\}$. 
\begin{definition}[pseudo-Lipschitz function]
\label{def:pseudo-lip}
For $k \in \mathbb{N}^{*}$ and any $N,m \in \mathbb{N}^{*}$, a function $\Phi : \mathbb{R}^{N \times q} \to \mathbb{R}^{m \times q}$ is said to be \emph{pseudo-Lipschitz of order k} if there exists a constant L such that for any $\mathbf{x},\mathbf{y} \in \mathbb{R}^{N \times q}$, 
\begin{equation}
    \frac{\norm{\Phi(\mathbf{x})-\Phi(\mathbf{y})}_{F}}{\sqrt{m}} \leqslant L \left(1+\left(\frac{\norm{\mathbf{x}}_{F}}{\sqrt{N}}\right)^{k-1}+\left(\frac{\norm{\mathbf{y}}_{F}}{\sqrt{N}}\right)^{k-1}\right)\frac{\norm{\mathbf{x}-\mathbf{y}}_{F}}{\sqrt{N}}
\end{equation}
For a function $\phi : \mathbb{R}\to \mathbb{R}$, the property becomes 
\begin{align}
    \forall \thickspace (x,y) \in \mathbb{R}^{2}, \thickspace \vert \phi(x)-\phi(y) \vert \leqslant L(1+\vert x \vert^{k-1} + \vert y \vert^{k-1}) \vert x-y \vert
\end{align}
and a straightforward calculation shows that for any scalar pseudo-Lipshitz function of order $2$, the function
\begin{align}
    \varphi : \mathbb{R}^{d} \to \mathbb{R}, \\
    \mathbf{x} \mapsto \frac{1}{d}\sum_{i=1}^{d}\phi(x_{i})
\end{align}
is pseudo-Lipschitz of order 2 according to the definition above. This definition is handy for proofs involving non-separable functions and leads to Gaussian
concentration using the Gauss-Poincaré inequality (see Lemma C.8. from  \cite{berthier2020state}), while in the separable case, a strong law of large number is proven for a class 
of distributions including sub-Gaussian ones in Lemma 5 of \cite{bayati2011dynamics}.
\end{definition}
\subsection{State evolution for generic multilayer AMP iterations with matrix valued variables and dense Gaussian matrices}
In the notations of \cite{gerbelot2021graph}, consider the AMP iteration indexed by the following directed graph $G = (V, \overrightarrow{E})$, where the set of vertices is denoted $V = \left\{v_{0},v_{1},...,v_{L}\right\}$, and the set of edges $\overrightarrow{E} = \left\{\overrightarrow{e}_{1},...,\overrightarrow{e}_{l},\overleftarrow{e}_{1},...,\overleftarrow{e}_{L}\right\}$. For any edge $\overrightarrow{e}_{l}$, the corresponding matrix $\mathbf{A}_{\overrightarrow{e}_{l}}$ has dimensions $\mathbb{R}^{n_{l}\times n_{l-1}}$ with $\mathbf{A}_{\overleftarrow{e}_{l}} = \mathbf{A}_{\overrightarrow{e}_{l}}^{\top}$, and the variables $\mathbf{x}_{\overrightarrow{e}_{l}} \in \mathbb{R}^{n_{l}\times q},\mathbf{x}_{\overleftarrow{e}_{l}} \in \mathbb{R}^{n_{l-1}\times q}$ for some finite $q \in \mathbb{N}$, with $N = \sum_{l=1}^{L}n_{l}$. Finally, we define the non-linearities of the iteration by specifying the variables they are acting on as follows:
\begin{itemize}
\item $f^{t}_{\overrightarrow{e}_{1}}: \mathbb{R}^{n_{0}\times q}\to \mathbb{R}^{n_{0}\times q}, \mathbf{x}^{t}_{\overleftarrow{e}_{1}} \mapsto f^t_{\overrightarrow{e}_1}\left(\bx^t_{\overleftarrow{e_1}}\right)$,
\item for any $2\leqslant l\leqslant L$, $f^{t}_{\overrightarrow{e}_{l}}:(\mathbb{R}^{n_{l-1}\times q})^{2} \to \mathbb{R}^{n_{l-1} \times q}$, $(\mathbf{x}^{t}_{\overrightarrow{e}_{l-1}},\mathbf{x}^{t}_{\overleftarrow{e}_{l}}) \mapsto f^{t}_{\overrightarrow{e}_{l}}(\mathbf{x}^{t}_{\overrightarrow{e}_{l-1}},\mathbf{x}^{t}_{\overleftarrow{e}_{l}})$, 
\item for any $1\leqslant l\leqslant L-1$, $f^{t}_{\overleftarrow{e}_{l}}:(\mathbb{R}^{n_{l}\times q})^{3} \to \mathbb{R}^{n_{l} \times q}$, $(\mathbf{x}^{t}_{\overrightarrow{e}_{l}},\mathbf{x}^{t}_{\overleftarrow{e}_{l+1}}) \mapsto f^{t}_{\overleftarrow{e}_{l}}(\mathbf{A}_{\overrightarrow{e}_{l}}\mathbf{w}_{\overrightarrow{e}_{l}}, \mathbf{x}^{t}_{\overrightarrow{e}_{l}},\mathbf{x}^{t}_{\overleftarrow{e}_{l+1}})$
\item $f^{t}_{\overleftarrow{e}_{L}}: (\mathbb{R}^{n_{L}\times q})^{2}\to \mathbb{R}^{n_{L}\times q}, \mathbf{x}^{t}_{\overleftarrow{e}_{L}} \to f^t_{\overleftarrow{e}_L}\left(\mathbf{A}_{\overrightarrow{e}_{L}}\mathbf{w}_{\overrightarrow{e}_{L}}, \bx^t_{\overleftarrow{e_L}}\right)$
\end{itemize}
where $\mathbf{w}_{\overrightarrow{e}_{1}},...,\mathbf{w}_{\overrightarrow{e}_{L}}$ are low-rank matrices respectively in $\mathbb{R}^{n_{0}\times q}, ..., \mathbb{R}^{n_{L-1}\times q}$, whose rows are sampled i.i.d. from subgaussian probability distributions in $\mathbb{R}^{q}$. The graph indexing the iteration then reads:
\begin{center}
	\begin{tikzpicture}[scale = 1]
    \tikzstyle{point}=[draw,circle,minimum width=3em];
    \tikzstyle{fleche}=[->];
    \node[point] (v0) at (0,0) {$v_0$};
    \node[point] (v1) at (4,0) {$v_1$};
    \node[point] (v2) at (8,0) {$v_2$};
    \node[] (v3) at (11,0) {\Huge$\cdots$};
    \node[point] (vL) at (14,0) {$v_L$};
    \draw[fleche] (v0) to[bend left]
    node[above,very near start]{$f^t_{\overrightarrow{e_1}}$}
    node[above,midway]{$\bA_{\overrightarrow{e_1}}$} node[below,midway]{$\overrightarrow{e_1}$} 
    node[above,very near end]{$\bx^t_{\overrightarrow{e_1}}$}
    (v1);
    \draw[fleche] (v1) to[bend left]
    node[below,very near start]{$f^t_{\overleftarrow{e_1}}$}
    node[below,midway]{$\bA_{\overrightarrow{e_1}}^\top$} node[above,midway]{$\overleftarrow{e_1}$} 
    node[below,very near end]{$\bx^t_{\overleftarrow{e_1}}$}
    (v0);
    \draw[fleche] (v1) to[bend left]
    node[above,very near start]{$f^t_{\overrightarrow{e_2}}$}
    node[above,midway]{$\bA_{\overrightarrow{e_2}}$} node[below,midway]{$\overrightarrow{e_2}$} 
    node[above,very near end]{$\bx^t_{\overrightarrow{e_2}}$}
    (v2);
    \draw[fleche] (v2) to[bend left]
    node[below,very near start]{$f^t_{\overleftarrow{e_2}}$}
    node[below,midway]{$\bA_{\overrightarrow{e_2}}^\top$} node[above,midway]{$\overleftarrow{e_2}$} 
    node[below,very near end]{$\bx^t_{\overleftarrow{e_2}}$}
    (v1);
	\end{tikzpicture}
\end{center}
with the corresponding iteration:
	\begin{align}
	\begin{split}
	\label{eq:mlamp}
    \bx^{t+1}_{\overrightarrow{e_1}} &= \bA_{\overrightarrow{e_1}} \bm^t_{\overrightarrow{e_1}} -  \bm^{t-1}_{\overleftarrow{e_1}}\left(\mathbf{b}^t_{\overrightarrow{e_1}}\right)^{\top} \, ,  \\
    &\bm^t_{\overrightarrow{e_1}} = f^t_{\overrightarrow{e}_1}\left(\bx^t_{\overleftarrow{e_1}}\right) \, , \\
        \bx^{t+1}_{\overleftarrow{e_1}} &= \bA_{\overrightarrow{e_1}}^\top\bm^t_{\overleftarrow{e_1}} -  \bm^{t-1}_{\overrightarrow{e_1}}\left( \mathbf{b}^t_{\overleftarrow{e_1}}\right)^{\top} \, ,  \\
    &\bm^t_{\overleftarrow{e_1}} = f^t_{\overleftarrow{e_1}}\left(\mathbf{A}_{\overrightarrow{e}_{1}}\mathbf{w}_{\overrightarrow{e}_{1}},\bx^t_{\overrightarrow{e_1}},\bx^t_{\overleftarrow{e_2}}\right) \, , \\
    &\qquad \\
        \bx^{t+1}_{\overrightarrow{e_2}} &= \bA_{\overrightarrow{e_2}} \bm^t_{\overrightarrow{e_2}} - \bm^{t-1}_{\overleftarrow{e_2}}\left(\mathbf{b}^t_{\overrightarrow{e_2}}\right)^{\top} \, ,  \\
    &\bm^t_{\overrightarrow{e_2}} = f^t_{\overrightarrow{e}_2}\left(\bx^t_{\overrightarrow{e_1}},\bx^t_{\overleftarrow{e_2}}\right) \, , \\
        \bx^{t+1}_{\overleftarrow{e_2}} &= \bA_{\overrightarrow{e_2}}^\top\bm^t_{\overleftarrow{e_2}} - \bm^{t-1}_{\overrightarrow{e_2}}(\mathbf{b}^t_{\overleftarrow{e_2}})^{\top} \, ,  \\
    &\bm^t_{\overleftarrow{e_2}} = f^t_{\overleftarrow{e_2}}\left(\mathbf{A}_{\overrightarrow{e}_{2}}\mathbf{w}_{\overrightarrow{e}_{2}},\bx^t_{\overrightarrow{e_2}},\bx^t_{\overleftarrow{e_3}}\right) \, , \\
    &\qquad\\
    &\qquad\quad\vdots \\
     &\qquad \\
        \bx^{t+1}_{\overrightarrow{e_L}} &= \bA_{\overrightarrow{e_L}} \bm^t_{\overrightarrow{e_L}} - \bm^{t-1}_{\overleftarrow{e_L}}\left(\mathbf{b}^t_{\overrightarrow{e_L}}\right)^{\top} \, ,  \\
    &\bm^t_{\overrightarrow{e_L}} = f^t_{\overrightarrow{e}_L}\left(\bx^t_{\overrightarrow{e}_{L-1}},\bx^t_{\overleftarrow{e_L}}\right) \, , \\
        \bx^{t+1}_{\overleftarrow{e_L}} &= \bA_{\overrightarrow{e_L}}^\top\bm^t_{\overleftarrow{e_L}} - \bm^{t-1}_{\overrightarrow{e_L}}(\mathbf{b}^t_{\overleftarrow{e_L}})^{\top} \, ,  \\
    &\bm^t_{\overleftarrow{e_L}} = f^t_{\overleftarrow{e_L}}\left(\mathbf{A}_{\overrightarrow{e}_{L}}\mathbf{w}_{\overrightarrow{e}_{L}},\bx^t_{\overrightarrow{e_L}}\right) \,
    \end{split}
\end{align}
and Onsager terms, for the right oriented edges 
    \begin{equation*} 
    \mathbf{b}^t_{\overrightarrow{e}_{l}} = \frac{1}{N} \sum_{i=1}^{n_{l-1}} \frac{\partial f^t_{\overrightarrow{e}_{l},i}}{\partial \mathbf{x}_{\overleftarrow{e}_{l},i}} \left(\left(\mathbf{x}^t_{\overrightarrow{e}_{l}'}\right)_{\overrightarrow{e}_{l}':\overrightarrow{e}_{l}' \to \overrightarrow{e}_{l}}\right) \qquad \in \mathbb{R}^{q \times q} \, .
\end{equation*}
and left oriented edges
    \begin{equation*} 
    \mathbf{b}^t_{\overleftarrow{e}_{l}} = \frac{1}{N} \sum_{i=1}^{n_l} \frac{\partial f^t_{\overleftarrow{e}_{l},i}}{\partial \mathbf{x}_{\overrightarrow{e}_{l},i}} \left(\mathbf{A}_{\overrightarrow{e}_{l}}\mathbf{w}_{\overrightarrow{e}_{l}},\left(\mathbf{x}^t_{\overleftarrow{e}_{l}'}\right)_{\overleftarrow{e}_{l}':\overleftarrow{e}_{l}' \to \overleftarrow{e}_{l}}\right) \qquad \in \mathbb{R}^{q \times q} \, .
\end{equation*}
We now make the following assumptions
\begin{enumerate}[font={\bfseries},label={(A\arabic*)}]
\item \label{ass:main1} The matrices $(\bA_{\overrightarrow{e}})_{\overrightarrow{e} \in \overrightarrow{E}}$ are random and independent, up to the symmetry condition $\bA_{\overleftarrow{e}} = \bA_{\overrightarrow{e}}^\top$. Moreover $\bA_{\overrightarrow{e}}$ has independent centered Gaussian entries with variance $1/N$.
\item For all $1 \leqslant l \leqslant L$, $n_l \to \infty$ and $n_l/N$ converges to a well-defined limit $\delta_l \in [0,1]$. We denote by $n \to \infty$ the limit under this scaling.
\item For all $t \in \mathbb{N}$ and $\overrightarrow{e} \in \overrightarrow{E}$, the non-linearity $f^{t}_{\overrightarrow{e}}$ is pseudo-Lipschitz of finite order, uniformly with respect to the problem dimensions $(n_{l})_{0 \leqslant l \leqslant L}$
\item For all $\overrightarrow{e}\in E$, the lines of  $\mathbf{x}^{0}_{\overrightarrow{e}}, \mathbf{w}_{\overrightarrow{e}}$ are sampled from subgaussian probability distributions in $\mathbf{R}^{q}$.
\item For all $\overrightarrow{e} \in E$, the following limit exists and is finite:
\begin{equation*}
    \lim_{n\to\infty} \frac{1}{N}  \left\langle f^0_{\overrightarrow{e}} \left(\left(\bx^0_{\overrightarrow{e}'}\right)_{ {\overrightarrow{e}'}:{\overrightarrow{e}'} \to {\overrightarrow{e}}}\right), f^0_{\overrightarrow{e}} \left(\left(\bx^0_{\overrightarrow{e}'}\right)_{ {\overrightarrow{e}'}:{\overrightarrow{e}'} \to {\overrightarrow{e}}}\right) \right\rangle
\end{equation*}
\item Let $(\kappa_{\overrightarrow{e}})_{\overrightarrow{e}\in E}$ be an array of bounded non-negative reals and $\bZ_{\overrightarrow{e}} \sim \mathbf{N}(0,\kappa_{\overrightarrow{e}}\mathbf{I}_{n_{w}})$ independent random variables for all $\overrightarrow{e}$. For all $\overrightarrow{e}\in E$, for any $t \in \mathbb{N}_{>0}$, the following limit exists and is finite:
\begin{equation*}
    \lim_{n \to \infty} \frac{1}{N}  \mathbb{E}\left[\left \langle f^0_{\overrightarrow{e}}  \left(\left(\bx^0_{\overrightarrow{e}'}\right)_{ {\overrightarrow{e}'}:{\overrightarrow{e}'} \to {\overrightarrow{e}}}\right), f^{t}_{\overrightarrow{e}}\left(\left(\bZ^{t}_{\overrightarrow{e}'}\right)_{ {\overrightarrow{e}'}:{\overrightarrow{e}'} \to {\overrightarrow{e}}}\right)\right \rangle\right].
\end{equation*}

\item \label{ass:main7} Consider any array of $2\times 2$ positive definite matrices $(\boldsymbol{S}_{\overrightarrow{e}})_{\overrightarrow{e}\in E}$ and the collection of random variables $(\bZ_{\overrightarrow{e}},\bZ^{'}_{\overrightarrow{e}}) \sim \mathbf{N}(0,\boldsymbol{S}_{\overrightarrow{e}}\otimes\mathbf{I}_{n_{w}}))$ defined independently for each edge $\overrightarrow{e}$. Then for any $\overrightarrow{e}\in E$ and $s,t >0$, the following limit exists and is finite:
\begin{equation*}
    \lim_{n \to \infty} \frac{1}{N}  \mathbb{E}\left[\left \langle f^{s}_{\overrightarrow{e}}\left(\left(\bZ^{s}_{\overrightarrow{e}'}\right)_{ {\overrightarrow{e}'}:{\overrightarrow{e}'} \to {\overrightarrow{e}}}\right), f^{t}_{\overrightarrow{e}}\left(\left(\tilde{\bZ}^{t}_{\overrightarrow{e}'}\right)_{ {\overrightarrow{e}'}:{\overrightarrow{e}'} \to {\overrightarrow{e}}}\right)\right \rangle \right].
\end{equation*}
\end{enumerate}

Under these assumptions, we define the following state evolution recursion: \\
\begin{itemize}
\item for $l=1$ :
\begin{align}
    &\boldsymbol{\nu}_{\overrightarrow{e}_{1}}^{0} = \lim_{N \to \infty} \frac{1}{N}\mathbf{w}_{\overrightarrow{e}_{1}}^{\top}f^{0}_{\overrightarrow{e}_{1}}(\mathbf{x}^{0}_{\overleftarrow{e}_{1}}), \thickspace \boldsymbol{\kappa}^{1,1}_{\overrightarrow{e}_{1}} = \lim_{N \to \infty} \frac{1}{N}f^{0}_{\overrightarrow{e}_{1}}(\mathbf{x}^{0}_{\overleftarrow{e}_{1}})^{\top}f^{0}_{\overrightarrow{e}_{1}}(\mathbf{x}^{0}_{\overleftarrow{e}_{1}}) \\
    &\boldsymbol{\nu}^{t+1}_{\overrightarrow{e}_{1}} = \lim_{N \to +\infty} \frac{1}{N}\mathbb{E}\left[\mathbf{w}_{\overrightarrow{e}_{1}}^{\top}f^{t}_{\overrightarrow{e}_{1}}\left(\mathbf{w}_{\overrightarrow{e}_{1}}\hat{\boldsymbol{\nu}}^{t}_{\overleftarrow{e}_{1}}+\mathbf{Z}^{t}_{\overleftarrow{e}_{1}}\right)\right] \\
    &\boldsymbol{\kappa}_{\overrightarrow{e}_{1}}^{s+1,t+1} = \boldsymbol{\kappa}_{\overrightarrow{e}_{1}}^{t+1,s+1} = \lim_{N \to +\infty} \frac{1}{N}\mathbb{E}\bigg[\left(f^{s}_{\overrightarrow{e}_{1}}\left(\mathbf{w}_{\overrightarrow{e}_{1}}\hat{\boldsymbol{\nu}}^{s}_{\overleftarrow{e}_{1}}+\mathbf{Z}^{s}_{\overleftarrow{e}_{1}}\right)-\mathbf{w}_{\overrightarrow{e}_{1}}\rho_{\mathbf{w}_{\overrightarrow{e}_{1}}}^{-1}\boldsymbol{\nu}_{\overrightarrow{e}_{1}}^{s+1}\right)^{\top} \notag \\
    &\left(f^{t}_{\overrightarrow{e}_{1}}\left(\mathbf{w}_{\overrightarrow{e}_{1}}\hat{\boldsymbol{\nu}}^{t}_{\overleftarrow{e}_{1}}+\mathbf{Z}^{t}_{\overleftarrow{e}_{1}}\right)-\mathbf{w}_{\overrightarrow{e}_{1}}\rho_{\mathbf{w}_{\overrightarrow{e}_{1}}}^{-1}\boldsymbol{\nu}_{\overrightarrow{e}_{1}}^{t+1}\right)\bigg] \\
    &\hat{\boldsymbol{\nu}}_{\overleftarrow{e}_{1}}^{0}, \boldsymbol{\kappa}_{\overleftarrow{e}_{1}}^{1,1} = \lim_{n \to \infty} \frac{1}{N}
f^{0}_{\overleftarrow{e}_{1}}\left(\mathbf{z}_{\mathbf{w}_{\overrightarrow{e}_{1}}},\mathbf{x}^{0}_{\overrightarrow{e}_{1}},\mathbf{x}^{0}_{\overleftarrow{e}_{2}}\right)^{\top}f^{0}_{\overleftarrow{e}_{1}}\left(\mathbf{z}_{\mathbf{w}_{\overrightarrow{e}_{1}}},\mathbf{x}^{0}_{\overrightarrow{e}_{1}},\mathbf{x}^{0}_{\overleftarrow{e}_{2}}\right) \\
&\hat{\boldsymbol{\nu}}^{t+1}_{\overleftarrow{e}_{1}} = \lim_{N \to \infty} \frac{1}{N}\mathbb{E}\bigg[\sum_{i=1}^{N}\frac{\partial f_{\overleftarrow{e}_{1},i}^{t}}{\partial \mathbf{z}_{\mathbf{w}_{\overleftarrow{e}_{1}},i},\varphi_{\overleftarrow{e}_{1}}}\left(\mathbf{z}_{\mathbf{w}_{\overrightarrow{e}_{1}}},\mathbf{z}_{\mathbf{w}_{\overrightarrow{e}_{1}}}\rho^{-1}_{\mathbf{w}_{\overrightarrow{e}_{1}}}\boldsymbol{\nu}_{\overrightarrow{e}_{1}}^{t}+\mathbf{Z}^{t}_{\overrightarrow{e}_{1}},\mathbf{w}_{\overrightarrow{e}_{2}}\hat{\boldsymbol{\nu}}_{\overleftarrow{e}_{2}}^{t}+\mathbf{Z}^{t}_{\overleftarrow{e}_{2}}\right)\bigg] \\
&\boldsymbol{\kappa}_{\overleftarrow{e}_{1}}^{s+1,t+1} = \lim_{n \to \infty} \frac{1}{N}
\mathbb{E}\bigg [f^{s}_{\overleftarrow{e}_{1}}\left(\mathbf{z}_{\mathbf{w}_{\overrightarrow{e}_{1}}},\mathbf{z}_{\mathbf{w}_{\overrightarrow{e}_{1}}}\rho^{-1}_{\mathbf{w}_{\overrightarrow{e}_{1}}}\boldsymbol{\nu}_{\overrightarrow{e}_{1}}^{s}+\mathbf{Z}^{s}_{\overrightarrow{e}_{1}},\mathbf{w}_{\overrightarrow{e}_{2}}\hat{\boldsymbol{\nu}}_{\overleftarrow{e}_{2}}^{s}+\mathbf{Z}^{s}_{\overleftarrow{e}_{2}}\right)^{\top} \notag \\
&f^{t}_{\overleftarrow{e}_{1}}\left(\mathbf{z}_{\mathbf{w}_{\overrightarrow{e}_{1}}},\mathbf{z}_{\mathbf{w}_{\overrightarrow{e}_{1}}}\rho^{-1}_{\mathbf{w}_{\overrightarrow{e}_{1}}}\boldsymbol{\nu}_{\overrightarrow{e}_{1}}^{t}+\mathbf{Z}^{t}_{\overrightarrow{e}_{1}},\mathbf{w}_{\overrightarrow{e}_{2}}\hat{\boldsymbol{\nu}}_{\overleftarrow{e}_{2}}^{t}+\mathbf{Z}^{t}_{\overleftarrow{e}_{2}}\right)\bigg ]\end{align}
\item for any $2\leqslant l \leqslant L-1$
\begin{align}
    &\boldsymbol{\nu}_{\overrightarrow{e}_{l}}^{0} = \lim_{N \to \infty} \frac{1}{N}\mathbf{w}_{\overrightarrow{e}_{l}}^{\top}f^{0}_{\overrightarrow{e}_{l}}(\mathbf{x}^{0}_{\overleftarrow{e}_{l}}), \thickspace \boldsymbol{\kappa}^{1,1}_{\overrightarrow{e}_{l}} = \lim_{N \to \infty} \frac{1}{N}f^{0}_{\overrightarrow{e}_{l}}(\mathbf{x}^{0}_{\overleftarrow{e}_{l}})^{\top}f^{0}_{\overrightarrow{e}_{l}}(\mathbf{x}^{0}_{\overleftarrow{e}_{l}}) \\
    &\boldsymbol{\nu}^{t+1}_{\overrightarrow{e}_{l}} = \lim_{N \to +\infty} \frac{1}{N}\mathbb{E}\left[\mathbf{w}_{\overrightarrow{e}_{l}}^{\top}f^{t}_{\overrightarrow{e}_{l}}\left(\mathbf{z}_{\mathbf{w}_{\overrightarrow{e}_{l-1}}}\rho^{-1}_{\mathbf{w}_{\overrightarrow{e}_{l-1}}}\boldsymbol{\nu}_{\overrightarrow{e}_{l-1}}^{t}+\mathbf{Z}^{t}_{\overrightarrow{e}_{l-1}}, \mathbf{w}_{\overrightarrow{e}_{l}}\hat{\boldsymbol{\nu}}^{t}_{\overleftarrow{e}_{l}}+\mathbf{Z}^{t}_{\overleftarrow{e}_{l}}\right)\right] \\
    &\boldsymbol{\kappa}_{\overrightarrow{e}_{l}}^{s+1,t+1} = \boldsymbol{\kappa}_{\overrightarrow{e}_{l}}^{t+1,s+1} = \lim_{N \to +\infty} \\
    &\frac{1}{N}\mathbb{E}\bigg[\left(f^{s}_{\overrightarrow{e}_{l}}\left(\mathbf{z}_{\mathbf{w}_{\overrightarrow{e}_{l-1}}}\rho^{-1}_{\mathbf{w}_{\overrightarrow{e}_{l-1}}}\boldsymbol{\nu}_{\overrightarrow{e}_{l-1}}^{s}+\mathbf{Z}^{s}_{\overrightarrow{e}_{l-1}},\mathbf{w}_{\overrightarrow{e}_{l}}\hat{\boldsymbol{\nu}}^{s}_{\overleftarrow{e}_{l}}+\mathbf{Z}^{s}_{\overleftarrow{e}_{l}}\right)-\mathbf{w}_{\overrightarrow{e}_{l}}\rho_{\mathbf{w}_{\overrightarrow{e}_{l}}}^{-1}\boldsymbol{\nu}_{\overrightarrow{e}_{l}}^{s+1}\right)^{\top} \notag  \\
    &\left(f^{t}_{\overrightarrow{e}_{l}}\left(\mathbf{z}_{\mathbf{w}_{\overrightarrow{e}_{l-1}}}\rho^{-1}_{\mathbf{w}_{\overrightarrow{e}_{l-1}}}\boldsymbol{\nu}_{\overrightarrow{e}_{l-1}}^{t}+\mathbf{Z}^{t}_{\overrightarrow{e}_{l-1}},\mathbf{w}_{\overrightarrow{e}_{l}}\hat{\boldsymbol{\nu}}^{t}_{\overleftarrow{e}_{l}}+\mathbf{Z}^{t}_{\overleftarrow{e}_{l}}\right)-\mathbf{w}_{\overrightarrow{e}_{l}}\rho_{\mathbf{w}_{\overrightarrow{e}_{l}}}^{-1}\boldsymbol{\nu}_{\overrightarrow{e}_{l}}^{t+1}\right)\bigg] \\
    &\hat{\boldsymbol{\nu}}_{\overleftarrow{e}_{l}}^{0}, \boldsymbol{\kappa}_{\overleftarrow{e}_{l}}^{1,1} = \lim_{n \to \infty} \frac{1}{N}
f^{0}_{\overleftarrow{e}_{l}}\left(\mathbf{z}_{\mathbf{w}_{\overrightarrow{e}_{l}}},\mathbf{x}^{0}_{\overrightarrow{e}_{l}},\mathbf{x}^{0}_{\overleftarrow{e}_{l+1}}\right)^{\top}f^{0}_{\overleftarrow{e}_{l}}\left(\mathbf{z}_{\mathbf{w}_{\overrightarrow{e}_{l}}},\mathbf{x}^{0}_{\overrightarrow{e}_{l}},\mathbf{x}^{0}_{\overleftarrow{e}_{l+1}}\right) \\
&\hat{\boldsymbol{\nu}}^{t+1}_{\overleftarrow{e}_{l}} = \lim_{N \to \infty} \frac{1}{N}\mathbb{E}\bigg[\sum_{i=1}^{N}\frac{\partial f_{\overleftarrow{e}_{l},i}^{t}}{\partial \mathbf{z}_{\mathbf{w}_{\overleftarrow{e}_{l}},i},\varphi_{\overleftarrow{e}_{l}}}\left(\mathbf{z}_{\mathbf{w}_{\overrightarrow{e}_{l}}},\mathbf{z}_{\mathbf{w}_{\overrightarrow{e}_{l}}}\rho^{-1}_{\mathbf{w}_{\overrightarrow{e}_{l}}}\boldsymbol{\nu}_{\overrightarrow{e}_{l}}^{t}+\mathbf{Z}^{t}_{\overrightarrow{e}_{l}},\mathbf{w}_{\overrightarrow{e}_{l+1}}\hat{\boldsymbol{\nu}}_{\overrightarrow{e}_{l+1}}^{t}\mathbf{Z}^{t}_{\overleftarrow{e}_{l+1}}\right)\bigg] \\
&\boldsymbol{\kappa}_{\overleftarrow{e}_{l}}^{s+1,t+1} = \lim_{n \to \infty} \frac{1}{N}
\mathbb{E}\bigg [f^{s}_{\overleftarrow{e}_{l}}\left(\mathbf{z}_{\mathbf{w}_{\overrightarrow{e}_{l}}},\mathbf{z}_{\mathbf{w}_{\overrightarrow{e}_{l}}}\rho^{-1}_{\mathbf{w}_{\overrightarrow{e}_{l}}}\boldsymbol{\nu}_{\overrightarrow{e}_{l}}^{s}+\mathbf{Z}^{s}_{\overrightarrow{e}_{l}},\mathbf{w}_{\overrightarrow{e}_{l+1}}\hat{\boldsymbol{\nu}}_{\overrightarrow{e}_{l+1}}^{s}\mathbf{Z}^{s}_{\overleftarrow{e}_{l+1}}\right)^{\top} \notag \\
&f^{t}_{\overleftarrow{e}_{l}}\left(\mathbf{z}_{\mathbf{w}_{\overrightarrow{e}_{l}}},\mathbf{z}_{\mathbf{w}_{\overrightarrow{e}_{l}}}\rho^{-1}_{\mathbf{w}_{\overrightarrow{e}_{l}}}\boldsymbol{\nu}_{\overrightarrow{e}_{l}}^{t}+\mathbf{Z}^{t}_{\overrightarrow{e}_{l}},\mathbf{w}_{\overrightarrow{e}_{l+1}}\hat{\boldsymbol{\nu}}_{\overrightarrow{e}_{l+1}}^{t}\mathbf{Z}^{t}_{\overleftarrow{e}_{l+1}}\right)\bigg ]\end{align}
\item for l=L
\begin{align}
    &\boldsymbol{\nu}_{\overrightarrow{e}_{L}}^{0} = \lim_{N \to \infty} \frac{1}{N}\mathbf{w}_{\overrightarrow{e}_{l}}^{\top}f^{0}_{\overrightarrow{e}_{L}}(\mathbf{x}^{0}_{\overleftarrow{e}_{L}}), \thickspace \boldsymbol{\kappa}^{1,1}_{\overrightarrow{e}_{L}} = \lim_{N \to \infty} \frac{1}{N}f^{0}_{\overrightarrow{e}_{L}}(\mathbf{x}^{0}_{\overleftarrow{e}_{L}})^{\top}f^{0}_{\overrightarrow{e}_{L}}(\mathbf{x}^{0}_{\overleftarrow{e}_{L}}) \\
    &\boldsymbol{\nu}^{t+1}_{\overrightarrow{e}_{L}} = \lim_{N \to +\infty} \frac{1}{N}\mathbb{E}\left[\mathbf{w}_{\overrightarrow{e}_{L}}^{\top}f^{t}_{\overrightarrow{e}_{L}}\left(\mathbf{z}_{\mathbf{w}_{\overrightarrow{e}_{L-1}}}\rho^{-1}_{\mathbf{w}_{\overrightarrow{e}_{L-1}}}\boldsymbol{\nu}_{\overrightarrow{e}_{L-1}}^{t}+\mathbf{Z}^{t}_{\overrightarrow{e}_{L-1}}, \mathbf{w}_{\overrightarrow{e}_{L}}\hat{\boldsymbol{\nu}}^{t}_{\overleftarrow{e}_{L}}+\mathbf{Z}^{t}_{\overleftarrow{e}_{L}}\right)\right] \\
    &\boldsymbol{\kappa}_{\overrightarrow{e}_{L}}^{s+1,t+1} = \boldsymbol{\kappa}_{\overrightarrow{e}_{L}}^{t+1,s+1} = \lim_{N \to +\infty} \\
    &\frac{1}{N}\mathbb{E}\bigg[\left(f^{s}_{\overrightarrow{e}_{L}}\left(\mathbf{z}_{\mathbf{w}_{\overrightarrow{e}_{L-1}}}\rho^{-1}_{\mathbf{w}_{\overrightarrow{e}_{L-1}}}\boldsymbol{\nu}_{\overrightarrow{e}_{L-1}}^{s}+\mathbf{Z}^{s}_{\overrightarrow{e}_{L-1}},\mathbf{w}_{\overrightarrow{e}_{L}}\hat{\boldsymbol{\nu}}^{s}_{\overleftarrow{e}_{L}}+\mathbf{Z}^{s}_{\overleftarrow{e}_{L}}\right)-\mathbf{w}_{\overrightarrow{e}_{L}}\rho_{\mathbf{w}_{\overrightarrow{e}_{L}}}^{-1}\boldsymbol{\nu}_{\overrightarrow{e}_{L}}^{s+1}\right)^{\top} \notag \\
    &\left(f^{t}_{\overrightarrow{e}_{L}}\left(\mathbf{z}_{\mathbf{w}_{\overrightarrow{e}_{L-1}}}\rho^{-1}_{\mathbf{w}_{\overrightarrow{e}_{L-1}}}\boldsymbol{\nu}_{\overrightarrow{e}_{L-1}}^{t}+\mathbf{Z}^{t}_{\overrightarrow{e}_{L-1}},\mathbf{w}_{\overrightarrow{e}_{L}}\hat{\boldsymbol{\nu}}^{t}_{\overleftarrow{e}_{L}}+\mathbf{Z}^{t}_{\overleftarrow{e}_{L}}\right)-\mathbf{w}_{\overrightarrow{e}_{L}}\rho_{\mathbf{w}_{\overrightarrow{e}_{L}}}^{-1}\boldsymbol{\nu}_{\overrightarrow{e}_{L}}^{t+1}\right)\bigg] \\
    &\hat{\boldsymbol{\nu}}_{\overleftarrow{e}_{L}}^{0}, \boldsymbol{\kappa}_{\overleftarrow{e}_{L}}^{1,1} = \lim_{n \to \infty} \frac{1}{N}
f^{0}_{\overleftarrow{e}_{L}}\left(\mathbf{z}_{\mathbf{w}_{\overrightarrow{e}_{L}}},\mathbf{x}^{0}_{\overrightarrow{e}_{L}}\right)^{\top}f^{0}_{\overleftarrow{e}_{L}}\left(\mathbf{z}_{\mathbf{w}_{\overrightarrow{e}_{L}}}\right) \\
&\hat{\boldsymbol{\nu}}^{t+1}_{\overleftarrow{e}_{L}} = \lim_{N \to \infty} \frac{1}{N}\mathbb{E}\bigg[\sum_{i=1}^{N}\frac{\partial f_{\overleftarrow{e}_{L},i}^{t}}{\partial \mathbf{z}_{\mathbf{w}_{\overleftarrow{e}_{L}},i},\varphi_{\overleftarrow{e}_{L}}}\left(\mathbf{z}_{\mathbf{w}_{\overrightarrow{e}_{L}}},\mathbf{z}_{\mathbf{w}_{\overrightarrow{e}_{L}}}\rho^{-1}_{\mathbf{w}_{\overrightarrow{e}_{L}}}\boldsymbol{\nu}_{\overrightarrow{e}_{L}}^{t}+\mathbf{Z}^{t}_{\overrightarrow{e}_{L}}\right)\bigg] \\
&\boldsymbol{\kappa}_{\overleftarrow{e}_{L}}^{s+1,t+1} = \lim_{n \to \infty} \frac{1}{N}
\mathbb{E}\bigg [f^{s}_{\overleftarrow{e}_{L}}\left(\mathbf{z}_{\mathbf{w}_{\overrightarrow{e}_{L}}},\mathbf{z}_{\mathbf{w}_{\overrightarrow{e}_{L}}}\rho^{-1}_{\mathbf{w}_{\overrightarrow{e}_{L}}}\boldsymbol{\nu}_{\overrightarrow{e}_{L}}^{s}+\mathbf{Z}^{s}_{\overrightarrow{e}_{L}}\right)^{\top} \notag \\
&f^{t}_{\overleftarrow{e}_{L}}\left(\mathbf{z}_{\mathbf{w}_{\overrightarrow{e}_{L}}},\mathbf{z}_{\mathbf{w}_{\overrightarrow{e}_{L}}}\rho^{-1}_{\mathbf{w}_{\overrightarrow{e}_{L}}}\boldsymbol{\nu}_{\overrightarrow{e}_{L}}^{t}+\mathbf{Z}^{t}_{\overrightarrow{e}_{L}}\right)\bigg ]\end{align}
\end{itemize}
where, for any $1 \leqslant l \leqslant L$, the symbol $\partial \mathbf{z}_{\mathbf{w}_{\overrightarrow{e}},i},\varphi_{\overrightarrow{e}}$ denotes the partial derivative w.r.t. the argument of $\varphi_{\overrightarrow{e}}$, $(\bZ^1_{\overrightarrow{e}}, \dots, \bZ^t_{\overrightarrow{e}})$ is a centered Gaussian random vector with covariance $(\boldsymbol{\kappa}^{r,s}_{\overrightarrow{e}})_{r,s \leq t} \otimes \mathbf{I}_{n_{w}}$ (and similarly for left-oriented edges), and 
$z_{\mathbf{w}_{\overrightarrow{e}}}$ is distributed according to $\mathbf{N}(0,\boldsymbol{\rho}_{\mathbf{w}_{\overrightarrow{e}}})$.
\begin{theorem}
\label{thm:graph-AMP}
Assume \ref{ass:main1}-\ref{ass:main7}. Define, as above, independently for each $\overrightarrow{e}_{l}$, $\bZ^0_{\overrightarrow{e}_{l}}= \bx^0_{\overrightarrow{e}_{l}}$ and $(\bZ^1_{\overrightarrow{e}_{l}}, \dots, \bZ^t_{\overrightarrow{e}_{l}})$ a centered Gaussian random vector of covariance $(\boldsymbol{\kappa}^{r,s}_{\overrightarrow{e}_{l}})_{r,s \leq t} \otimes \mathbf{I}_{n_{l-1}}$. Then for any sequence of uniformly (in $n$) pseudo-Lipschitz function $\Phi:(\R^{n_{l-1} \times (t+1)q})^{2} \to \R $,  for any $1\leqslant l \leqslant L$
\begin{align*}
    &\Phi\left(\left(\bx^s_{\overrightarrow{e}_{l}}\right)_{0 \leq s \leq t}, \left(\bx^s_{\overleftarrow{e}_{l-1}}\right)_{0 \leq s \leq t}\right) \approxP \\
   & \hspace{2cm}\E \bigg[ \Phi\bigg(\left(\mathbf{z}_{\mathbf{w}_{\overrightarrow{e}_{l}}}\rho_{\mathbf{w}_{\overrightarrow{e}_{l}}}^{-1}\boldsymbol{\nu}_{\overrightarrow{e}_{l}}^{s}+\bZ^s_{\overrightarrow{e}_{l-1}}\right)_{0 \leq s \leq t}, \left(\mathbf{w}_{\overrightarrow{e}_{l-1}}\hat{\boldsymbol{\nu}}^{s}_{\overleftarrow{e}_{l-1}}+\bZ^s_{\overleftarrow{e}_{l-1}}\right)_{0 \leq s \leq t} \bigg]
\end{align*}
\end{theorem}
In summary, at each time step, the variables associated with right oriented edges $\mathbf{x}_{\overrightarrow{e}_{l}}$ asymptotically behave as the sum of the ground truth $\mathbf{w}_{\overrightarrow{e}_{l}}$ reweighted by a $q \times q$ matrix coefficient $\hat{\boldsymbol{\nu}}_{\overleftarrow{e}_{l}}$ and a $n_{l-1} \times q $ random matrix with i.i.d. lines $\mathbf{Z}_{\overrightarrow{e}_{l}}$ with $q \times q$ covariance $\boldsymbol{\kappa}_{\overleftarrow{e}_{l}}$ determined by the function associated to the corresponding left-oriented arrow $f^{t}_{\overleftarrow{e}_{l}}$. Similarly, the variables associated with left oriented edges $\mathbf{x}_{\overleftarrow{e}_{l}}$ asymptotically behave as the sum of the linear response to the ground truth $\mathbf{z}_{\mathbf{w}_{\overrightarrow{e}_{l}}}$ (asymptotic equivalent of $\mathbf{A}_{\overrightarrow{e}_{l}}\mathbf{w}_{\overrightarrow{e}_{l}}$) reweighted by a $q \times q$ matrix coefficient $\boldsymbol{\nu}_{\overleftarrow{e}_{l}}$ and a $n_{l} \times q $ random matrix with i.i.d. lines $\mathbf{Z}_{\overleftarrow{e}_{l}}$ with $q \times q$ covariance $\boldsymbol{\kappa}_{\overrightarrow{e}_{l}}$ determined by the function associated to the corresponding right-oriented arrow $f^{t}_{\overrightarrow{e}_{l}}$.
\begin{proof}
This result is a special case of Lemma 2 from \cite{gerbelot2021graph}, with a perturbation where only the left-oriented edges involve an additional dependence on $\mathbf{A}_{\overrightarrow{e}}\mathbf{w}_{\overrightarrow{e}}$. The required conditions are the same as in \cite{gerbelot2021graph}, barring the subgaussian assumption (A3) which ensures the scaled norm of the $\mathbf{x}_{\overrightarrow{e}}^{0}, \mathbf{w}_{\overrightarrow{e}}$ are finite with high-probability as $n \to \infty$.
\end{proof}
\subsection{State evolution for multilayer AMP iterations with random convolutional matrices}
\label{app_sec_sep}
The following lemma proves the state evolution equations for a multilayer AMP iteration where the dense Gaussian matrices are replaced with random convolutional ones (MCC from Def.\ref{dfn:mcc}) with variance $\frac{1}{N}$, 
with a vector valued variables, i.e. q=1, and separables non-linearities. We choose the variance as $\frac{1}{N}$ to follow the notations of \cite{gerbelot2021graph} for more convenience, recovering the variances of iteration Eq.\eqref{eqn:ml-amp-iterate} is a straightforward rescaling as done in \cite{berthier2020state} and will be discussed in the next section.
    Assume $q=1$ and that, for any $t \in \mathbb{N}$ and $1 \leqslant l \leqslant L$, the functions $f^{t}_{\overrightarrow{e}_{l}}, f^{t}_{\overleftarrow{e}_{l}}$ are separable in all their arguments, i.e there exists scalar valued, pseudo-Lipschitz functions $\sigma^{t}_{\overrightarrow{e}_{l}} : \mathbb{R}^{2} \to \mathbb{R},\sigma^{t}_{\overleftarrow{e}_{l}}:\mathbb{R}^{3} \to \mathbb{R}$ (where $\sigma^{t}_{\overrightarrow{e}_{1}} : \mathbb{R} \to \mathbb{R}, \sigma_{\overleftarrow{e}_{L}}^{t} : \mathbb{R}^{2} \to \mathbb{R}$) such that:
    \begin{align*}
    &\mbox{for $l=1$, for any $1 \leqslant i \leqslant n_{0}$}: \\
        &\hspace{1cm}f^{t}_{\overleftarrow{e}_{1}}(\mathbf{x}^{t}_{\overleftarrow{e}_{1}})_{i} = \sigma^{t}_{\overleftarrow{e}_{1}}(x^{t}_{\overleftarrow{e}_{1},i}) \\
        &\mbox{for any $1 \leqslant l \leqslant L-1$, for any $1\leqslant i \leqslant n_{l}$:} \\
        &\hspace{1cm}f^{t}_{\overleftarrow{e}_{l}}\left(\mathbf{A}_{\overrightarrow{e}_{l}}\mathbf{w}_{\overrightarrow{e}_{l}},\mathbf{x}^{t}_{\overrightarrow{e}_{l}}, \mathbf{x}^{t}_{\overleftarrow{e}_{l+1}}\right)_{i} = \sigma^{t}_{\overleftarrow{e}_{l}}\left((\mathbf{A}_{\overrightarrow{e}_{l}}\mathbf{w}_{\overrightarrow{e}_{l}})_{i},x^{t}_{\overrightarrow{e}_{l},i}, x^{t}_{\overleftarrow{e}_{l+1},i}\right) \\
        &\mbox{for any $2 \leqslant l \leqslant L$, $1 \leqslant i \leqslant n_{l-1}$:} \\
        &\hspace{1cm}f^{t}_{\overrightarrow{e}_{l}}\left(\mathbf{x}^{t}_{\overrightarrow{e}_{l-1}}, \mathbf{x}^{t}_{\overleftarrow{e}_{l}}\right)_{i} = \sigma^{t}_{\overrightarrow{e}_{l}}\left(x^{t}_{\overrightarrow{e}_{l-1},i}, x^{t}_{\overleftarrow{e}_{l},i}\right) \\
        &\mbox{for l=L, any $1 \leqslant i \leqslant n_{L}$:} \\
        &\hspace{1cm}f^{t}_{\overleftarrow{e}_{L}}(\mathbf{A}_{\overrightarrow{e}_{L}}\mathbf{w}_{\overleftarrow{e}_{L}},\mathbf{x}^{t}_{\overrightarrow{e}_{L}})_{i} = \sigma^{t}_{\overleftarrow{e}_{L}}((\mathbf{A}_{\overrightarrow{e}_{L}}\mathbf{w}_{\overleftarrow{e}_{L}})_{i},x^{t}_{\overrightarrow{e}_{L},i})
    \end{align*}
    Define the following scalar SE equations
    \begin{itemize}
    \label{eq:scalar_se}
        \item for $l=1$:
        \begin{align}
            &\nu_{\overrightarrow{e}_{1}}^{0} =  \delta_{0}\mathbb{E}\left[w_{\overrightarrow{e}_{1}}\sigma^{0}_{\overrightarrow{e}_{1}}(x^{0}_{\overleftarrow{e}_{1}})\right], \thickspace \kappa^{1,1}_{\overrightarrow{e}_{1}} = \delta_{0}\mathbb{E}\left[\sigma^{0}_{\overrightarrow{e}_{1}}(x^{0}_{\overleftarrow{e}_{1}})\sigma^{0}_{\overrightarrow{e}_{1}}(x^{0}_{\overleftarrow{e}_{1}})\right] \\
            &\nu^{t+1}_{\overrightarrow{e}_{1}} = \delta_{0}\mathbb{E}\left[w_{\overrightarrow{e}_{1}}\sigma^{t}_{\overrightarrow{e}_{1}}\left(w_{\overrightarrow{e}_{1}}\hat{\nu}^{t}_{\overleftarrow{e}_{1}}+Z^{t}_{\overleftarrow{e}_{1}}\right)\right] \\
            &\kappa_{\overrightarrow{e}_{1}}^{s+1,t+1} = \kappa_{\overrightarrow{e}_{1}}^{t+1,s+1} =\delta_{0}\mathbb{E}\bigg[\left(\sigma^{s}_{\overrightarrow{e}_{1}}\left(w_{\overrightarrow{e}_{1}}\hat{\nu}^{s}_{\overleftarrow{e}_{1}}+Z^{s}_{\overleftarrow{e}_{1}}\right)-w_{\overrightarrow{e}_{1}}\rho_{w_{\overrightarrow{e}_{1}}}^{-1}\nu_{\overrightarrow{e}_{1}}^{s+1}\right)\notag  \\
            &\left(\sigma^{t}_{\overrightarrow{e}_{1}}\left(w_{\overrightarrow{e}_{1}}\hat{\nu}^{t}_{\overleftarrow{e}_{1}}+Z^{t}_{\overleftarrow{e}_{1}}\right)-w_{\overrightarrow{e}_{1}}\rho_{w_{\overrightarrow{e}_{1}}}^{-1}\nu_{\overrightarrow{e}_{1}}^{t+1}\right)\bigg] \\
            &\hat{\nu}_{\overleftarrow{e}_{1}}^{0}, \kappa_{\overleftarrow{e}_{1}}^{1,1} = 
        \delta_{1}\mathbb{E}\bigg[\sigma^{0}_{\overleftarrow{e}_{1}}\left(z_{w_{\overrightarrow{e}_{1}}},x^{0}_{\overrightarrow{e}_{1}},x^{0}_{\overleftarrow{e}_{2}}\right)\sigma^{0}_{\overleftarrow{e}_{1}}\left(z_{w_{\overrightarrow{e}_{1}}},x^{0}_{\overrightarrow{e}_{1}},x^{0}_{\overleftarrow{e}_{2}}\right)\bigg] \\
        &\hat{\nu}^{t+1}_{\overleftarrow{e}_{1}} = \delta_{1}\mathbb{E}\bigg[\frac{\partial \sigma_{\overleftarrow{e}_{1},i}^{t}}{\partial z_{w_{\overleftarrow{e}_{1}},i},\varphi_{\overleftarrow{e}_{1}}}\left(z_{w_{\overrightarrow{e}_{1}}},z_{w_{\overrightarrow{e}_{1}}}\rho^{-1}_{w_{\overrightarrow{e}_{1}}}\nu_{\overrightarrow{e}_{1}}^{t}+Z^{t}_{\overrightarrow{e}_{1}},w_{\overrightarrow{e}_{2}}\hat{\nu}_{\overleftarrow{e}_{2}}^{t}+Z^{t}_{\overleftarrow{e}_{2}}\right)\bigg] \\
        &\kappa_{\overleftarrow{e}_{1}}^{s+1,t+1} = 
        \delta_{1}\mathbb{E}\bigg [\sigma^{s}_{\overleftarrow{e}_{1}}\left(z_{w_{\overrightarrow{e}_{1}}},z_{w_{\overrightarrow{e}_{1}}}\rho^{-1}_{w_{\overrightarrow{e}_{1}}}\nu_{\overrightarrow{e}_{1}}^{s}+Z^{s}_{\overrightarrow{e}_{1}},w_{\overrightarrow{e}_{2}}\hat{\nu}_{\overleftarrow{e}_{2}}^{s}+Z^{s}_{\overleftarrow{e}_{2}}\right) \notag \\
        &\sigma^{t}_{\overleftarrow{e}_{1}}\left(z_{w_{\overrightarrow{e}_{1}}},z_{w_{\overrightarrow{e}_{1}}}\rho^{-1}_{w_{\overrightarrow{e}_{1}}}\nu_{\overrightarrow{e}_{1}}^{t}+Z^{t}_{\overrightarrow{e}_{1}},w_{\overrightarrow{e}_{2}}\hat{\nu}_{\overleftarrow{e}_{2}}^{t}+Z^{t}_{\overleftarrow{e}_{2}}\right)\bigg ]\end{align}
        \item for any $2\leqslant l \leqslant L-1$
        \begin{align}
            &\nu_{\overrightarrow{e}_{l}}^{0} = \delta_{n_{l-1}}\mathbb{E}\bigg[ w_{\overrightarrow{e}_{l}}\sigma^{0}_{\overrightarrow{e}_{l}}(x^{0}_{\overleftarrow{e}_{l}})\bigg], \thickspace \kappa^{1,1}_{\overrightarrow{e}_{l}} = \delta_{n_{l-1}} \mathbb{E}\left[\sigma^{0}_{\overrightarrow{e}_{l}}(x^{0}_{\overleftarrow{e}_{l}})\sigma^{0}_{\overrightarrow{e}_{l}}(x^{0}_{\overleftarrow{e}_{l}})\right] \\
            &\nu^{t+1}_{\overrightarrow{e}_{l}} = \delta_{n_{l-1}}\mathbb{E}\left[w_{\overrightarrow{e}_{l}}\sigma^{t}_{\overrightarrow{e}_{l}}\left(z_{w_{\overrightarrow{e}_{l-1}}}\rho^{-1}_{w_{\overrightarrow{e}_{l-1}}}\nu_{\overrightarrow{e}_{l-1}}^{t}+Z^{t}_{\overrightarrow{e}_{l-1}}, w_{\overrightarrow{e}_{l}}\hat{\nu}^{t}_{\overleftarrow{e}_{l}}+Z^{t}_{\overleftarrow{e}_{l}}\right)\right] \\
            &\kappa_{\overrightarrow{e}_{l}}^{s+1,t+1} = \kappa_{\overrightarrow{e}_{l}}^{t+1,s+1} =\\
            &\delta_{n_{l-1}}\mathbb{E}\bigg[\left(\sigma^{s}_{\overrightarrow{e}_{l}}\left(z_{w_{\overrightarrow{e}_{l-1}}}\rho^{-1}_{w_{\overrightarrow{e}_{l-1}}}\nu_{\overrightarrow{e}_{l-1}}^{s}+Z^{s}_{\overrightarrow{e}_{l-1}},w_{\overrightarrow{e}_{l}}\hat{\nu}^{s}_{\overleftarrow{e}_{l}}+Z^{s}_{\overleftarrow{e}_{l}}\right)-w_{\overrightarrow{e}_{l}}\rho_{w_{\overrightarrow{e}_{l}}}^{-1}\nu_{\overrightarrow{e}_{l}}^{s+1}\right) \notag\\
            &\left(\sigma^{t}_{\overrightarrow{e}_{l}}\left(z_{w_{\overrightarrow{e}_{l-1}}}\rho^{-1}_{w_{\overrightarrow{e}_{l-1}}}\nu_{\overrightarrow{e}_{l-1}}^{t}+Z^{t}_{\overrightarrow{e}_{l-1}},w_{\overrightarrow{e}_{l}}\hat{\nu}^{t}_{\overleftarrow{e}_{l}}+Z^{t}_{\overleftarrow{e}_{l}}\right)-w_{\overrightarrow{e}_{l}}\rho_{w_{\overrightarrow{e}_{l}}}^{-1}\nu_{\overrightarrow{e}_{l}}^{t+1}\right)\bigg]  \\
            &\hat{\nu}_{\overleftarrow{e}_{l}}^{0}, \kappa_{\overleftarrow{e}_{l}}^{1,1} = 
        \delta_{n_{l}}\mathbb{E}\bigg[\sigma^{0}_{\overleftarrow{e}_{l}}\left(z_{w_{\overrightarrow{e}_{l}}},x^{0}_{\overrightarrow{e}_{l}},x^{0}_{\overleftarrow{e}_{l+1}}\right)\sigma^{0}_{\overleftarrow{e}_{l}}\left(z_{w_{\overrightarrow{e}_{l}}},x^{0}_{\overrightarrow{e}_{l}},x^{0}_{\overleftarrow{e}_{l+1}}\right)\bigg] \\
        &\hat{\nu}^{t+1}_{\overleftarrow{e}_{l}} = \delta_{n_{l}}\mathbb{E}\bigg[\frac{\partial \sigma_{\overleftarrow{e}_{l},i}^{t}}{\partial z_{w_{\overleftarrow{e}_{l}},i},\varphi_{\overleftarrow{e}_{l}}}\left(z_{w_{\overrightarrow{e}_{l}}},z_{w_{\overrightarrow{e}_{l}}}\rho^{-1}_{w_{\overrightarrow{e}_{l}}}\nu_{\overrightarrow{e}_{l}}^{t}+Z^{t}_{\overrightarrow{e}_{l}},w_{\overrightarrow{e}_{l+1}}\hat{\nu}_{\overrightarrow{e}_{l+1}}^{t}Z^{t}_{\overleftarrow{e}_{l+1}}\right)\bigg] \\
        &\kappa_{\overleftarrow{e}_{l}}^{s+1,t+1} = 
        \delta_{n_{l}}\mathbb{E}\bigg [\sigma^{s}_{\overleftarrow{e}_{l}}\left(z_{w_{\overrightarrow{e}_{l}}},z_{w_{\overrightarrow{e}_{l}}}\rho^{-1}_{w_{\overrightarrow{e}_{l}}}\nu_{\overrightarrow{e}_{l}}^{s}+Z^{s}_{\overrightarrow{e}_{l}},w_{\overrightarrow{e}_{l+1}}\hat{\nu}_{\overrightarrow{e}_{l+1}}^{s}Z^{s}_{\overleftarrow{e}_{l+1}}\right) \notag \\
        &\sigma^{t}_{\overleftarrow{e}_{l}}\left(z_{w_{\overrightarrow{e}_{l}}},z_{w_{\overrightarrow{e}_{l}}}\rho^{-1}_{w_{\overrightarrow{e}_{l}}}\nu_{\overrightarrow{e}_{l}}^{t}+Z^{t}_{\overrightarrow{e}_{l}},w_{\overrightarrow{e}_{l+1}}\hat{\nu}_{\overrightarrow{e}_{l+1}}^{t}Z^{t}_{\overleftarrow{e}_{l+1}}\right)\bigg ]\end{align}
        \item for l=L
        \begin{align}
            &\nu_{\overrightarrow{e}_{L}}^{0} = \delta_{n_{L-1}}\mathbb{E}\bigg[ w_{\overrightarrow{e}_{l}}\sigma^{0}_{\overrightarrow{e}_{L}}(x^{0}_{\overleftarrow{e}_{L}})\bigg], \thickspace \kappa^{1,1}_{\overrightarrow{e}_{L}} = \delta_{n_{L-1}}\mathbb{E}\bigg[\sigma^{0}_{\overrightarrow{e}_{L}}(x^{0}_{\overleftarrow{e}_{L}})\sigma^{0}_{\overrightarrow{e}_{L}}(x^{0}_{\overleftarrow{e}_{L}})\bigg] \\
            &\nu^{t+1}_{\overrightarrow{e}_{L}} = \delta_{n_{L-1}}\mathbb{E}\left[w_{\overrightarrow{e}_{L}}\sigma^{t}_{\overrightarrow{e}_{L}}\left(z_{w_{\overrightarrow{e}_{L-1}}}\rho^{-1}_{w_{\overrightarrow{e}_{L-1}}}\nu_{\overrightarrow{e}_{L-1}}^{t}+Z^{t}_{\overrightarrow{e}_{L-1}}, w_{\overrightarrow{e}_{L}}\hat{\nu}^{t}_{\overleftarrow{e}_{L}}+Z^{t}_{\overleftarrow{e}_{L}}\right)\right] \\
            &\kappa_{\overrightarrow{e}_{L}}^{s+1,t+1} = \kappa_{\overrightarrow{e}_{L}}^{t+1,s+1} =  \\
            &\delta_{n_{L-1}}\mathbb{E}\bigg[\left(\sigma^{s}_{\overrightarrow{e}_{L}}\left(z_{w_{\overrightarrow{e}_{L-1}}}\rho^{-1}_{w_{\overrightarrow{e}_{L-1}}}\nu_{\overrightarrow{e}_{L-1}}^{s}+Z^{s}_{\overrightarrow{e}_{L-1}},w_{\overrightarrow{e}_{L}}\hat{\nu}^{s}_{\overleftarrow{e}_{L}}+Z^{s}_{\overleftarrow{e}_{L}}\right)-w_{\overrightarrow{e}_{L}}\rho_{w_{\overrightarrow{e}_{L}}}^{-1}\nu_{\overrightarrow{e}_{L}}^{s+1}\right) \notag \\
            &\left(\sigma^{t}_{\overrightarrow{e}_{L}}\left(z_{w_{\overrightarrow{e}_{L-1}}}\rho^{-1}_{w_{\overrightarrow{e}_{L-1}}}\nu_{\overrightarrow{e}_{L-1}}^{t}+Z^{t}_{\overrightarrow{e}_{L-1}},w_{\overrightarrow{e}_{L}}\hat{\nu}^{t}_{\overleftarrow{e}_{L}}+Z^{t}_{\overleftarrow{e}_{L}}\right)-w_{\overrightarrow{e}_{L}}\rho_{w_{\overrightarrow{e}_{L}}}^{-1}\nu_{\overrightarrow{e}_{L}}^{t+1}\right)\bigg] \\
            &\hat{\nu}_{\overleftarrow{e}_{L}}^{0}, \kappa_{\overleftarrow{e}_{L}}^{1,1} = 
        \delta_{n_{L}}\mathbb{E}\bigg[\sigma^{0}_{\overleftarrow{e}_{L}}\left(z_{w_{\overrightarrow{e}_{L}}},x^{0}_{\overrightarrow{e}_{L}}\right)\sigma^{0}_{\overleftarrow{e}_{L}}\left(z_{w_{\overrightarrow{e}_{L}}}\right)\bigg] \\
        &\hat{\nu}^{t+1}_{\overleftarrow{e}_{L}} = \delta_{n_{L}}\mathbb{E}\bigg[\frac{\partial \sigma_{\overleftarrow{e}_{L},i}^{t}}{\partial z_{w_{\overleftarrow{e}_{L}},i},\varphi_{\overleftarrow{e}_{L}}}\left(z_{w_{\overrightarrow{e}_{L}}},z_{w_{\overrightarrow{e}_{L}}}\rho^{-1}_{w_{\overrightarrow{e}_{L}}}\nu_{\overrightarrow{e}_{L}}^{t}+Z^{t}_{\overrightarrow{e}_{L}}\right)\bigg] \\
        &\kappa_{\overleftarrow{e}_{L}}^{s+1,t+1} = 
        \delta_{n_{L}}\mathbb{E}\bigg [\sigma^{s}_{\overleftarrow{e}_{L}}\left(z_{w_{\overrightarrow{e}_{L}}},z_{w_{\overrightarrow{e}_{L}}}\rho^{-1}_{w_{\overrightarrow{e}_{L}}}\nu_{\overrightarrow{e}_{L}}^{s}+Z^{s}_{\overrightarrow{e}_{L}}\right) \notag\\
        &\sigma^{t}_{\overleftarrow{e}_{L}}\left(z_{w_{\overrightarrow{e}_{L}}},z_{w_{\overrightarrow{e}_{L}}}\rho^{-1}_{w_{\overrightarrow{e}_{L}}}\nu_{\overrightarrow{e}_{L}}^{t}+Z^{t}_{\overrightarrow{e}_{L}}\right)\bigg ]\end{align}
        \end{itemize}
        
\begin{lemma}  
\label{lemma:conv_SE_scalar}
        Under the assumptions of section \ref{app_sec_sep}, define, as above, independently for each $\overrightarrow{e}_{l}$, $Z^0_{\overrightarrow{e}_{l}}= x^0_{\overrightarrow{e}}$ and $(Z^1_{\overrightarrow{e}_{l}}, \dots, Z^t_{\overrightarrow{e}_{l}})$ a centered Gaussian random vector of covariance $(\boldsymbol{\kappa}^{r,s}_{\overrightarrow{e}_{l}})_{r,s \leq t}$ (and similarly for left-oriented edges). Then for any $1\leqslant l \leqslant L$, for any sequence of uniformly (in $n$) pseudo-Lipschitz function $\Phi_{l}:(\R^{ n_{l-1}\times (t+1)})^{2} \to \R $
\begin{align*}
    &\Phi\left(\left(\bx^s_{\overrightarrow{e}_{l}}\right)_{0 \leq s \leq t}, \left(\bx^s_{\overleftarrow{e}_{l}}\right)_{0 \leq s \leq t, \overleftarrow{e}_{l-1} \in \overleftarrow{E}}\right) \approxP \\
   & \hspace{2cm}\E \bigg[ \Phi\bigg(\left(z_{w_{\overrightarrow{e}_{l}}}\rho_{w_{\overrightarrow{e}_{l}}}^{-1}\nu_{\overrightarrow{e}_{l}}^{s}+Z^s_{\overrightarrow{e}_{l}}\right)_{0 \leq s \leq t, \overleftarrow{e}_{l} \in \overleftarrow{E}}, \left(w_{\overrightarrow{e}_{l-1}}\hat{\nu}^{s}_{\overleftarrow{e}_{l-1}}+Z^s_{\overleftarrow{e}_{l-1}}\right)_{0 \leq s \leq t}\bigg) \bigg]
\end{align*}
\end{lemma}
\begin{proof}
Consider the following iteration, corresponding to the algorithm presented in the previous section Eq.\eqref{eq:mlamp} with $q=1$ indexed on the same graph as above, but where the matrices $\mathbf{A}_{\overrightarrow{e}_{l}}$ are replaced with random convolutional ones, denoted $\hat{\mathbf{A}}_{\overrightarrow{e}_{l}}$ such that 
\begin{equation}
    \forall \thickspace \overrightarrow{e} \in \overrightarrow{E} \thickspace \hat{\mathbf{A}}_{\overrightarrow{e}_{l}} \sim \mathcal{M}(D_{\overrightarrow{e}_{l}},P_{\overrightarrow{e}_{l}},k_{\overrightarrow{e}_{l}},q_{\overrightarrow{e}_{l}})
\end{equation}
where $\mathbf{A}_{\overrightarrow{e}_{l}} \in \mathbb{R}^{D_{\overrightarrow{e}_{l}}q_{\overrightarrow{e}_{l}} \times P_{\overrightarrow{e}_{l}}q_{\overrightarrow{e}_{l}}}$, and we remind that we chose variances of $1/N$. Since we assume that $q=1$, thus the Onsager terms are scalars, which we denote with lowercase letters $b_{\overrightarrow{e}}^{t}$.
The corresponding iteration then reads:
	\begin{align}
	\begin{split}
	\label{eq:mlamp}
    \bx^{t+1}_{\overrightarrow{e_1}} &= \hat{\bA}_{\overrightarrow{e_1}} \bm^t_{\overrightarrow{e_1}} - b^t_{\overrightarrow{e_1}}\bm^{t-1}_{\overleftarrow{e_1}} \, ,  \\
    &\bm^t_{\overrightarrow{e_1}} = f^t_{\overrightarrow{e}_1}\left(\bx^t_{\overleftarrow{e_1}}\right) \, , \\
        \bx^{t+1}_{\overleftarrow{e_1}} &= \hat{\bA}_{\overrightarrow{e_1}}^\top\bm^t_{\overleftarrow{e_1}} -  b^t_{\overleftarrow{e_1}}\bm^{t-1}_{\overrightarrow{e_1}} \, ,  \\
    &\bm^t_{\overleftarrow{e_1}} = f^t_{\overleftarrow{e_1}}\left(\hat{\mathbf{A}}_{\overrightarrow{e}_{1}}\mathbf{w}_{\overrightarrow{e}_{1}},\bx^t_{\overrightarrow{e_1}},\bx^t_{\overleftarrow{e_2}}\right) \, , \\
    &\qquad \\
        \bx^{t+1}_{\overrightarrow{e_2}} &= \hat{\bA}_{\overrightarrow{e_2}} \bm^t_{\overrightarrow{e_2}} - b^t_{\overrightarrow{e_2}}\bm^{t-1}_{\overleftarrow{e_2}} \, ,  \\
    &\bm^t_{\overrightarrow{e_2}} = f^t_{\overrightarrow{e}_2}\left(\bx^t_{\overrightarrow{e_1}},\bx^t_{\overleftarrow{e_2}}\right) \, , \\
        \bx^{t+1}_{\overleftarrow{e_2}} &= \hat{\bA}_{\overrightarrow{e_2}}^\top\bm^t_{\overleftarrow{e_2}} - b^t_{\overleftarrow{e_2}}\bm^{t-1}_{\overrightarrow{e_2}} \, ,  \\
    &\bm^t_{\overleftarrow{e_2}} = f^t_{\overleftarrow{e_2}}\left(\hat{\mathbf{A}}_{\overrightarrow{e}_{2}}\mathbf{w}_{\overrightarrow{e}_{2}},\bx^t_{\overrightarrow{e_2}},\bx^t_{\overleftarrow{e_3}}\right) \, , \\
    &\qquad\\
    &\qquad\quad\vdots \\
    &\qquad \\
     \bx^{t+1}_{\overrightarrow{e_L}} &= \hat{\bA}_{\overrightarrow{e_L}} \bm^t_{\overrightarrow{e_L}} - b^t_{\overrightarrow{e_L}}\bm^{t-1}_{\overleftarrow{e_L}} \, ,  \\
    &\bm^t_{\overrightarrow{e_L}} = f^t_{\overrightarrow{e}_L}\left(\bx^t_{\overrightarrow{e}_{L-1}},\bx^t_{\overleftarrow{e_L}}\right) \, , \\
    \bx^{t+1}_{\overleftarrow{e_L}} &= \hat{\bA}_{\overrightarrow{e_L}}^\top\bm^t_{\overleftarrow{e_L}} - b^t_{\overleftarrow{e_L}}\bm^{t-1}_{\overrightarrow{e_L}} \, ,  \\
    &\bm^t_{\overleftarrow{e_L}} = f^t_{\overleftarrow{e_L}}\left(\hat{\mathbf{A}}_{\overrightarrow{e}_{L}}\mathbf{w}_{\overrightarrow{e}_{L}},\bx^t_{\overrightarrow{e_L}}\right) \,
    \end{split}
\end{align}
Then, according to Lemma \ref{lem:permutation}, for any $1\leqslant l \leqslant L$, there exists a pair of orthogonal matrices $\mathbf{U}_{\overrightarrow{e}_{l}} \in \mathbb{R}^{D_{\overrightarrow{e}_{l}}q_{\overrightarrow{e}_{l}} \times D_{\overrightarrow{e}_{l}}q_{\overrightarrow{e}_{l}}}, \mathbf{V}_{\overrightarrow{e}_{l}} \in \mathbb{R}^{P_{\overrightarrow{e}_{l}}q_{\overrightarrow{e}_{l}} \times P_{\overrightarrow{e}_{l}}q_{\overrightarrow{e}_{l}}}$ such that $\hat{\mathbf{A}}_{\overrightarrow{e}_{l}} = \mathbf{U}_{\overrightarrow{e}_{l}}\tilde{\mathbf{A}}_{\overrightarrow{e}_{l}}\mathbf{V}_{\overrightarrow{e}_{l}}^{\top}$ and $\tilde{\mathbf{A}}_{\overrightarrow{e}_{l}} = \left[\left(\Pcl_{P_{\overrightarrow{e}_{l}},q_{\overrightarrow{e}_{l}}}\right)^{i-1}\mathbf{Q}_{\overrightarrow{e}_{l}}\right]_{i=1}^{q_{\overrightarrow{e}_{l}}}$, where $\mathbf{Q}_{\overrightarrow{e}_{l}} \in \mathbb{R}^{D_{\overrightarrow{e}_{l}}\times P_{\overrightarrow{e}_{l}}q_{\overrightarrow{e}_{l}}}$ is composed of $q_{\overrightarrow{e}_{l}}$ blocks of size $D_{\overrightarrow{e}_{l}}\times P_{\overrightarrow{e}_{l}}$, denoted $\mathbf{Q}^{j}_{\overrightarrow{e}_{l}}$, verifying 
\begin{itemize}
    \item for any $1\leqslant j \leqslant k_{\overrightarrow{e}}$,  $\mathbf{Q}^{j}_{\overrightarrow{e}}$ has i.i.d. $\mathcal{N}(0,\frac{1}{N})$ elements
    \item for any $k_{\overrightarrow{e}} < j \leqslant q_{\overrightarrow{e}}$, all elements of $\mathbf{Q}^{j}_{\overrightarrow{e}}$ are zero.
\end{itemize}
In the preceding definition of $\tilde{\mathbf{A}}_{\overrightarrow{e}_{l}}$, $\mathbf{Q}_{\overrightarrow{e}_{l}}$ is understood as a vector of size $\mathbb{R}^{P_{\overrightarrow{e}}q_{\overrightarrow{e}}}$ with elements in $\mathbb{R}^{D_{\overrightarrow{e}}}$, such that the permutation matrix $\Pcl_{P_{\overrightarrow{e}},q_{\overrightarrow{e}}}$ shifts blocks of size $D_{\overrightarrow{e}} \times P_{\overrightarrow{e}}$, yielding
\begin{equation}
    \tilde{\mathbf{A}}_{\overrightarrow{e}} = \begin{bmatrix}
    \mathbf{Q}_{\overrightarrow{e}_{l}}^{(1)} & \mathbf{Q}_{\overrightarrow{e}_{l}}^{(2)} & \ldots & \mathbf{Q}_{\overrightarrow{e}_{l}}^{(k_{\overrightarrow{e}})} & & & & &  \\
     & \mathbf{Q}_{\overrightarrow{e}_{l}}^{(1)} & \mathbf{Q}_{\overrightarrow{e}_{l}}^{(2)} & \ldots & \mathbf{Q}_{\overrightarrow{e}_{l}}^{(k_{\overrightarrow{e}})} & & & & \\
    & & \mathbf{Q}_{\overrightarrow{e}_{l}}^{(1)} & \mathbf{Q}_{\overrightarrow{e}_{l}}^{(2)} & \ldots & \mathbf{Q}_{\overrightarrow{e}_{l}}^{(k_{\overrightarrow{e}})} & & & \vdots  \\
    \vdots &  \vdots & \ddots & & & & \\
    \mathbf{Q}_{\overrightarrow{e}_{l}}^{(2)} & \mathbf{Q}_{\overrightarrow{e}_{l}}^{(3)} & \ldots & \mathbf{Q}_{\overrightarrow{e}_{l}}^{(k_{\overrightarrow{e}})} & & & & & \mathbf{Q}_{\overrightarrow{e}_{l}}^{(1)}
    \end{bmatrix}
\end{equation}
The iteration then reads
\begin{align}
	\begin{split}
	\label{eq:mlamp1}
    \bx^{t+1}_{\overrightarrow{e_1}} &= \mathbf{U}_{\overrightarrow{e}_{1}}\tilde{\mathbf{A}}_{\overrightarrow{e}_{1}}\mathbf{V}_{\overrightarrow{e}_{1}}^{\top} \bm^t_{\overrightarrow{e_1}} - b^t_{\overrightarrow{e_1}}\bm^{t-1}_{\overleftarrow{e_1}} \, ,  \\
    &\bm^t_{\overrightarrow{e_1}} = f^t_{\overrightarrow{e}_1}\left(\bx^t_{\overleftarrow{e_1}}\right) \, , \\
        \bx^{t+1}_{\overleftarrow{e_1}} &= \mathbf{V}_{\overrightarrow{e}_{1}}\tilde{\mathbf{A}}_{\overrightarrow{e}_{1}}^{\top}\mathbf{U}_{\overrightarrow{e}_{1}}^{\top}\bm^t_{\overleftarrow{e_1}} - b^t_{\overleftarrow{e_1}}\bm^{t-1}_{\overrightarrow{e_1}} \, ,  \\
    &\bm^t_{\overleftarrow{e_1}} = f^t_{\overleftarrow{e_1}}\left(\mathbf{U}_{\overrightarrow{e}_{1}}\tilde{\mathbf{A}}_{\overrightarrow{e}_{1}}\mathbf{V}_{\overrightarrow{e}_{1}}^{\top}\mathbf{w}_{\overrightarrow{e}_{1}},\bx^t_{\overrightarrow{e_1}},\bx^t_{\overleftarrow{e_2}}\right) \, , \\
    &\qquad \\
        \bx^{t+1}_{\overrightarrow{e_2}} &= \mathbf{U}_{\overrightarrow{e}_{2}}\tilde{\mathbf{A}}_{\overrightarrow{e}_{2}}\mathbf{V}_{\overrightarrow{e}_{2}}^{\top} \bm^t_{\overrightarrow{e_2}} - b^t_{\overrightarrow{e_2}}\bm^{t-1}_{\overleftarrow{e_2}} \, ,  \\
    &\bm^t_{\overrightarrow{e_2}} = f^t_{\overrightarrow{e}_2}\left(\bx^t_{\overrightarrow{e_1}},\bx^t_{\overleftarrow{e_2}}\right) \, , \\
        \bx^{t+1}_{\overleftarrow{e_2}} &= \mathbf{V}_{\overrightarrow{e}_{2}}\tilde{\mathbf{A}}_{\overrightarrow{e}}^{\top}\mathbf{U}_{\overrightarrow{e}_{2}}^{\top}\bm^t_{\overleftarrow{e_2}} - b^t_{\overleftarrow{e_2}}\bm^{t-1}_{\overrightarrow{e_2}} \, ,  \\
    &\bm^t_{\overleftarrow{e_2}} = f^t_{\overleftarrow{e_2}}\left(\mathbf{U}_{\overrightarrow{e}_{2}}\tilde{\mathbf{A}}_{\overrightarrow{e}_{2}}\mathbf{V}_{\overrightarrow{e}_{2}}^{\top}\mathbf{w}_{\overrightarrow{e}_{2}},\bx^t_{\overrightarrow{e_2}},\bx^t_{\overleftarrow{e_3}}\right) \, , \\
    &\qquad\\
    &\qquad\quad\vdots \\
    &\qquad \\
     \bx^{t+1}_{\overrightarrow{e_L}} &= \mathbf{U}_{\overrightarrow{e}_{L}}\tilde{\mathbf{A}}_{\overrightarrow{e}_{L}}\mathbf{V}_{\overrightarrow{e}_{L}}^{\top} \bm^t_{\overrightarrow{e_L}} - b^t_{\overrightarrow{e_L}}\bm^{t-1}_{\overleftarrow{e_L}} \, ,  \\
    &\bm^t_{\overrightarrow{e_L}} = f^t_{\overrightarrow{e}_L}\left(\bx^t_{\overrightarrow{e}_{L-1}},\bx^t_{\overleftarrow{e_L}}\right) \, , \\
    \bx^{t+1}_{\overleftarrow{e_L}} &= \mathbf{V}_{\overrightarrow{e}_{L}}\tilde{\mathbf{A}}_{\overrightarrow{e}_{L}}^{\top}\mathbf{U}_{\overrightarrow{e}_{L}}^{\top}\bm^t_{\overleftarrow{e_L}} - b^t_{\overleftarrow{e_L}}\bm^{t-1}_{\overrightarrow{e_L}} \, ,  \\
    &\bm^t_{\overleftarrow{e_L}} = f^t_{\overleftarrow{e_L}}\left(\mathbf{U}_{\overrightarrow{e}_{L}}\tilde{\mathbf{A}}_{\overrightarrow{e}_{L}}\mathbf{V}_{\overrightarrow{e}_{L}}^{\top}\mathbf{w}_{\overrightarrow{e}_{L}},\bx^t_{\overrightarrow{e_L}}\right) \,
    \end{split}
\end{align}
Since we will not be making any change of variable on the $\mathbf{w}_{\overrightarrow{e}_{l}}$, we will keep the $\hat{\mathbf{A}}_{\overrightarrow{e}_{l}}$ notation for the quantities related to the planted model.
Define, for any $1\leqslant l \leqslant L$ and any $t \in \mathbb{N}$:
\begin{align*}
    &\tilde{\mathbf{x}}_{\overrightarrow{e}_{l}} = \mathbf{U}_{\overrightarrow{e}_{l}}^{\top}\mathbf{x}_{\overrightarrow{e}_{l}} \qquad \tilde{\mathbf{x}}_{\overleftarrow{e}_{l}} = \mathbf{V}_{\overrightarrow{e}_{l}}^{\top}\mathbf{x}_{\overleftarrow{e}_{l}} \\
    &\tilde{\mathbf{m}}^{t}_{\overrightarrow{e}_{l}} = \mathbf{V}_{\overrightarrow{e}_{l}}^{\top}\mathbf{m}^{t}_{\overrightarrow{e}_{l}} \qquad \tilde{\mathbf{m}}^{t}_{\overleftarrow{e}_{l}} = \mathbf{U}^{\top}_{\overrightarrow{e}_{l}}\mathbf{m}^{t}_{\overleftarrow{e}_{l}} \\
   &\tilde{f}^t_{\overrightarrow{e}_1}(\tilde{\bx}^t_{\overleftarrow{e_1}}) =  \mathbf{V}_{\overrightarrow{e}_{1}}^{\top}f^t_{\overrightarrow{e}_1}\left(\mathbf{V}_{\overrightarrow{e}_{1}}\tilde{\bx}^t_{\overleftarrow{e_1}}\right) \\
   &\tilde{f}^t_{\overleftarrow{e_1}}\left(\hat{\mathbf{A}}_{\overrightarrow{e}_{1}}\mathbf{w}_{\overrightarrow{e}_{1}},\tilde{\bx}^t_{\overrightarrow{e_1}},\tilde{\bx}^t_{\overleftarrow{e_2}}\right) = \mathbf{U}^{\top}_{\overrightarrow{e}_{1}}f^t_{\overleftarrow{e_1}}\left(\hat{\mathbf{A}}_{\overrightarrow{e}_{1}}\mathbf{w}_{\overrightarrow{e}_{1}},\mathbf{U}_{\overrightarrow{e}_{1}}\tilde{\bx}^t_{\overrightarrow{e_1}},\mathbf{V}_{\overrightarrow{e}_{2}}\tilde{\bx}^t_{\overleftarrow{e_2}}\right) \\
    &\tilde{f}^t_{\overrightarrow{e}_2}\left(\tilde{\bx}^t_{\overrightarrow{e_1}},\tilde{\bx}^t_{\overleftarrow{e_2}}\right)=\mathbf{V}_{\overrightarrow{e}_{2}}^{\top}f^t_{\overrightarrow{e}_2}\left(\mathbf{U}_{\overrightarrow{e}_{1}}\tilde{\bx}^t_{\overrightarrow{e_1}},\mathbf{V}_{\overrightarrow{e}_{2}}\tilde{\bx}^t_{\overleftarrow{e_2}}\right) \\
    &\tilde{f}^t_{\overleftarrow{e_2}}\left(\hat{\mathbf{A}}_{\overrightarrow{e}_{2}}\mathbf{w}_{\overrightarrow{e}_{2}},\tilde{\bx}^t_{\overrightarrow{e_2}},\tilde{\bx}^t_{\overleftarrow{e_3}}\right)= \mathbf{U}_{\overrightarrow{e}_{2}}^{\top}f^t_{\overleftarrow{e_2}}\left(\hat{\mathbf{A}}_{\overrightarrow{e}_{2}}\mathbf{w}_{\overrightarrow{e}_{2}},\mathbf{U}_{\overrightarrow{e}_{2}}\tilde{\bx}^t_{\overrightarrow{e_2}},\mathbf{V}_{\overrightarrow{e}_{3}}\tilde{\bx}^t_{\overleftarrow{e_3}}\right)\\
       &\qquad\\
    &\qquad\quad\vdots \\
    &\qquad \\
    & \tilde{f}^t_{\overrightarrow{e}_L}\left(\tilde{\bx}^t_{\overrightarrow{e}_{L-1}},\tilde{\bx}^t_{\overleftarrow{e_L}}\right) = \mathbf{V}_{\overrightarrow{e}_{L}}^{\top}f^t_{\overrightarrow{e}_L}\left(\mathbf{U}_{\overrightarrow{e}_{L-1}}\tilde{\bx}^t_{\overrightarrow{e}_{L-1}},\mathbf{V}_{\overrightarrow{e}_{L}}\tilde{\bx}^t_{\overleftarrow{e_L}}\right) \\
    &\tilde{f}^t_{\overleftarrow{e_L}}\left(\hat{\mathbf{A}}_{\overrightarrow{e}_{L}}\mathbf{w}_{\overrightarrow{e}_{L}},\tilde{\bx}^t_{\overrightarrow{e_L}}\right) = \mathbf{U}_{\overrightarrow{e}_{L}}^{\top}f^t_{\overleftarrow{e_L}}\left(\mathbf{U}_{\overrightarrow{e}_{L}}\tilde{\mathbf{A}}_{\overrightarrow{e}_{L}}\mathbf{V}_{\overrightarrow{e}_{L}}\mathbf{w}_{\overrightarrow{e}_{L}},\mathbf{U}_{\overrightarrow{e}_{L}}\tilde{\bx}^t_{\overrightarrow{e_L}}\right)
\end{align*}
Using the orthogonality of the permutation matrices $\mathbf{U}_{\overrightarrow{e}},\mathbf{V}_{\overrightarrow{e}}$, the iteration may be rewritten
\begin{align}
\label{eq:mlamp2}
	\begin{split}
    \tilde{\bx}^{t+1}_{\overrightarrow{e_1}} &= \tilde{\mathbf{A}}_{\overrightarrow{e}_{1}} \tilde{\bm}^t_{\overrightarrow{e_1}} -  b^t_{\overrightarrow{e_1}}\tilde{\bm}^{t-1}_{\overleftarrow{e_1}} \, ,  \\
    &\tilde{\bm}^t_{\overrightarrow{e_1}} = \tilde{f}^t_{\overrightarrow{e}_1}(\tilde{\bx}^t_{\overleftarrow{e_1}})\, , \\
        \tilde{\bx}^{t+1}_{\overleftarrow{e_1}} &= \tilde{\mathbf{A}}_{\overrightarrow{e}_{1}}^{\top}\tilde{\bm}^t_{\overleftarrow{e_1}} -  b^t_{\overleftarrow{e_1}}\tilde{\bm}^{t-1}_{\overrightarrow{e_1}} \, ,  \\
    &\tilde{\bm}^t_{\overleftarrow{e_1}} = \tilde{f}^t_{\overleftarrow{e_1}}\left(\hat{\mathbf{A}}_{\overrightarrow{e}_{1}}\mathbf{w}_{\overrightarrow{e}_{1}},\tilde{\bx}^t_{\overrightarrow{e_1}},\tilde{\bx}^t_{\overleftarrow{e_2}}\right) \, , \\
    &\qquad \\
        \tilde{\bx}^{t+1}_{\overrightarrow{e_2}} &= \tilde{\mathbf{A}}_{\overrightarrow{e}_{2}}\tilde{\bm}^t_{\overrightarrow{e_2}} - b^t_{\overrightarrow{e_2}}\tilde{\bm}^{t-1}_{\overleftarrow{e_2}} \, ,  \\
    &\tilde{\bm}^t_{\overrightarrow{e_2}} = \tilde{f}^t_{\overrightarrow{e}_2}\left(\tilde{\bx}^t_{\overrightarrow{e_1}},\tilde{\bx}^t_{\overleftarrow{e_2}}\right) \, , \\
        \tilde{\bx}^{t+1}_{\overleftarrow{e_2}} &= \tilde{\mathbf{A}}_{\overrightarrow{e}}^{\top}\tilde{\bm}^t_{\overleftarrow{e_2}} - b^t_{\overleftarrow{e_2}}\tilde{\bm}^{t-1}_{\overrightarrow{e_2}} \, ,  \\
    &\tilde{\bm}^t_{\overleftarrow{e_2}} = \tilde{f}^t_{\overleftarrow{e_2}}\left(\hat{\mathbf{A}}_{\overrightarrow{e}_{2}}\mathbf{w}_{\overrightarrow{e}_{2}},\tilde{\bx}^t_{\overrightarrow{e_2}},\tilde{\bx}^t_{\overleftarrow{e_3}}\right) \, , \\
    &\qquad\\
    &\qquad\quad\vdots \\
    &\qquad \\
     \tilde{\bx}^{t+1}_{\overrightarrow{e_L}} &= \tilde{\mathbf{A}}_{\overrightarrow{e}_{L}} \tilde{\bm}^t_{\overrightarrow{e_L}} - b^t_{\overrightarrow{e_L}}\tilde{\bm}^{t-1}_{\overleftarrow{e_L}} \, ,  \\
    &\tilde{\bm}^t_{\overrightarrow{e_L}} = \tilde{f}^t_{\overrightarrow{e}_L}\left(\tilde{\bx}^t_{\overrightarrow{e}_{L-1}},\tilde{\bx}^t_{\overleftarrow{e_L}}\right) \, , \\
    \tilde{\bx}^{t+1}_{\overleftarrow{e_L}} &= \tilde{\mathbf{A}}_{\overrightarrow{e}_{L}}^{\top}\tilde{\bm}^t_{\overleftarrow{e_L}} - b^t_{\overleftarrow{e_L}}\tilde{\bm}^{t-1}_{\overrightarrow{e_L}}\, ,  \\
    &\tilde{\bm}^t_{\overleftarrow{e_L}} = \tilde{f}^t_{\overleftarrow{e_L}}\left(\hat{\mathbf{A}}_{\overrightarrow{e}_{L}}\mathbf{w}_{\overrightarrow{e}_{L}},\tilde{\bx}^t_{\overrightarrow{e_L}}\right) \,
    \end{split}
\end{align}
Recall, for any $1 \leqslant l \leqslant L$, the dimensions  $\tilde{\mathbf{A}}_{\overrightarrow{e}_{l}} \in \mathbb{R}^{D_{\overrightarrow{e}_{l}}q_{\overrightarrow{e}_{l}}\times P_{\overrightarrow{e}_{l}}q_{\overrightarrow{e}_{l}}}$ and $\tilde{f}^{t}_{\overrightarrow{e}_{l}}(...)\in \mathbb{R}^{P_{\overrightarrow{e}_{l}}q_{\overrightarrow{e}_{l}}}$. Consider then 
\begin{align}
    \tilde{f}^{t}_{\overrightarrow{e}_{l}}(...) = \begin{bmatrix}
    \left(\tilde{f}^{t}_{\overrightarrow{e}_{l}}\right)^{(1)}(...) \\
    \vdots \\
    \left(\tilde{f}^{t}_{\overrightarrow{e}_{l}}\right)^{(q_{\overrightarrow{e}_{l}})}(...)
    \end{bmatrix}
\end{align}
where, for any $1 \leqslant k \leqslant q_{\overrightarrow{e}_{l}}$, $(\tilde{f}^{t}_{\overrightarrow{e}_{l}})^{(k)}(...) \in \mathbb{R}^{ P_{\overrightarrow{e}}}$. The product $\tilde{\mathbf{A}}_{\overrightarrow{e}_{l}}\tilde{f}^{t}_{\overrightarrow{e}_{l}}(...) \in \mathbb{R}^{D_{\overrightarrow{e}_{l}}q_{\overrightarrow{e}_{l}}}$ then reads, using the circulant structure of $\tilde{\mathbf{A}}_{\overrightarrow{e}_{l}}$
\begin{align}
    &\begin{bmatrix}
    \mathbf{Q}_{\overrightarrow{e}_{l}}^{(1)} & \mathbf{Q}_{\overrightarrow{e}_{l}}^{(2)} & \ldots & \mathbf{Q}_{\overrightarrow{e}_{l}}^{(k_{\overrightarrow{e}})} & & & & &  \\
     & \mathbf{Q}_{\overrightarrow{e}_{l}}^{(1)} & \mathbf{Q}_{\overrightarrow{e}_{l}}^{(2)} & \ldots & \mathbf{Q}_{\overrightarrow{e}_{l}}^{(k_{\overrightarrow{e}})} & & & & \\
    & & \mathbf{Q}_{\overrightarrow{e}_{l}}^{(1)} & \mathbf{Q}_{\overrightarrow{e}_{l}}^{(2)} & \ldots & \mathbf{Q}_{\overrightarrow{e}_{l}}^{(k_{\overrightarrow{e}})} & & & \vdots  \\
    \vdots &  \vdots & \ddots & & & & \\
    \mathbf{Q}_{\overrightarrow{e}_{l}}^{(2)} & \mathbf{Q}_{\overrightarrow{e}_{l}}^{(3)} & \ldots & \mathbf{Q}_{\overrightarrow{e}_{l}}^{(k_{\overrightarrow{e}})} & & & & & \mathbf{Q}_{\overrightarrow{e}_{l}}^{(1)}
    \end{bmatrix} \begin{bmatrix}
    \left(\tilde{f}^{t}_{\overrightarrow{e}_{l}}\right)^{(1)}(...) \\
    \vdots \\
    \left(\tilde{f}^{t}_{\overrightarrow{e}_{l}}\right)^{(q_{\overrightarrow{e}_{l}})}(...)
    \end{bmatrix} \\
   &= \left[\left(\left(\Pcl_{P_{\overrightarrow{e}_{l}},q_{\overrightarrow{e}_{l}}}\right)^{i-1}\mathbf{Q}_{\overrightarrow{e}_{l}}\right)\tilde{f}^{t}_{\overrightarrow{e}_{l}}(...)\right]_{i=1}^{q_{\overrightarrow{e}_{l}}} \\
    &=\left[\sum_{j=1}^{k_{\overrightarrow{e}_{l}}}\mathbf{Q}^{(j)}_{\overrightarrow{e}_{l}}(\tilde{f}^{t}_{\overrightarrow{e}_{l}})^{(\lfloor j+n-2\rfloor_{q_{\overrightarrow{e}_{l}}}+1)}(...)\right]_{n=1}^{q_{\overrightarrow{e}_{l}}}
\end{align}
where the notation $\lfloor.\rfloor_{q_{\overrightarrow{e}_{l}}}$ denotes the modulo $q_{\overrightarrow{e}_{l}}$, i.e. the remainder of the euclidian division by $q_{\overrightarrow{e}_{l}}$. Now define 
\begin{align}
    \tilde{F}_{\overrightarrow{e}_{l}}^{t}(...) =\begin{bmatrix}\left[\left(\Pcl_{P_{\overrightarrow{e}_{l}},q_{\overrightarrow{e}_{l}}}\right)^{1-i} \left[(\tilde{f}_{\overrightarrow{e}_{l}}^{t})^{(1)} \ldots (\tilde{f}_{\overrightarrow{e}_{l}}^{t})^{(q_{\overrightarrow{e}_{l}})}\right] \right]_{i=1}^{k_{\overrightarrow{e}_{l}}} \in \mathbb{R}^{P_{\overrightarrow{e}_{l}}k_{\overrightarrow{e}_{l}}\times q_{\overrightarrow{e}_{l}}} \\
    \left[0_{P_{\overrightarrow{e}_{l}}} \ldots0_{P_{\overrightarrow{e}_{l}}}\right]_{j=1}^{q_{\overrightarrow{e}_{l}}-k_{\overrightarrow{e}_{l}}}\end{bmatrix} \in \mathbb{R}^{P_{\overrightarrow{e}_{l}}q_{\overrightarrow{e}_{l}}\times q_{\overrightarrow{e}_{l}}}
\end{align}
and the matrix $\tilde{\mathbf{Q}}_{\overrightarrow{e}_{l}} \in \mathbb{R}^{D_{\overrightarrow{e}_{l}}q_{\overrightarrow{e}_{l}} \times P_{\overrightarrow{e}_{l}}q_{\overrightarrow{e}_{l}}}$ is a dense Gaussian matrix with i.i.d. elements. Then
\begin{align*}
    &\tilde{\mathbf{Q}}_{\overrightarrow{e}_{l}}\tilde{F}_{\overrightarrow{e}_{l}}^{t}(...) = \begin{bmatrix}\sum_{j=1}^{k_{\overrightarrow{e}_{l}}}\mathbf{Q}_{\overrightarrow{e}_{l}}^{(j)}(\tilde{f}^{t}_{\overrightarrow{e}_{l}})^{\lfloor j-1\rfloor_{q_{\overrightarrow{e}_{l}}}+1}(...) \quad  \ldots \quad  \sum_{j=1}^{k_{\overrightarrow{e}_{l}}}(\mathbf{Q}^{(j)}_{\overrightarrow{e}_{l}})(\tilde{f}^{t}_{\overrightarrow{e}_{l}})^{\lfloor j+q_{\overrightarrow{e}_{l}}-2\rfloor_{q_{\overrightarrow{e}_{l}}}+1}(...) \\
    \ldots \\
    \ldots \end{bmatrix}\\
     &\hspace{5cm}\in \mathbb{R}^{D_{\overrightarrow{e}_{l}}q_{\overrightarrow{e}_{l}} \times q_{\overrightarrow{e}_{l}}}
\end{align*}
where each $\ldots$ is an identical copy of the first $D_{\overrightarrow{e}_{l}} \times q_{\overrightarrow{e}_{l}}$ block, for a total of $k_{\overrightarrow{e}_{l}}$ blocks.
This means the $D_{\overrightarrow{e}_{l}}q_{\overrightarrow{e}_{l}}$ output of the product $\tilde{\mathbf{A}}_{\overrightarrow{e}_{l}}f^{t}_{\overrightarrow{e}_{l}}(...)$ may be rewritten as a $D_{\overrightarrow{e}_{l}}\times q_{\overrightarrow{e}_{l}}$ matrix (copied $k_{\overrightarrow{e}_{l}}$ times) resulting from the product of a dense Gaussian matrix with i.i.d. elements and a matrix valued function $\tilde{F}_{\overrightarrow{e}_{l}}^{t}$ which verifies the same regularity conditions as $f^{t}_{\overrightarrow{e}_{l}}$. Note that, owing to the separability assumption, we may use any permutation of the $(\tilde{f}^{t}_{\overrightarrow{e}_{l}})^{(i)}, 1\leqslant i\leqslant q_{\overrightarrow{e}_{l}}$ and will thus drop the permutations to write 
\begin{align}
    \tilde{F}_{\overrightarrow{e}_{l}}^{t}(...) =\begin{bmatrix}\left[ (\tilde{f}_{\overrightarrow{e}_{l}}^{t})^{(1)} \ldots (\tilde{f}_{\overrightarrow{e}_{l}}^{t})^{(q_{\overrightarrow{e}_{l}})} \right]_{i=1}^{k_{\overrightarrow{e}_{l}}} \in \mathbb{R}^{P_{\overrightarrow{e}_{l}}k_{\overrightarrow{e}_{l}}\times q_{\overrightarrow{e}_{l}}} \\
    \left[0_{P_{\overrightarrow{e}_{l}}} \ldots0_{P_{\overrightarrow{e}_{l}}}\right]_{j=1}^{q_{\overrightarrow{e}_{l}}-k_{\overrightarrow{e}_{l}}}\end{bmatrix}\in \mathbb{R}^{P_{\overrightarrow{e}_{l}}q_{\overrightarrow{e}_{l}}\times q_{\overrightarrow{e}_{l}}}
\end{align}
Similarly, for products of the form $\left(\tilde{\mathbf{A}}_{\overrightarrow{e}_{l}}\right)^{\top}\tilde{f}^{t}_{\overleftarrow{e}_{l}}(...) \in \mathbb{R}^{P_{\overrightarrow{e}_{l}}q_{\overrightarrow{e}_{l}}}$, we may write:
\begin{align}
    &\begin{bmatrix}
    \mathbf{Q}_{\overrightarrow{e}_{l}}^{(1)} & \mathbf{Q}_{\overrightarrow{e}_{l}}^{(2)} & \ldots & \mathbf{Q}_{\overrightarrow{e}_{l}}^{(k_{\overrightarrow{e}_{l}})} & & & & &  \\
     & \mathbf{Q}_{\overrightarrow{e}_{l}}^{(1)} & \mathbf{Q}_{\overrightarrow{e}_{l}}^{(2)} & \ldots & \mathbf{Q}_{\overrightarrow{e}_{l}}^{(k_{\overrightarrow{e}_{l}})} & & & & \\
    & & \mathbf{Q}_{\overrightarrow{e}_{l}}^{(1)} & \mathbf{Q}_{\overrightarrow{e}_{l}}^{(2)} & \ldots & \mathbf{Q}_{\overrightarrow{e}_{l}}^{(k_{\overrightarrow{e}_{l}})} & & & \vdots  \\
    \vdots &  \vdots & \ddots & & & & \\
    \mathbf{Q}_{\overrightarrow{e}_{l}}^{(2)} & \mathbf{Q}_{\overrightarrow{e}_{l}}^{(3)} & \ldots & \mathbf{Q}_{\overrightarrow{e}_{l}}^{(k_{\overrightarrow{e}_{l}})} & & & & & \mathbf{Q}_{\overrightarrow{e}_{l}}^{(1)}
    \end{bmatrix}^{\top} \begin{bmatrix}
    \left(\tilde{f}^{t}_{\overleftarrow{e}_{l}}\right)^{(1)}(...) \\
    \vdots \\
    \left(\tilde{f}^{t}_{\overleftarrow{e}_{l}}\right)^{(q_{\overrightarrow{e}_{l}})}(...)
    \end{bmatrix} \\
   &= \left[\left(\left(\Pcl_{P_{\overrightarrow{e}_{l}},q_{\overrightarrow{e}_{l}}}\right)^{i-1}\left[(\mathbf{Q}^{(1)}_{\overrightarrow{e}_{l}})^{\top} (0 \ldots 0)(\mathbf{Q}^{(k_{\overrightarrow{e}_{l}})}_{\overrightarrow{e}_{l}})^{\top} \ldots (\mathbf{Q}^{(2)}_{\overrightarrow{e}_{l}})^{\top}\right]\right)\tilde{f}^{t}_{\overleftarrow{e}_{l}}(...)\right]_{i=1}^{q_{\overrightarrow{e}_{l}}}
\end{align}
Then, using once again the separability assumption, we may define:
\begin{align}
    \tilde{F}_{\overleftarrow{e}_{l}}^{t}(...) =\begin{bmatrix}\left[ (\tilde{f}_{\overleftarrow{e}_{l}}^{t})^{(1)} \ldots (\tilde{f}_{\overleftarrow{e}_{l}}^{t})^{(q_{\overleftarrow{e}_{l}})} \right]_{i=1}^{k_{\overrightarrow{e}_{l}}} \in \mathbb{R}^{D_{\overrightarrow{e}_{l}}k_{\overrightarrow{e}_{l}}\times q_{\overrightarrow{e}_{l}}} \\
    \left[0_{D_{\overrightarrow{e}_{l}}} \ldots0_{D_{\overrightarrow{e}_{l}}}\right]\end{bmatrix} \in \mathbb{R}^{D_{\overrightarrow{e}_{l}}q_{\overrightarrow{e}_{l}}\times q_{\overrightarrow{e}_{l}}}
\end{align}
such that the term $
    \tilde{\mathbf{Q}}_{\overrightarrow{e}_{l}}^{\top}\tilde{F}_{\overleftarrow{e}_{l}}^{t}(...)$ also contains $k_{\overrightarrow{e}_{l}}$ copies of a $P_{\overrightarrow{e}_{l}}\times q_{\overrightarrow{e}_{l}}$ block containing the $q_{\overrightarrow{e}_{l}}$ blocks of size $P_{\overrightarrow{e}_{l}}$ of the original $P_{\overrightarrow{e}_{l}}q_{\overrightarrow{e}_{l}}$ vector $\tilde{\mathbf{A}}_{\overrightarrow{e}_{l}}^{\top}\tilde{f}_{\overleftarrow{e}_{l}}(...)$.
The iterates of the sequences defined by Eq.\eqref{eq:mlamp2} may then be rewritten as a subset of the lines of the following matrix valued iteration, i.e.: 
\begin{align}
\label{eq:mlamp3}
	\begin{split}
    \tilde{\bX}^{t+1}_{\overrightarrow{e_1}} &= \tilde{\mathbf{Q}}_{\overrightarrow{e}_{1}} \tilde{\bM}^t_{\overrightarrow{e_1}} -  b^t_{\overrightarrow{e_1}}\tilde{\bM}^{t-1}_{\overleftarrow{e_1}} \, ,  \\
    &\tilde{\bM}^t_{\overrightarrow{e_1}} = \tilde{F}^t_{\overrightarrow{e}_1}(\tilde{\bX}^t_{\overleftarrow{e_1}})\, , \\
        \tilde{\bX}^{t+1}_{\overleftarrow{e_1}} &= \tilde{\mathbf{Q}}_{\overrightarrow{e}_{1}}^{\top}\tilde{\bM}^t_{\overleftarrow{e_1}} -  b^t_{\overleftarrow{e_1}}\tilde{\bM}^{t-1}_{\overrightarrow{e_1}} \, ,  \\
    &\tilde{\bM}^t_{\overleftarrow{e_1}} = \tilde{F}^t_{\overleftarrow{e_1}}\left(\tilde{\mathbf{Q}}_{\overrightarrow{e}_{1}}\mathbf{W}_{\overrightarrow{e}_{1}},\tilde{\bX}^t_{\overrightarrow{e_1}},\tilde{\bX}^t_{\overleftarrow{e_2}}\right) \, , \\
    &\qquad \\
        \tilde{\bX}^{t+1}_{\overrightarrow{e_2}} &= \tilde{\mathbf{Q}}_{\overrightarrow{e}_{2}}\tilde{\bM}^t_{\overrightarrow{e_2}} - b^t_{\overrightarrow{e_2}}\tilde{\bM}^{t-1}_{\overleftarrow{e_2}} \, ,  \\
    &\tilde{\bM}^t_{\overrightarrow{e_2}} = \tilde{F}^t_{\overrightarrow{e}_2}\left(\tilde{\bX}^t_{\overrightarrow{e_1}},\tilde{\bX}^t_{\overleftarrow{e_2}}\right) \, , \\
        \tilde{\bX}^{t+1}_{\overleftarrow{e_2}} &= \tilde{\mathbf{Q}}_{\overrightarrow{e}}^{\top}\tilde{\bM}^t_{\overleftarrow{e_2}} - b^t_{\overleftarrow{e_2}}\tilde{\bM}^{t-1}_{\overrightarrow{e_2}} \, ,  \\
    &\tilde{\bM}^t_{\overleftarrow{e_2}} = \tilde{F}^t_{\overleftarrow{e_2}}\left(\tilde{\mathbf{Q}}_{\overrightarrow{e}_{2}}\mathbf{W}_{\overrightarrow{e}_{2}},\tilde{\bX}^t_{\overrightarrow{e_2}},\tilde{\bX}^t_{\overleftarrow{e_3}}\right) \, , \\
    &\qquad\\
    &\qquad\quad\vdots \\
    &\qquad \\
     \tilde{\bX}^{t+1}_{\overrightarrow{e_L}} &= \tilde{\mathbf{Q}}_{\overrightarrow{e}_{L}} \tilde{\bM}^t_{\overrightarrow{e_L}} - b^t_{\overrightarrow{e_L}}\tilde{\bM}^{t-1}_{\overleftarrow{e_L}} \, ,  \\
    &\tilde{\bM}^t_{\overrightarrow{e_L}} = \tilde{F}^t_{\overrightarrow{e}_L}\left(\tilde{\bX}^t_{\overrightarrow{e}_{L-1}},\tilde{\bX}^t_{\overleftarrow{e_L}}\right) \, , \\
    \tilde{\bX}^{t+1}_{\overleftarrow{e_L}} &= \tilde{\mathbf{Q}}_{\overrightarrow{e}_{L}}^{\top}\tilde{\bM}^t_{\overleftarrow{e_L}} - b^t_{\overleftarrow{e_L}}\tilde{\bM}^{t-1}_{\overrightarrow{e_L}}\, ,  \\
    &\tilde{\bM}^t_{\overleftarrow{e_L}} = \tilde{F}^t_{\overleftarrow{e_L}}\left(\tilde{\mathbf{Q}}_{\overrightarrow{e}_{L}}\mathbf{W}_{\overrightarrow{e}_{L}},\tilde{\bX}^t_{\overrightarrow{e_L}}\right) \,
    \end{split}
\end{align}
where each $\mathbf{W}_{\overrightarrow{e}_{l}}$ contains $k_{\overrightarrow{e}_{l}}$ copies of the initial $\mathbf{w}_{\overrightarrow{e}_{l}}$ reorganised into matrices as described above. The dimensions of the variables are  Note that
at this point we have almost reached an iteration verifying the structure of that appearing in Theorem \ref{thm:graph-AMP}, except the Onsager term isn't, a priori, the correct one. Consider the following iteration, where we replaced the original, scalar Onsager terms with the correct, matrix-valued ones:
\begin{align}
\label{eq:mlamp3}
\begin{split}
    \tilde{\bX}^{t+1}_{\overrightarrow{e_1}} &= \tilde{\mathbf{Q}}_{\overrightarrow{e}_{1}} \tilde{\bM}^t_{\overrightarrow{e_1}} -  \tilde{\bM}^{t-1}_{\overleftarrow{e_1}}\left(\tilde{\mathbf{b}}^t_{\overrightarrow{e_1}}\right)^{\top} \, ,  \\
    &\tilde{\bM}^t_{\overrightarrow{e_1}} = \tilde{F}^t_{\overrightarrow{e}_1}(\tilde{\bX}^t_{\overleftarrow{e_1}})\, , \\
        \tilde{\bX}^{t+1}_{\overleftarrow{e_1}} &= \tilde{\mathbf{Q}}_{\overrightarrow{e}_{1}}^{\top}\tilde{\bM}^t_{\overleftarrow{e_1}} - \tilde{\bM}^{t-1}_{\overrightarrow{e_1}}\left(\tilde{\mathbf{b}}^t_{\overleftarrow{e_1}}\right)^{\top} \, ,  \\
    &\tilde{\bM}^t_{\overleftarrow{e_1}} = \tilde{F}^t_{\overleftarrow{e_1}}\left(\tilde{\mathbf{Q}}_{\overrightarrow{e}_{1}}\mathbf{W}_{\overrightarrow{e}_{1}},\tilde{\bX}^t_{\overrightarrow{e_1}},\tilde{\bX}^t_{\overleftarrow{e_2}}\right) \, , \\
    &\qquad \\
        \tilde{\bX}^{t+1}_{\overrightarrow{e_2}} &= \tilde{\mathbf{Q}}_{\overrightarrow{e}_{2}}\tilde{\bM}^t_{\overrightarrow{e_2}} - \tilde{\bM}^{t-1}_{\overleftarrow{e_2}}\left(\tilde{\mathbf{b}}^t_{\overrightarrow{e_2}}\right)^{\top} \, ,  \\
    &\tilde{\bM}^t_{\overrightarrow{e_2}} = \tilde{F}^t_{\overrightarrow{e}_2}\left(\tilde{\bX}^t_{\overrightarrow{e_1}},\tilde{\bX}^t_{\overleftarrow{e_2}}\right) \, , \\
        \tilde{\bX}^{t+1}_{\overleftarrow{e_2}} &= \tilde{\mathbf{Q}}_{\overrightarrow{e}}^{\top}\tilde{\bM}^t_{\overleftarrow{e_2}} - \tilde{\bM}^{t-1}_{\overrightarrow{e_2}}\left(\tilde{\mathbf{b}}^t_{\overleftarrow{e_2}}\right)^{\top} \, ,  \\
    &\tilde{\bM}^t_{\overleftarrow{e_2}} = \tilde{F}^t_{\overleftarrow{e_2}}\left(\tilde{\mathbf{Q}}_{\overrightarrow{e}_{2}}\mathbf{W}_{\overrightarrow{e}_{2}},\tilde{\bX}^t_{\overrightarrow{e_2}},\tilde{\bX}^t_{\overleftarrow{e_3}}\right) 
    \end{split}
\end{align}
\begin{align}
\begin{split}
    &\qquad\\
    &\qquad\quad\vdots \\
    &\qquad \\
     \tilde{\bX}^{t+1}_{\overrightarrow{e_L}} &= \tilde{\mathbf{Q}}_{\overrightarrow{e}_{L}} \tilde{\bM}^t_{\overrightarrow{e_L}} - \tilde{\bM}^{t-1}_{\overleftarrow{e_L}}\left(\tilde{\mathbf{b}}^t_{\overrightarrow{e_L}}\right)^{\top} \, ,  \\
    &\tilde{\bM}^t_{\overrightarrow{e_L}} = \tilde{F}^t_{\overrightarrow{e}_L}\left(\tilde{\bX}^t_{\overrightarrow{e}_{L-1}},\tilde{\bX}^t_{\overleftarrow{e_L}}\right) \, , \\
    \tilde{\bX}^{t+1}_{\overleftarrow{e_L}} &= \tilde{\mathbf{Q}}_{\overrightarrow{e}_{L}}^{\top}\tilde{\bM}^t_{\overleftarrow{e_L}} - \tilde{\bM}^{t-1}_{\overrightarrow{e_L}}\left(\tilde{\mathbf{b}}^t_{\overleftarrow{e_L}}\right)^{\top}\, ,  \\
    &\tilde{\bM}^t_{\overleftarrow{e_L}} = \tilde{F}^t_{\overleftarrow{e_L}}\left(\tilde{\mathbf{Q}}_{\overrightarrow{e}_{L}}\mathbf{W}_{\overrightarrow{e}_{L}},\tilde{\bX}^t_{\overrightarrow{e_L}}\right) \,
    \end{split}
\end{align}
where, for any $\overrightarrow{e} \in \overrightarrow{E}$ and any $t \in \mathbb{N}$
for the right oriented edges 
    \begin{equation*} 
    \mathbf{b}^t_{\overrightarrow{e}_{l}} = \frac{1}{N} \sum_{i=1}^{n_{l-1}} \frac{\partial \tilde{F}^t_{\overrightarrow{e}_{l},i}}{\partial \mathbf{X}_{\overleftarrow{e}_{l},i}} \left(\left(\mathbf{X}^t_{\overrightarrow{e}_{l}'}\right)_{\overrightarrow{e}_{l}':\overrightarrow{e}_{l}' \to \overrightarrow{e}_{l}}\right) \qquad \in \mathbb{R}^{q_{\overrightarrow{e}_{l}} \times q_{\overrightarrow{e}_{l}}} \, .
\end{equation*}
and left oriented edges
    \begin{equation*} 
    \mathbf{b}^t_{\overleftarrow{e}_{l}} = \frac{1}{N} \sum_{i=1}^{n_{l}} \frac{\partial \tilde{F}^t_{\overleftarrow{e}_{l},i}}{\partial \mathbf{X}_{\overrightarrow{e}_{l},i}} \left(\tilde{\mathbf{Q}}_{\overrightarrow{e}_{l}}\mathbf{W}_{\overrightarrow{e}_{l}},\left(\mathbf{X}^t_{\overleftarrow{e}_{l}'}\right)_{\overleftarrow{e}_{l}':\overleftarrow{e}_{l}' \to \overleftarrow{e}_{l}}\right) \qquad \in \mathbb{R}^{q_{\overleftarrow{e}_{l}} \times q_{\overleftarrow{e}_{l}}} \, .
\end{equation*}
Using the separability assumption, we can simplify this expression.
To take a concrete example, consider $\tilde{F}^t_{\overrightarrow{e}_2}\left(\tilde{\bX}^t_{\overrightarrow{e_1}},\tilde{\bX}^t_{\overleftarrow{e_2}}\right)$. Let's start with the dimensions. Recall
\begin{align}
    &\tilde{f}^{t}_{\overrightarrow{e}_{2}}\left(\tilde{\mathbf{x}}^{t}_{\overrightarrow{e}_{1}},\tilde{\mathbf{x}}^{t}_{\overleftarrow{e}_{2}}\right) \in \mathbb{R}^{P_{\overrightarrow{e}_{2}}q_{\overrightarrow{e}_{2}}} = \mathbf{V}_{\overrightarrow{e}_{2}}^{\top}f^t_{\overrightarrow{e}_2}\left(\mathbf{U}_{\overrightarrow{e}_{1}}\tilde{\bx}^t_{\overrightarrow{e_1}},\mathbf{V}_{\overrightarrow{e}_{2}}\tilde{\bx}^t_{\overleftarrow{e_2}}\right) \\
    &\mbox{where} \quad \tilde{\mathbf{x}}^{t}_{\overrightarrow{e}_{1}} \in \mathbb{R}^{D_{\overrightarrow{e}_{1}}q_{\overrightarrow{e}_{1}}}=\mathbb{R}^{P_{\overrightarrow{e}_{2}}q_{\overrightarrow{e}_{2}}} \thickspace \mbox{and} \thickspace \tilde{\mathbf{x}}^{t}_{\overleftarrow{e}_{2}} \in \mathbb{R}^{P_{\overrightarrow{e}_{2}}q_{\overrightarrow{e}_{2}}} 
\end{align}
using the separability assumption, we may write 
\begin{align}
    &\forall \thickspace 1\leqslant i \leqslant P_{\overrightarrow{e}_{2}}q_{\overrightarrow{e}_{2}} \\
    &\left(f^t_{\overrightarrow{e}_2}\left(\mathbf{U}_{\overrightarrow{e}_{1}}\tilde{\bx}^t_{\overrightarrow{e_1}},\mathbf{V}_{\overrightarrow{e}_{2}}\tilde{\bx}^t_{\overleftarrow{e_2}}\right)\right)_{i} = \sigma^{t}_{\overrightarrow{e}_{2}}\left(\left(\mathbf{U}_{\overrightarrow{e}_{1}}\tilde{\bx}^t_{\overrightarrow{e_1}}\right)_{i},\left(\mathbf{V}_{\overrightarrow{e}_{2}}\tilde{\bx}^t_{\overleftarrow{e_2}}\right)_{i}\right)
\end{align}
And 
\begin{align}
    &\tilde{F}^{t}_{\overrightarrow{e}_{2}}\left(\tilde{\mathbf{X}}^{t}_{\overrightarrow{e}_{1}},\tilde{\mathbf{X}}^{t}_{\overleftarrow{e}_{2}}\right) \in \mathbb{R}^{P_{\overrightarrow{e}_{2}}q_{\overrightarrow{e}_{2}}\times q_{\overrightarrow{e}_{2}}}\\
    &\mbox{where} \quad \tilde{\mathbf{X}}^{t}_{\overrightarrow{e}_{1}}\mathbb{R}^{P_{\overrightarrow{e}_{2}}q_{\overrightarrow{e}_{2}}\times q_{\overrightarrow{e}_{2}}} \thickspace \mbox{and} \thickspace \tilde{\mathbf{X}}^{t}_{\overleftarrow{e}_{2}} \in \mathbb{R}^{P_{\overrightarrow{e}_{2}}q_{\overrightarrow{e}_{2}}\times q_{\overrightarrow{e}_{2}}} \\
    &\tilde{F}^{t}_{\overrightarrow{e}_{2}}\left(\tilde{\mathbf{X}}^{t}_{\overrightarrow{e}_{1}},\tilde{\mathbf{X}}^{t}_{\overleftarrow{e}_{2}}\right) = \begin{bmatrix}\left[ (\tilde{f}_{\overrightarrow{e}_{l}}^{t})^{(1)}(\tilde{\mathbf{x}}^{t,(1)}_{\overrightarrow{e}_{1}},\tilde{\mathbf{x}}^{t,(1)}_{\overleftarrow{e}_{2}}) \ldots (\tilde{f}_{\overrightarrow{e}_{l}}^{t})^{(q_{\overrightarrow{e}_{l}})}(\tilde{\mathbf{x}}^{t,(q_{\overrightarrow{e}_{l}})}_{\overrightarrow{e}_{1}},\tilde{\mathbf{x}}^{t,(q_{\overrightarrow{e}_{l}})}_{\overleftarrow{e}_{2}}) \right]_{i=1}^{k_{\overrightarrow{e}_{2}}} \\
    0_{P_{\overrightarrow{e}_{2}}(q_{\overrightarrow{e}_{2}}-k_{\overrightarrow{e}_{2}})\times q_{\overrightarrow{e}_{2}}}\end{bmatrix} \\
    &=\left[ (\tilde{g}_{\overrightarrow{e}_{l}}^{t})^{(1)}(\tilde{\mathbf{x}}^{t,(1)}_{\overrightarrow{e}_{1}},\tilde{\mathbf{x}}^{t,(1)}_{\overleftarrow{e}_{2}}) \ldots (\tilde{g}_{\overrightarrow{e}_{l}}^{t})^{(q_{\overrightarrow{e}_{l}})}(\tilde{\mathbf{x}}^{t,(q_{\overrightarrow{e}_{l}})}_{\overrightarrow{e}_{1}},\tilde{\mathbf{x}}^{t,(q_{\overrightarrow{e}_{l}})}_{\overleftarrow{e}_{2}}) \right]_{i=1}^{q_{\overrightarrow{e}_{2}}}
\end{align}
where each $\tilde{\mathbf{x}}^{t,(i)}_{\overleftarrow{e}_{2}} \in \mathbb{R}^{P_{\overrightarrow{e}_{2}}q_{\overrightarrow{e}_{2}}}$.
Recall that, for any $1\leqslant i \leqslant Pk$, $\tilde{F}^{t}_{\overrightarrow{e}_{2},i}:\mathbb{R}^{q_{\overrightarrow{e}_{2}}} \to \mathbb{R}^{q_{\overrightarrow{e}_{2}}}$.
Then, for any $1 \leqslant k,l \leqslant q_{\overrightarrow{e}_{2}}$
\begin{align}
    \left(\tilde{\mathbf{b}}^{t}_{\overrightarrow{e}_{2}}\right)_{k,l} &= \frac{1}{N}\sum_{i=1}^{P_{\overrightarrow{e}_{2}}q_{\overrightarrow{e}_{2}}}\frac{\partial \tilde{F}^{t}_{\overrightarrow{e}_{2},i,k}}{\partial \mathbf{X}_{\overleftarrow{e}_{2},i,l}}\left(\tilde{\mathbf{X}}^{t}_{\overrightarrow{e}_{1}},\tilde{\mathbf{X}}^{t}_{\overleftarrow{e}_{2}}\right) \\
    &=\frac{1}{N}\sum_{i=1}^{P_{\overrightarrow{e}_{2}}q_{\overrightarrow{e}_{2}}}\frac{\partial (\tilde{g}^{t}_{\overrightarrow{e}_{2},i})^{(k)}}{\partial \tilde{\mathbf{x}}^{(l)}_{\overleftarrow{e}_{2},i}}(\tilde{\mathbf{x}}^{t,(k)}_{\overrightarrow{e}_{1}},\tilde{\mathbf{x}}^{t,(k)}_{\overleftarrow{e}_{2}}) \\
    &=\frac{1}{N}\sum_{i=1}^{P_{\overrightarrow{e}_{2}}q_{\overrightarrow{e}_{2}}}\frac{\partial}{\partial \tilde{\mathbf{x}}^{t,(l)}_{\overleftarrow{e}_{2}}}\mathbf{V}_{\overrightarrow{e}_{2}}^{\top}(g^t_{\overrightarrow{e}_2})^{(k)}\left(\mathbf{U}_{\overrightarrow{e}_{1}}\tilde{\bx}^{t,(l)}_{\overrightarrow{e_1}},\mathbf{V}_{\overrightarrow{e}_{2}}\tilde{\bx}^{t,(l)}_{\overleftarrow{e_2}}\right) \\
    &= \frac{1}{N}\mbox{Tr}\left(\mathbf{V}^{\top}_{\overrightarrow{e}_{2}}\mathcal{J}_{(g^t_{\overrightarrow{e}_2})^{(k)}}\left(\mathbf{U}_{\overrightarrow{e}_{1}}\tilde{\bx}^{t,(l)}_{\overrightarrow{e_1}},\mathbf{V}_{\overrightarrow{e}_{2}}\tilde{\bx}^{t,(l)}_{\overleftarrow{e_2}}\right)\mathbf{V}_{\overrightarrow{e}_{2}}\right)\delta_{k,l} \\
    & =\frac{1}{N}\mbox{Tr}\left(\mathcal{J}_{(g^t_{\overrightarrow{e}_2})}\left(\mathbf{U}_{\overrightarrow{e}_{1}}\tilde{\bx}^{t}_{\overrightarrow{e_1}},\mathbf{V}_{\overrightarrow{e}_{2}}\tilde{\bx}^{t}_{\overleftarrow{e_2}}\right)\right)\delta_{k,l} \\
    &= \frac{1}{N}\sum_{i=1}^{P_{\overrightarrow{e}_{2}}q_{\overrightarrow{e}_{2}}}(\sigma^{t})'_{\overrightarrow{e}_{2}}\left(\left(\mathbf{U}_{\overrightarrow{e}_{1}}\tilde{\bx}^t_{\overrightarrow{e_1}}\right)_{i},\left(\mathbf{V}_{\overrightarrow{e}_{2}}\tilde{\bx}^t_{\overleftarrow{e_2}}\right)_{i}\right)\delta_{k,l}
\end{align}
where we wrote $\mathcal{J}_{(g^{t}_{\overrightarrow{e}_{2}})^{(k)}}$ the $N \times N$ Jacobian matrix of the function $(g^{t}_{\overrightarrow{e}_{2}})^{(k)} : \mathbb{R}^{N} \to \mathbb{R}^{N}$.
Using \cite{berthier2020state} corollary 2, the Onsager term can be replaced by any estimator based on the asymptotically Gaussian iterates converging, in the high-dimensional limit, to the correct expectation. The adaptation to the graph framework of \cite{gerbelot2021graph} is immediate (see the proof in \cite{berthier2020state} and corresponding comment in \cite{gerbelot2021graph}). Using the permutation invariance of the Gaussian distribution, we can therefore replace each element of the matrix the Onsager term with 
\begin{equation}
    \frac{1}{P_{\overrightarrow{e}_{2}}q_{\overrightarrow{e}_{2}}}\sum_{i=1}^{P_{\overrightarrow{e}_{2}}q_{\overrightarrow{e}_{2}}}(\sigma^{t})'_{\overrightarrow{e}_{2}}\left(\left(\tilde{\bx}^t_{\overrightarrow{e_1}}\right)_{i},\left(\tilde{\bx}^t_{\overleftarrow{e_2}}\right)_{i}\right)\delta_{k,l}
\end{equation}
which amounts to 
\begin{equation}
    \tilde{\mathbf{b}}^{t}_{\overrightarrow{e}_{2}} = b^{t}_{\overrightarrow{e}_{2}}\mathbf{I}_{q_{\overrightarrow{e}_{2}}\times q_{\overrightarrow{e}_{2}}}
\end{equation}
We therefore obtain an exact reformulation of the initial MLAMP iteration with convolutional matrices in terms of a subset (first line of size $P_{\overrightarrow{e}_{l}}\times q_{\overrightarrow{e}_{l}}$ for right oriented edges and $D_{\overrightarrow{e}_{l}}\times q_{\overrightarrow{e}_{l}}$ for left-oriented variables) of the variables of a matrix-valued iteration with dense Gaussian matrices verifying the SE equations. Isolating the aforementioned first lines, recalling that the SE equations prescribes i.i.d. lines in the asymptotically Gaussian fields, we recover that, for any $1 \leqslant l \leqslant L$, the variable $\mathbf{x}_{\overrightarrow{e}_{l}} \in \mathbb{R}^{P_{\overrightarrow{e}_{l}}q_{\overrightarrow{e}_{l}}}$ is composed of $q_{\overrightarrow{e}_{l}}$ copies of block of size $P_{\overrightarrow{e}_{l}}$ with i.i.d. Gaussian elements distributed according to the SE equations \eqref{eq:scalar_se}. The distribution of the variables associated to left-oriented edges is obtained similarly. Note that, from a finite size point of view, the effect of $D_{\overrightarrow{e}_{l}}, P_{\overrightarrow{e}_{l}}$ is different from that of $q_{\overrightarrow{e}_{l}}$ : the former results in subGaussian concentration i.e. exponential in the dimension, while the latter only represents copies (and not i.i.d. samples), and thus only has an averaging effect. This is observed in simulations.

\end{proof}
\subsection{Bayes-optimal MLAMP with random convolutional matrices}
In this section, we specialize the equations obtained in the previous section to the Bayes-optimal MLAMP iteration of the main body of the paper. Several functions are reminded for convenience.
Consider the MLAMP iteration outlined in section \ref{sec:MLAMP}. The scalar updates described in Eq.\eqref{eqn:ml-amp-iterate} can be rewritten as vector-valued updates as follows, for any $t \in \mathbb{N}$, and any $0\leqslant l \leqslant L$:
\begin{align}
    \boldsymbol{\omega}^{(l)}(t) &= \mathbf{W}^{(l)}\hat{\mathbf{h}}^{(l)}(t)-V^{(l)}(t)\mathbf{g}^{(l)}(t-1) \\
    \mathbf{B}^{(l)}(t) &= \left(\mathbf{W}^{(l)}\right)^{\top}\mathbf{g}^{(l)}(t)-\hat{V}^{(l)}(t)\hat{\mathbf{h}}(t).
\end{align}

To define the update functions and terms $V^{(l)},\hat{V}^{(l)}$, the following partition functions were introduced.
\begin{itemize}
\item for $l=1$ \begin{align}
    \mathcal{Z}^{(1)}\left(y,V^{(1)},\omega^{(1)}\right) &= \frac{1}{\sqrt{2\pi V^{(1)}}}\int dzP_{out}^{(1)}(y\vert z)e^{-\frac{(z-\omega^{(1)})^{2}}{2V^{(1)}}}
\end{align}
\item for any  $2\leqslant l \leqslant L-1$ : 
\begin{align}
     \hspace{-3cm}\mathcal{Z}^{(l)}\left(A^{(l-1)},B^{(l-1)},V^{(l)},\omega^{(l)}\right) =  \notag \\
     \frac{1}{\sqrt{2\pi V^{(l)}}}\int dhdzP_{out}^{(l)}(h\vert z)e^{-\frac{1}{2}A^{(l-1)}h^{2}+B^{(l-1)}h}e^{-\frac{(z-\omega^{(l)})^{2}}{2V^{(l)}}} 
\end{align}
\item for $l=L $\begin{align}
     \mathcal{Z}^{(L)}(A^{(L)},B^{(L)}) &= \int dhP_{X}(h)e^{-\frac{1}{2}A^{(L)}h^{2}+B^{(L)}h}
\end{align}
\end{itemize}
We then define the layer-dependent, time-dependent, scalar update functions $f^{(l),t},\tilde{f}^{(l),t}$
\begin{align}
&\forall \thickspace (B,\omega) \in \mathbb{R}^{2} \notag \\
f^{(1),t}(\omega) &= \partial_{\omega}\mbox{log}\mathcal{Z}^{(1)}\left(y,V^{(1)}(t),\omega\right) \\
    f^{(l),t}(B,\omega) &= \partial_{\omega}\mbox{log}\mathcal{Z}^{(l)}\left(A^{(l-1)}(t),B,V^{(l)}(t),\omega\right) \thickspace 2\leqslant l\leqslant L \\
    \tilde{f}^{(l),t}(B,\omega) &= \partial_{B}\mbox{log}\mathcal{Z}^{(l+1)}\left(A^{(l)}(t-1),B,V^{(l+1)}(t-1),\omega\right) \thickspace 1\leqslant l\leqslant L-1 \\
    \tilde{f}^{(L,t)}(B) &= \partial_{ B}\mbox{log}\mathcal{Z}^{(L+1)}\left(A^{(L)}(t-1),B\right),
\end{align}
and their corresponding separable, vector valued counterparts $\mathbf{f}^{(l)},\tilde{\mathbf{f}}^{(l)}$, which leads to the following iteration
\begin{align}
\label{eq:MLAMP_vec}
    \boldsymbol{\omega}^{(l)}(t) &= \mathbf{W}^{(l)}\tilde{\mathbf{f}}^{(l),t}(\mathbf{B}^{(l),t-1},\boldsymbol{\omega}^{(l+1),t-1})-V^{(l)}(t)\mathbf{f}^{(l),t-1}(\mathbf{B}^{(l-1),t-1},\boldsymbol{\omega}^{(l),t-1}) \\
    \mathbf{B}^{(l)}(t) &= \left(\mathbf{W}^{(l)}\right)^{\top}\mathbf{f}^{(l),t}(\mathbf{B}^{(l-1),t},\boldsymbol{\omega}^{(l),t})-\hat{V}^{(l)}(t)\tilde{\mathbf{f}}^{(l),t}(\mathbf{B}^{(l),t-1},\boldsymbol{\omega}^{(l+1),t-1}),
\end{align}
where the Onsager terms $V^{(l),t}$ and $\hat{V}^{(l),t}$ reduce to, using the separability of the update functions,
\begin{align}
    V^{(l),t} = \frac{1}{n_{l}}\sum_{i=1}^{n_{l-1}} \partial_{B}\tilde{f}^{(l),t}(B^{(l),t-1}_{i},\omega^{(l+1),t-1}_{i}) \\
    \hat{V}^{(l),t} = \frac{1}{n_{l}}\sum_{j=1}^{n_{l}}\partial_{\omega}f^{(l),t}(B_{j}^{(l-1),t},\omega_{j}^{(l),t}) = -A^{(l),t}
\end{align}

We now show that the update functions defined above are Lipschitz continuous and increasing, thus ensuring that the integrals are well defined through positivity of the parameters $V,\hat{V}$.
\begin{lemma}
For any $1\leqslant l \leqslant L$, and any $t \in \mathbb{N}$, the functions $f^{(l),t}, \tilde{f}^{(l),t}$ are Lipschitz continuous in $B,\omega$. Furthermore, the functions $f^{(l),t}, \tilde{f}^{(l),t}$ are respectively decreasing in $\omega$ and increasing in $B$. As a consequence, the variance terms $A^{(l),t}$ and $V^{(l),t}$ are strictly positive.
\begin{proof}
Recall the partition function, omitting the layer index since all regularity assumptions are the same for all layers and time indices, 
\begin{align}
 \mathcal{Z}(A, B, V, \omega) & := \frac{1}{\sqrt{2 \pi V}}\int P(h \mid z) \exp\left( B h - \frac{1}{2} A h^2 - \frac{(z-\omega)^2}{2 V } \right) \, dh \, dz
\end{align}
recalling $p(h\vert z) = \int p(\xi)\delta(h-f_{\xi}(z))d\xi$, integrating in $h$ yields
\begin{align}
    \mathcal{Z}(A, B, V, \omega) & := \frac{1}{\sqrt{2 \pi V}}\int P(\xi)\exp\left( B f_{\xi}(z) - \frac{1}{2} A f_{\xi}(z)^2 - \frac{(z-\omega)^2}{2 V } \right) \, d\xi \, dz
\end{align}
Starting with $\tilde{f}$, we can straightforwardly verify the conditions to apply the dominated convergence theorem and differentiate under the integral to obtain
\begin{align}
    &\partial_{B}\tilde{f}(B,\omega) = \partial^{2}_{B} \log \left(\mathcal{Z}(A, B, V, \omega)\right) \notag\\
     &= \frac{1}{(\sqrt{2\pi V}\mathcal{Z}(A, B, V, \omega))^{2}}\bigg(\int P(\xi)f_{\xi}^{2}(z)\exp\left( B f_{\xi}(z) - \frac{1}{2} A f_{\xi}(z)^2 - \frac{(z-\omega)^2}{2 V } \right) \, d\xi \, dz \times \notag\\
     &\int P(\xi)\exp\left( B f_{\xi}(z) - \frac{1}{2} A f_{\xi}(z)^2 - \frac{(z-\omega)^2}{2 V } \right) \, d\xi \, dz - \notag\\
     &\hspace{3cm}\left(\int P(\xi)f_{\xi}(z)\exp\left( B f_{\xi}(z) - \frac{1}{2} A f_{\xi}(z)^2 - \frac{(z-\omega)^2}{2 V } \right) \, d\xi \, dz\right)^{2}\bigg) \geqslant 0
\end{align}
where the positivity comes from the Cauchy-Schwarz inequality and positivity of the term $P(\xi)\exp\left( B f_{\xi}(z) - \frac{1}{2} A f_{\xi}(z)^2 - \frac{(z-\omega)^2}{2 V } \right)$.
Turning to $f$, we complete the square in the variable $h$ to obtain
\begin{align}
    \mathcal{Z}(A, B, V, \omega) & := \frac{\exp\left(\frac{B^{2}}{2A}\right)}{\sqrt{2 \pi V}}\int P(\xi)\exp\left(-\frac{A}{2}\left(f_{\xi}(z)-\frac{B}{A}\right)^{2}\right) \exp\left(- \frac{(z-\omega)^2}{2 V } \right) \, d\xi \, dz
\end{align}
and differentiating under the integral yields
\begin{align}
    &f(B,\omega) = \partial_{\omega} \log \left(\mathcal{Z}(A, B, V, \omega)\right) \\
    &=\frac{1}{V}\left(\frac{\int P(\xi)z\exp\left(-\frac{A}{2}\left(f_{\xi}(z)-\frac{B}{A}\right)^{2}\right) \exp\left(- \frac{(z-\omega)^2}{2 V } \right) \, d\xi \, dz}{\left(\int P(\xi)\exp\left(-\frac{A}{2}\left(f_{\xi}(z)-\frac{B}{A}\right)^{2}\right) \exp\left(- \frac{(z-\omega)^2}{2 V } \right) \, d\xi \, dz \right)}-\omega\right)
\end{align}
where the term $\frac{\int P(\xi)z\exp\left(-\frac{A}{2}\left(f_{\xi}(z)-\frac{B}{A}\right)^{2}\right) \exp\left(- \frac{(z-\omega)^2}{2 V } \right) \, d\xi \, dz}{\left(\int P(\xi)\exp\left(-\frac{A}{2}\left(f_{\xi}(z)-\frac{B}{A}\right)^{2}\right) \exp\left(- \frac{(z-\omega)^2}{2 V } \right) \, d\xi \, dz \right)}$ is the conditional mean of the distribution with density $\frac{ \int P(\xi)\exp\left(-\frac{A}{2}\left(f_{\xi}(z)-\frac{B}{A}\right)^{2}\right) \exp\left(- \frac{(z-\omega)^2}{2 V } \right)d \xi}{\left(\int P(\xi)\exp\left(-\frac{A}{2}\left(f_{\xi}(z)-\frac{B}{A}\right)^{2}\right) \exp\left(- \frac{(z-\omega)^2}{2 V } \right) \, d\xi \, dz \right)}$.
The Lipschitz property is straightforward to verify using the polynomial bound assumption on the activation functions and the inverse exponential factors.
\end{proof}
In the Bayes-optimal MLAMP, see \cite{manoel2017multi}, the planted vectors $\mathbf{w}_{\overrightarrow{e}_{l}}$ are chosen as independently distributed as the asymptotic SE representation of the output of the previous layer, and are therefore Lipschitz transforms of subGaussian random variables, and thus are also subgaussian. Using the permuation invariance of the Gaussian distribution, the quantities $\mathbf{z}_{\overrightarrow{e}_{l}} = \hat{\mathbf{A}}_{\overrightarrow{e}_{l}}$ remain Gaussian. We can therefore apply the result of Lemma \ref{lemma:conv_SE_scalar} to this iteration and obtain that iterates of Eq.\eqref{eqn:ml-amp-iterate} verify the SE equations from Lemma \ref{lemma:conv_SE_scalar} with the corresponding update functions. Furthermore, in the Bayes optimal case, the Nishimori conditions, see e.g. \cite{krzakala2012statistical}, allow to only keep the parameters $\nu_{\overrightarrow{e}_{l}}, \hat{\nu}_{\overleftarrow{e}_{l}}$ to describe the distribution of of the iterates, recovering the equations of Theorem \ref{th:main}. Finally, the rescaling of the variances to go from the factors $\delta_{l}$ to the $\beta_{l}$ of the main can be done by rescaling each non-linearity $f^{t}_{\overrightarrow{e}_{l}}$ by $\sqrt{N/n_{l-1}}$ (and similary for the $f^{t}_{\overleftarrow{e}_{l}}$ with $\sqrt{N/n_{l}}$) as done in \cite{javanmard2013state,berthier2020state}.
\end{lemma}
\section{Fast MCC-vector Products} \label{sec:fast-mcc-vec}
Here is a simple sketch of an algorithm for multiplying $M \sim \text{MCC}(D, P, q, k)$ with a vector $v \in \mathbb{R}^{Pq}$ that runs in time $O(DP q \log q)$. If $k \gg \log q$, this improves on the runtime required by a simple sparse matrix-vector product. We use Matlab index notation for matrix and vector coordinates, for example $M[i:j, k] = [ M_{rk} \, : \, r = i \ldots j]$, and we write shorthand $M_{ij}$ for $M[i, j]$. 

\begin{algorithm}
\caption{$O(DPq\log q)$ time algorithm for MCC-vector products}\label{alg:mcc-vector-product}
\KwData{matrix $M \sim \text{MCC}(D, P, q, k)$, vector $v \in \mathbb{R}^{Pq}$}
Initialize $s \in \mathbb{R}^{Dq}$ the zero vector\;
\For{$i = 1 \ldots D$}{
    \For{$j = 1 \ldots P$}{
        $C_{ij} \gets M[qi : q(i+1) ,\, qj: q(j+1)]$\;
         
        $\omega_{ij} = C_{ij}[0,\, 0:k]$\;
        
        $\hat{\omega}_{ij} = \text{FFT}(w_{ij})$\; 
        
        $\hat{v}_j = \text{FFT}(v[qj,\,  q(j+1)]$\; 
        
        $\hat{s}_i = \hat{\omega}_{ij} \ast \hat{v}_j$\; 
        
        $s[qi:q(i+1)] = \text{IFFT}(\hat{s}_i)$\;
    }
}
\end{algorithm}

\section{Additional Experiments} \label{sec:addl-expts}

\subsection{Sparse Compressive Sensing} \label{sec:q10-sparse-cs}
We observe empirically that in the sparse compressive sensing task of Figure \ref{fig:sparse-cs}, the relative sizes of $(D, P)$ and $q$ have little impact on the performance of the corresponding AMP iteration. In Figure \ref{fig:q10-sparse-cs}, we show a replica of this figure with $q = 10$ and $P = 10000$. Despite a significant difference between the relative sizes of these parameters, the AMP iterations behave largely the same. 

\begin{figure}[t]
    \centering
    \includegraphics[width=0.48\linewidth]{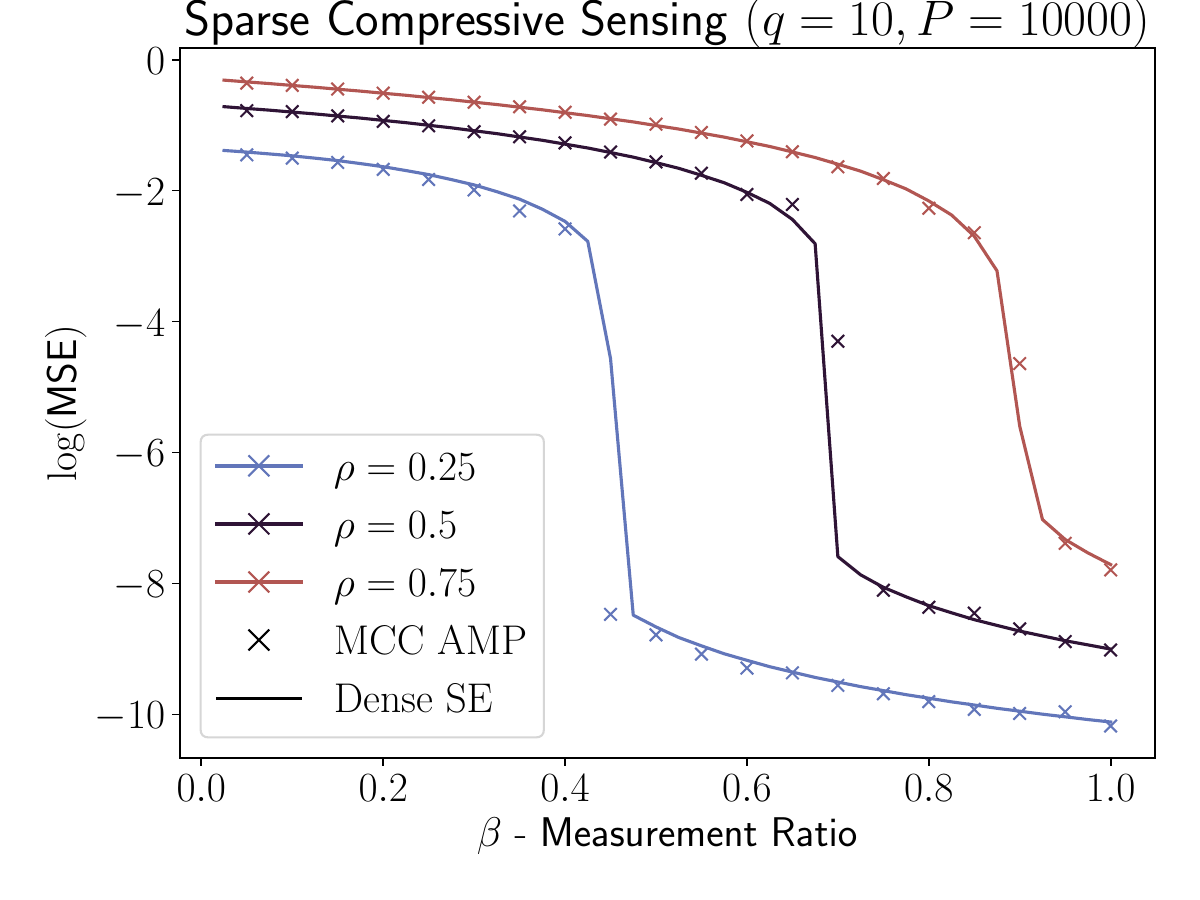}
    \includegraphics[width=0.51\linewidth]{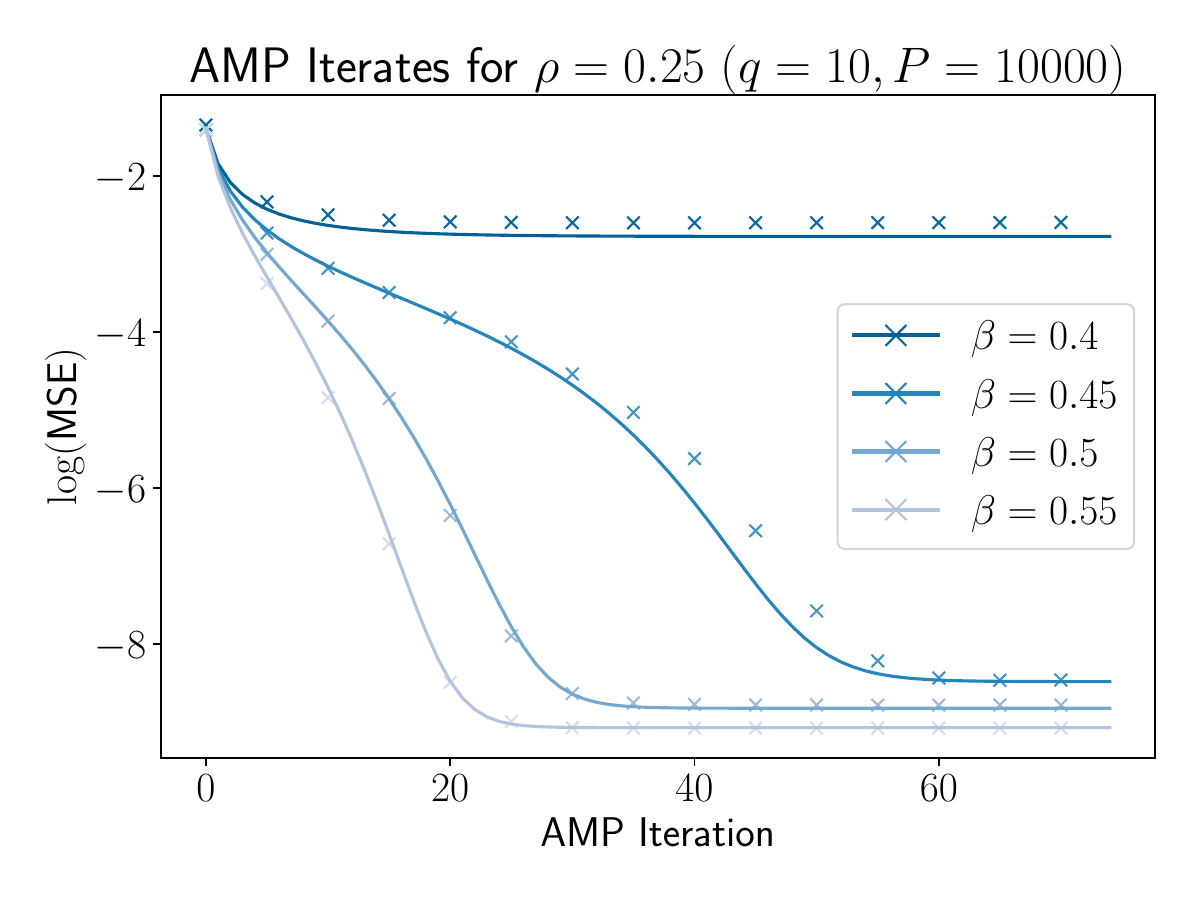}
    \caption{
    Replica of Figure \ref{fig:sparse-cs} for $q=10$ and $P=10000$. 
    (\textbf{left}) Compressive sensing $y_0 = W x_0 + \zeta$ for noise $\zeta_i \sim \mathcal{N}(0, 10^{-4})$ and signal prior $x_0 \sim \rho \mathcal{N}(0, 1) + (1-\rho) \delta(x)$, where $W \in \mathbb{R}^{Dq \times Pq}$ has varying aspect ratio $\beta = D / P$. Crosses correspond to AMP evaluations for $W \sim \text{MCC}(D, P, q, k)$ according to Definition \ref{dfn:mcc}, averaged over 10 independent trials. Lines show the state evolution predictions when $W_{ij} \sim \mathcal{N}(0,1/Pq)$. The system size is $P = 10000$, $q=10$, $k=3$, where $\beta$ and $D = \beta P$ vary. 
    (\textbf{right}) AMP iterates at $\rho = 0.25$ and $\beta$ near the recovery transition. }
    \label{fig:q10-sparse-cs}
\end{figure}

\subsection{Empirical Results for Vector-AMP Algorithms}\label{sec:vamp-results}
We observe that a similar equivalence property as Theorem \ref{th:main} holds for algorithms based on the VAMP framework \cite{schniter2016vector, fletcher2018inference, baker2020tramp}. Previously, state evolution has been proven for such algorithms when their sensing matrices are drawn from a right-orthogonally-invariant ensemble. While the random MCC ensemble does not satisfy this property, we show in Figure \ref{fig:tramp-sparse-cs} a comparison between empirical VAMP performance and the corresponding SE predictions for dense matrices, which are almost identical.

\begin{figure}
 \includegraphics[width=0.48\linewidth]{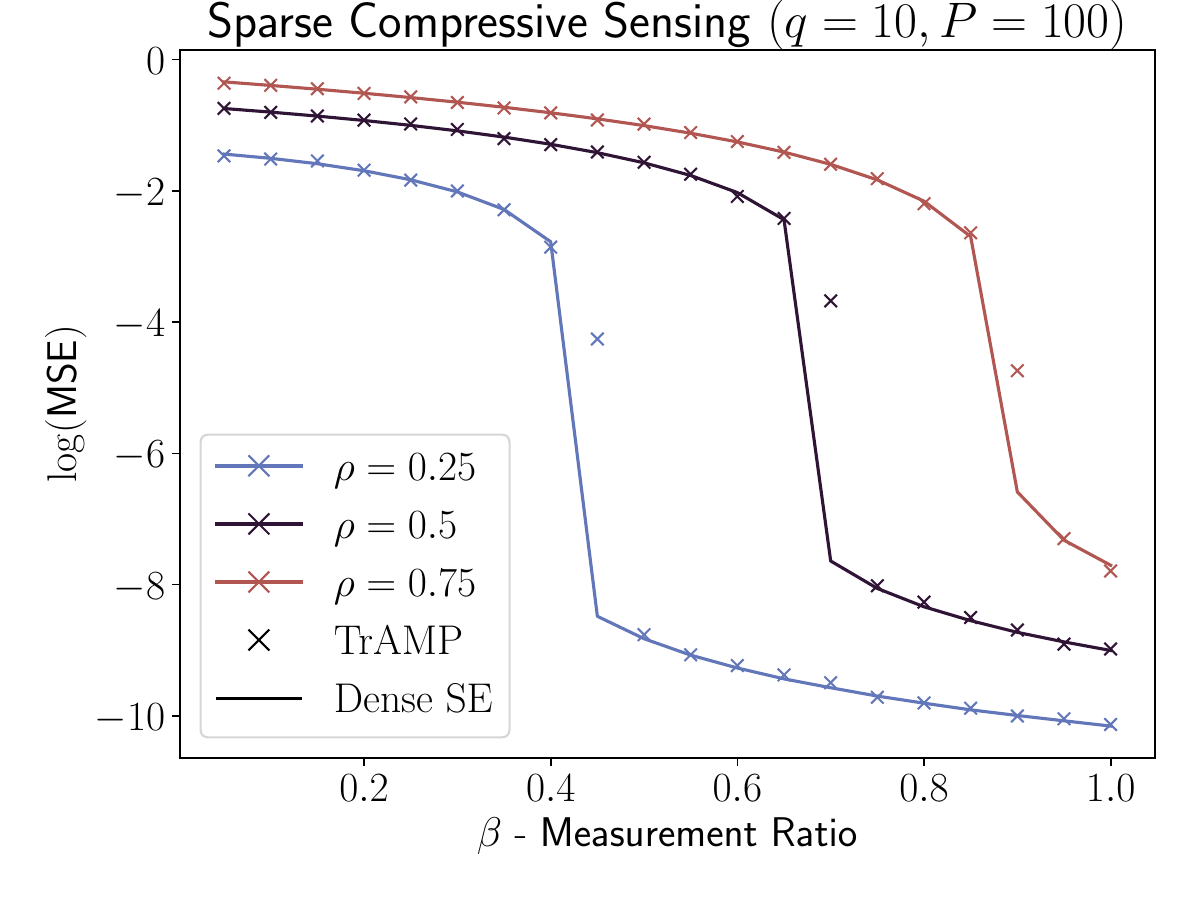}
    \includegraphics[width=0.51\linewidth]{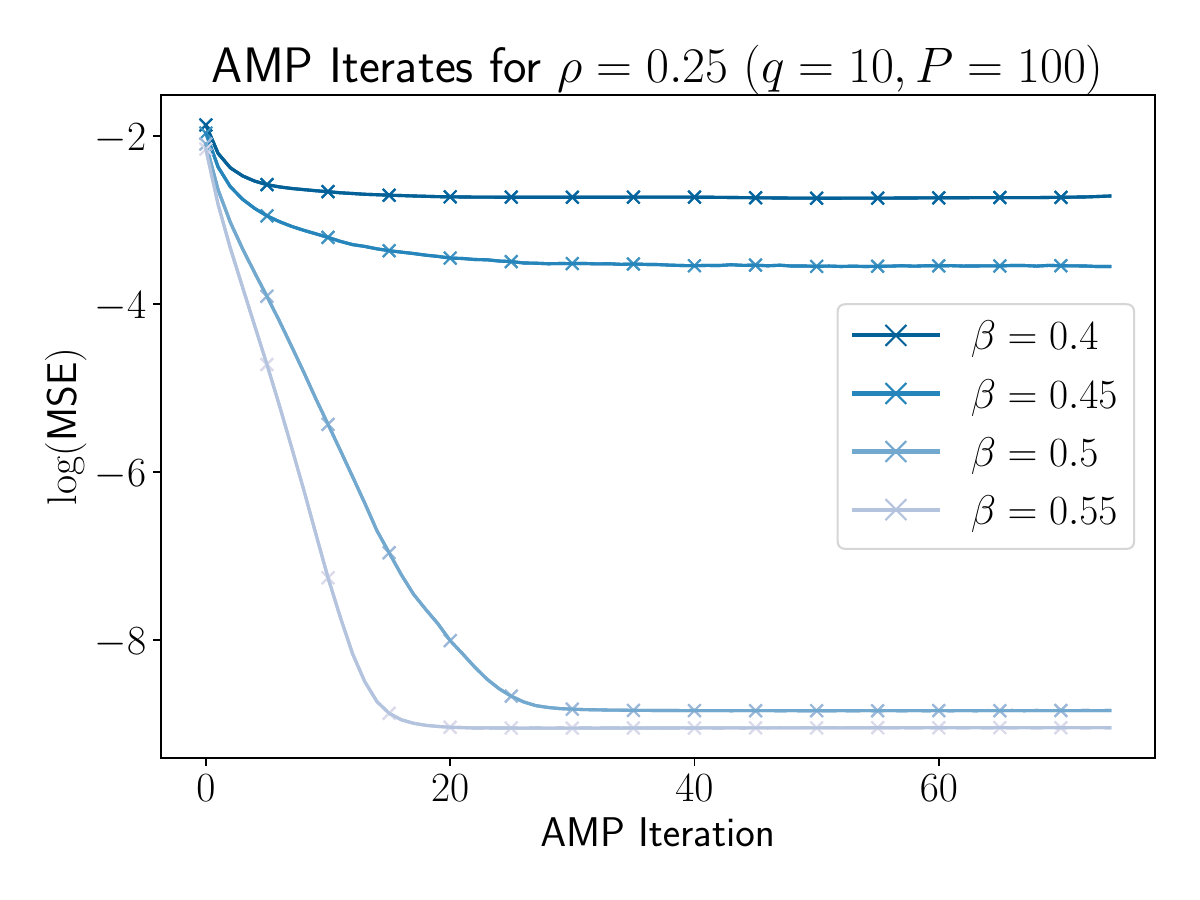}
    \caption{
    Replica of Figure \ref{fig:sparse-cs} using Tree-AMP \cite{baker2020tramp}, a compositional VAMP type algorithm, for $q=10$ and $P=100$. 
    (\textbf{left}) Compressive sensing $y_0 = W x_0 + \zeta$ for noise $\zeta_i \sim \mathcal{N}(0, 10^{-4})$ and signal prior $x_0 \sim \rho \mathcal{N}(0, 1) + (1-\rho) \delta(x)$, where $W \in \mathbb{R}^{Dq \times Pq}$ has varying aspect ratio $\beta = D / P$. Crosses correspond to AMP evaluations for $W \sim \text{MCC}(D, P, q, k)$ according to Definition \ref{dfn:mcc}, averaged over 30 independent trials. Lines show the state evolution predictions when $W_{ij} \sim \mathcal{N}(0,1/Pq)$. The system size is $P = 100$, $q=10$, $k=3$, where $\beta$ and $D = \beta P$ vary. 
    (\textbf{right}) AMP iterates at $\rho = 0.25$ and $\beta$ near the recovery transition. }
    \label{fig:tramp-sparse-cs}
\end{figure}

\section{Structured Convolutions and Non-separable Denoising}\label{app:nonseparable}
Our proof uses a relatively simple version of spatial coupling, leaving avenues for potential generalizations. Spatially coupled sensing matrices typically consist of a block structured matrix whose blocks are i.i.d. Gaussian with different variances, as in (for instance) \cite{krzakala2012statistical,barbier2015approximate}. As a model, consider $\tilde{M}_{\text{sp}}$ of the following form, with variances $\kappa \in \mathbb{R}_+^{q \times q}$,
\begin{align*}
    \tilde{M}_{\text{sp}} = \begin{bmatrix}
    \kappa_{11} A_{11} & \kappa_{12} A_{12} & \hdots & \kappa_{1q} A_{1q}\\
    \kappa_{21} A_{21} & \ddots & & \vdots \\
    \vdots & & & \\
    \kappa_{1q} A_{1q} & \hdots & & \kappa_{qq} A_{qq}
    \end{bmatrix}.
\end{align*}
As a result of Lemma \ref{lem:permutation}, a given MCC matrix $M$ is equivalent to $\tilde{M}$ corresponding to the case where $\kappa$ is a convolutional matrix according to Definition \ref{dfn:c-ensemble}, with filter $\omega = [1\ 1\ \ldots \ 1]\in \mathbb{R}^k$. One avenue to extend our results is to consider general $\tilde{M}$ where $\kappa$ is any circulant matrix. Inverting the permutation lemma, this corresponds to MCC matrices whose convolutional blocks have filters with independent non-isotropic coordinates, as in the following definition, which may be viewed as a simple model for structured convolutional filters. 
\begin{definition}[Independent Gaussian Random Convolutions] \label{dfn:c-diag-ensemble}
Let $\vec{\kappa} = [\kappa_1, \ldots, \kappa_k] \in \mathbb{R}^k_+$ and let $\Sigma = \diag(\vec{\kappa})$. 
Let $q \geq k > 0$ be integers. The Gaussian convolutional ensemble $\mathcal{C}(q, k)$ contains random circulant matrices $C \in \mathbb{R}^{q \times q}$ whose first rows are given by $C_1 = \texttt{Zero-pad}_{q, k}[\omega]$
where $\omega \sim \mathcal{N}(0, \Sigma)$.
\end{definition}
This model is a natural extension of our current setting, which is also amenable to proof techniques designed for spatial coupling. However, because the nonzero coordinates of the sensing matrix are no longer i.i.d., the Bayes-optimal denoising functions corresponding to this problem are non-separable. So, an equivalence theorem analogous to Theorem \ref{th:main} is not expected to hold -- in other words, state evolution in this convolutional model is not expected to reduce to that of a signal model with dense i.i.d. couplings. More generally, multilayer AMP iterations with non-separable non-linearities may be written to compute marginals of posterior distributions involving such functions, and will verify SE equations. However there will be no direct correspondance with the iteration and SE equations of the fully separable case.

\end{document}